\documentclass[runningheads]{llncs}
\usepackage[T1]{fontenc}
\usepackage[utf8]{inputenc}
\usepackage{graphicx}
\usepackage[bibliography=common]{apxproof}

\usepackage{amssymb}
\usepackage{graphicx}
\usepackage{subcaption}
\usepackage{hyperref}
\usepackage{amsmath}
\usepackage{comment}
\usepackage{breqn}
\usepackage{cleveref} 
\crefname{algocf}{algorithm}{algorithms}
\Crefname{algocf}{Algorithm}{Algorithms}
\usepackage[ruled, linesnumbered]{algorithm2e}
\usepackage{booktabs, longtable, multicol, adjustbox}
\usepackage{tabularx}
\usepackage{todonotes}

\usepackage{geometry}
\usepackage{soul, xcolor}	
\newcommand{\dk}[1]{\textcolor{violet}{#1}}
\newcommand{\dkb}[1]{\dk{(DKB says: #1)}}

\newcommand{\er}[1]{\textcolor{blue}{#1}}
\newcommand{\erel}[1]{\er{(Erel says: #1)}}

\newcommand{\kbp}{$k\text{BP }$}

\newtheoremrep{theorem}{Theorem}[section]
\newtheoremrep{lemma}{Lemma}[section]
\newtheoremrep{corollary}[section]{Corollary}

\theoremstyle{definition}
\newtheoremrep{definition}{Definition}[section]
\newtheoremrep{example}{Example}[section]
\newtheoremrep{remark}{Remark}[section]

\begin{document}

    \begin{toappendix}
        \renewcommand{\thesection}{\Alph{section}}
        \renewcommand{\thesubsection}{\Alph{section}.\arabic{subsection}}
        \renewcommand{\thesubsubsection}{\Alph{section}.\arabic{subsection}.\arabic{subsubsection}}
    \end{toappendix}

\title{Time and Supply Fairness in Electricity Distribution using \texorpdfstring{$k$}{k}-times bin packing \thanks{This is an extended version of our paper presented at SAGT-2024 (\url{https://link.springer.com/chapter/10.1007/978-3-031-71033-9_27}). 
This version corrects a bug present in the accepted SAGT paper that affected \Cref{subsec: results-egal-connection-time}.
In addition, it adds the following over the conference version of the paper 
We give a lower bound on $k$ that depends on the number of households (see \Cref{subsec: lower bound on k} );
In \Cref{sec: Egalitarian allocation of supply}, we study another version of the egalitarian allocation, that is, the egalitarian allocation of watts problem. In this direction, we present impossibility results that state that there does not exist a finite $k$ that depends on the number of households. In addition, we develop four heuristic algorithms to solve the egalitarian allocation of watts problem.
}
}

%
%
\author{Dinesh Kumar Baghel\inst{1,2}
	\orcidID{0000-0001-8518-7908} 
	\and
	Alex Ravsky\inst{3}
	\and
	Erel Segal-Halevi\inst{1}
	\orcidID{0000-0002-7497-5834}
}
%
%
\institute{Ariel University, Ariel 40700, Israel \\
	\email{\{dinkubag21, erelsgl\}@gmail.com}\\
	\and
    UPES, Dehradun, Uttarakhand, India
    \and
	Pidstryhach Institute for Applied Problems of Mechanics and Mathematics of National Academy of Sciences of Ukraine, Lviv, Ukraine \\
	\email{alexander.ravsky@uni-wuerzburg.de}
    }
\maketitle             

\begin{abstract}
Given items of different sizes and a fixed bin capacity, the bin-packing problem is to pack these items into a minimum number of bins such that the sum of item
sizes in a bin does not exceed the capacity. 
We define a new variant called \emph{$k$-times bin-packing ($k$BP)}, where the goal is to pack the items such that each item appears exactly $k$ times, in $k$ different bins. 
We generalize some existing approximation algorithms for bin-packing to solve $k$BP, and analyze their performance ratio.
	
The study of $k$BP is motivated by the problem of \emph{fair electricity distribution}. In many developing countries, the total electricity demand is higher than the supply capacity. 
Our goal is to allocate the available supply among the households according to the egalitarian principle, which aims to maximize the smallest amount of time a household is connected to electricity.
    
We prove that every electricity division problem can be solved by $k$-times bin-packing for some finite $k$, that depends only on the number of households.
We also show that $k$-times bin-packing can be used in practice to distribute the electricity in a fair and efficient way. 
Particularly, we implement generalizations of the First-Fit and First-Fit Decreasing bin-packing algorithms to solve $k$BP, and apply the generalizations to real electricity demand data. We show that our generalizations outperform existing heuristic solutions to the same problem in terms of the egalitarian allocation of connection time.  
    
We then study another variant of the egalitarian allocation problem, in which the goal is to maximize the smallest amount of watts allocated to a household. For this variant, we prove an impossibility result: there does not exist such a $k$ that depends only on the number of agents.
This impossibility result motivates us to develop four different heuristic algorithms to solve the egalitarian allocation of watts problem. We evaluate the heuristics using the sum of the minimum watts allocated to any household in each hour, providing a fairness metric that reflects the lowest watt allocation in all hours.
A higher total minimum of watts indicates a more equitable distribution.
Thus we establish new benchmarks for fair allocation of watts.

	
\keywords{Approximation algorithms  \and bin-packing \and First-Fit \and First-Fit Decreasing \and Next-Fit \and fair division \and Karmarkar-Karp algorithms \and Fernandez de la Vega-Lueker algorithm \and electricity distribution \and utilitarian metric \and egalitarian metric \and utility difference.}
\end{abstract}

\section{Introduction} \label{sec:introduction}
This work is motivated by the problem of \emph{fair electricity distribution}. 
In developing countries, the demand for electricity often surpasses the available supply \cite{Kaygusuz2012}.  
Such countries have to come up with a fair and efficient method of allocating of the available electricity among the households.

Formally, we consider a power-station that produces a fixed supply $S$ of electricity. 
The station should provide electricity to $n$ households. The demands of the households in a given period are given by a (multi)set $D$. Typically, $\sum_i D[i] > S$ (where $D[i]$ is the electricity demand of a household $i$), 
so it is not possible to connect all households simultaneously.

\paragraph{Egalitarian allocation of connection time:}
In the basic variant of the problem, the goal is to ensure that each household is connected the same amount of time, and that this amount is as large as possible.
We assume that an agent gains utility only if the requested demand is fulfilled; otherwise it is zero. Practically it can be understood as follows: Suppose at some time, a household $i$ is running some activity that requires $D[i]$ kilowatt of electricity to operate; in the absence of that amount, the activity will not function. 
Therefore, an allocation where demands are fractionally fulfilled is not relevant. 
We also assume that an agent's utility does not increase for receiving electricity more than his/her demand, as large batteries for storing a substantial amount of electricity are still expensive and not very common in developing countries.

While it is not necessary that agents consume electricity equal to their demand. In addition, some agents may consume more electricity than their demand and some agents consume less. It may lead to the cancellation of under and excess electricity consumption. 
We also assume that any uncertainty in demand is met by emergency power by the grid. Such assumptions have been made in \cite{oluwasuji_solving_2020}.
These assumptions are necessary to cover uncertainties in the demand consumption by some agents.

A simple approach to this problem is to partition the households into some $q$ subsets such that the sum of demands in each subset is at most $S$, and then connect the agents in each subset for a fraction $1/q$ of the time. To maximize the amount of time each agent is connected, we have to minimize $q$. This problem is equivalent to the classic problem of \emph{bin-packing}. In this problem, we are given some $n$ items, of sizes given by a multiset of positive numbers $D$, and a positive number $S$ representing the capacity of a bin. The goal is to pack items in $D$ in the smallest possible number of bins, so that the sum of the item sizes in each bin is at most $S$. The problem is NP-complete \cite{10.5555/574848}, but it has many efficient approximation algorithms.

However, even an optimal solution to the bin-packing problem may provide a suboptimal solution to the electricity division problem. As an example, suppose that we have three households $x,y,z$ with demands $2,1,1$, respectively, and an electricity supply $S = 3$. Then, the optimal bin-packing results in $2$ bins, for instance, $\{x,y\}$ and $\{z\}$. This means that each agent would be connected $1/2$ of the time. However, it is possible to connect each agents for $2/3$ of the time, by connecting each pair $\{x,y\}, \{x,z\}, \{y,z\}$ for $1/3$ of the time,
as each agent appears in $2$ different subsets.
More generally, suppose that we construct $q$ subsets of agents, such that each agent appears in exactly $k$ different subsets. Then we can connect each subset for $1/q$ of the time, and each agent will be connected for $k/q$ of the time.

\subsection{The \texorpdfstring{$k$}{k}-times bin-packing problem} 
\label{sec:intro subsec:kbp}
To study this problem more abstractly, we define the \emph{$k$-times bin-packing problem (or \kbp)}. 
The input to \kbp{} is a set of $n$ items of sizes given by a multiset $D$,
a positive number $S$ representing the capacity of a bin, and an integer $k\geq 1$.
The goal is to pack items in $D$ in the smallest possible number of bins, such that the sum of the item sizes in each bin is at most $S$, and each item appears in  $k$ different bins, where each item occurs at most once in a bin.
In the above example, $k=2$. 
It is easy to see that, in the above example, $2$-times bin-packing yields the optimal solution to the electricity division problem.

Our first main contribution \textbf{(\Cref{section:existence of finite k for kBP})} is to prove ,that for every electricity division problem, there exists some finite $k$ for which the optimal solution to the \kbp{} problem yields the optimal solution to the electricity division problem. Moreover, we give an explicit upper bound on $k$, as a function of the number of households.

We note that \kbp{} may have other applications beyond the electricity division. For example, it could be used to create a backup of files on different file servers \cite{Jansen_1998}.
We would like to store $k$ different copies of each file, with at most one copy of the same file on the same server. This can be solved by solving \kbp{} on the files as items, and the server disk space as the bin capacity.

Motivated by these applications, we would like to find ways to efficiently solve \kbp{}. However, it is well-known that \kbp{} is NP-hard even for $k=1$. We therefore look for efficient approximation algorithms for \kbp{}. 

\subsection{Using existing bin-packing algorithms for \texorpdfstring{$k$}{k}BP} \label{using existing bin packing algorithms for kBP}
Several existing algorithms for bin-packing can be naturally extended to $k$BP. However, it is not clear whether the extension will have a good approximation ratio.

As an example, consider the simple algorithm called \emph{First-Fit (FF)}: process the items in an arbitrary order; pack each item into the first bin it fits into; if it does not fit into any existing bin, open a new bin for it. 
In the example $D = [10,20,11], S=31$, the $FF$ would pack two bins: $\{10,20\}$ and $\{11\}$. This is clearly optimal. 
The extension of $FF$ to $k$BP would process the items as follows: for each item $x_r$ in the list (in order), suppose that the $b$ bins have been used thus far.  Let $j$ be the lowest index ($1 \leq j\leq b$) such that (a) the bin $j$ can accommodate $x_r$ and (b) the bin $j$ does not contain any copy of $x_r$, should such $j$ exist; otherwise open a new bin with index $j=b+1$.  Place $x_r$ in the bin $j$.

There are two ways to process the input. One way 
is by processing each item $k$ times in sequence. In the above example, with $k=2$, $FF$ will process the items in order $[10^1, 10^2, 20^1, 20^2, 11^1, 11^2]$, where the superscript specifies the instance to which an item belongs.
This results in four bins: \sloppy $\{10^1, 20^1\}, \{10^2, 20^2\}, \{11^1\}, \{11^2\}$, which simply repeats $k$ times the solution obtained from $FF$ on $D$. However, the optimal solution here is 3 bins: $\{10^1,20^1\}, \{11^1,10^2\}, \{20^2,11^2\}$.

Another way 
is to process the whole sequence $D$, $k$ times. In the above example, $FF$ will process the sequence $D_2 = DD = [10^1,20^1,11^1, 10^2, 20^2, 11^2]$ . 
Applying the $FFk$ algorithm to this input instance will result in three bins $\{10^1,20^1\}, \{11^1,10^2\}, \{20^2,11^2\}$, which is optimal.
Thus, while the extension of $FF$ to $k$BP is simple, it is not trivial, and it is vital to study the approximation ratio of such algorithms in this case. 

As another example, consider the approximation schemes of 
de la Vega and Lueker \cite{fernandez_de_la_vega_bin_1981} and Karmarkar and Karp \cite{karmarkar-efficient-1982}.
These algorithms use a linear program that counts the number of bins of each different \emph{configuration} in the packing (see {\Cref{section:EAA:subsection:general}} for the definitions) 
One way in which these algorithms can be extended, without modifying the linear program, is to give $D_k$  as input. But then a configuration might have more than one copy of an item in $D$, which violates the $k$BP constraint. 
Another approach is to modify the constraint in the configuration linear program, to check that there are $k$ copies of each item in the solution, while keeping the same configurations as for the input $D$. Doing so will respect the $k$BP constraint.
Again, while the extension of the algorithm is straightforward, it is not clear what the approximation ratio would be; this is the main task of the present paper.

The most trivial way to extend existing algorithms is to run an existing bin-packing algorithm, and duplicate the output $k$ times. However, this will not let us enjoy the benefits of $k$BP for electricity division (in the above example, this method will yield $4$ bins, so each agent will be connected for $2/4=1/2$ of the time).
Therefore, we present more elaborate extensions, that attain better performance. 
The algorithms we extend can be classified into two classes:

\paragraph{1. Fast constant-factor approximation algorithms \textbf{(\Cref{section:approx_algo})}.} 
Examples are First-Fit ($FF$) and First-Fit-Decreasing ($FFD$). 
For bin-packing, these algorithms find a packing with at most $1.7 \cdot OPT(D)$ and  $\frac{11}{9} \cdot {OPT}(D) + \frac{6}{9}$  bins respectively \cite{Dsa2013,dosa_tight_2007,dosa_tight_2013}
We adapt these algorithms by running them on an instance made of $k$ copies, $DD\ldots D$ ($k$ copies of $D$),
which we denote by $D_k$.
We show that, for $k>1$, the extension of $FF$ to $k$BP (which we call $FFk$) finds a packing with at most $\left(1.5+\frac{1}{5k}\right) \cdot OPT(D_k) + 3\cdot k$ bins.
For any fixed $k>1$, the asymptotic approximation ratio of $FFk$ for large instances (when $OPT(D_k)\to \infty$)
	is $(1.5+\frac{1}{5k})$, which is better than that of $FF$, and improves towards $1.5$ when $k$ increases.

We also prove that the lower bound for $FFDk$ (the extension of $FFD$ to $k$BP) is $\frac{7}{6} \cdot {OPT}(D_k) + 1$, and conjecture by showing on simulated data that $FFDk$ solves $k$BP with at most $\frac{11}{9} \cdot OPT(D_k) + \frac{6}{9}$ bins which gives us an asymptotic approximation ratio of at most $11/9$. 

We also show that the extension of $NF$ (next-fit algorithm) to $k$BP (we call this extension $NFk$) has the asymptotic ratio of 2.

\paragraph{2. Polynomial-time approximation schemes \textbf{(\Cref{section:ptas})}.}
Examples are the algorithms by Fernandez de la Vega and Lueker \cite{fernandez_de_la_vega_bin_1981} and Karmarkar and Karp \cite{karmarkar-efficient-1982}.
We show that the algorithm by Fernandez de la Vega and Lueker can be extended to solve $k$BP using at most $(1 + 2\cdot \epsilon)OPT(D_k) + k$ bins for any fixed $\epsilon \in (0,1/2)$. 
For every $\epsilon >0$, Algorithm 1 of Karmarkar and Karp  \cite{karmarkar-efficient-1982} solves $k$BP using at most  $ (1 + 2 \cdot k \cdot \epsilon)OPT(D_k) + \frac{1}{2 \cdot \epsilon^2} + (2 \cdot k+1)$ bins, and runs in time
$O(n(D_k) \cdot \log {n(D_k)} + T(\frac{1}{\epsilon^2}, n(D_k))$, where $n(D_k)$ is the number of items in $D_k$, and  $T$ is a polynomially-bounded function.
Algorithm 2 of Karmarkar and Karp \cite{karmarkar-efficient-1982} can be generalized to solve $k$BP using at most  $OPT(D_k) + O(k \cdot \log^2 {OPT(D)})$ bins, and runs in time
$O(T(\frac{n(D)}{2},n(D_k)) + n(D_k) \cdot \log{n(D_k)})$.

\subsection{Egalitarian allocation of watts (\Cref{sec: Egalitarian allocation of supply}) }
\label{subsec: Egalitarian allocation of supply: }

At a given point in time, the equal supply allocation problem concerns the allocation of electricity as equal as possible (in kW) for agents that are not connected to electricity all the time, as these are the agents that do not get their stipulated demand. 

\begin{example}
The input set $D$ contains three agents $x,y,z$ with demands $1,2, \text{ and }3$ respectively and supply $S=5$. 
    An equal allocation of time would result in the following possible packing of items in $D$ in $3$ bins: $\{1,2\}, \{2,3\}, \{3,1\}$. The allocation of supply in this case would be: $0.667, 1.33, \text{ and }2$ watts respectively. The solution is fair w.r.t. the equal allocation of time, but unfair w.r.t. the equal allocation of supply. 
    It is unfair to agents with low demand. In this example, the solution is unfair to the agent $x$ with demand $1$.
    A better solution in this example could have $5$ bins in total: 
    $\{1,2\}, \{1,2\}, \{1,2\}, \{1,3\}, \{1,3\}$.
    It will yield the allocation of supply $1,1.2, \text{ and } 1.2$ for the agents $x,y, \text{ and }z$ respectively.
    This solution is much better w.r.t to the equal allocation of supply.
\end{example}

Now, we describe our solution approaches for equal supply allocation problem that distributes the electricity (in kW) as equal as possible:
\begin{itemize}
	\item[--] We derive an instance from the given input instance. In this derived instance, the sum of the item sizes of all the occurrences of items is as equal as possible.
	For example, consider the input instance
	$D = \{x=1,y=2,z=3\}$ and supply $S=5$.
	Let
	$D' = \{1,2,3,1,2,3,1,2,1,1,1\}$
	be the instance derived from $D$. The sum of item sizes of all the occurrences of items $x,y,z$ is $6,6,6$ respectively.
	\item[--] Now, if we use any algorithm that connects the agents equally in terms of time, it will result in an equal allocation of supply for the instance derived in the above step. For the above example, a possible solution could result in $6$ bins:
	$\{1,2\}, \{3,1\}, \{2,3\}, \{1,2\}, \{1\}, \{1\}, \{1\}$. Connecting each bin $1/6$ of the time will result in $1,1,1$ kW of supply for agents $x,y,z$ respectively.
\end{itemize}

However, the above approach suffers if there are some items with very small item sizes compared to the largest item size. Imagine the situation where the smallest item in the instance is, say, $0.05$ and the maximum item size in the instance is, say, $14$. In this case, there are $280$ copies of this smallest item w.r.t a single copy of the maximum item size. The final derived instance will be a large instance in terms of the number of items. Moreover, in this case, the equal supply allocation is limited by the small item size i.e. agent with max demand $14$ will not get electricity more than $0.05$(kW).

However, it is possible to improve the solution: we connect the agent with demand $x=1$ all the time. Now, the remaining agents are packed in the bin whose capacity is $S' = 5 - 1 = 4$ and the structure of each such bin is $\{1,\}$. After packing the remaining items, a possible final solution could be $\{1,2\}, \{1, 3\}$. Since the agent $x=1$ is connected all the time, connecting each bin $1/2$ of the time will result in the allocation of supply w.r.t the remaining agents as $1, 1.5$ which better than the previously obtained solution. Overall, the allocation vector of supply $(1,1,1.5)$ is preferred over $(1,1,1)$.

More generally, suppose we find a set of agents that are connected all the time, then for the remaining agents and the remaining supply we can determine a better allocation vector of supply. 
In \Cref{sec: Egalitarian allocation of supply}, we discuss four heuristic approaches to solve the egalitarian allocation of Watts problem.

\paragraph{Simulation experiments \textbf{( 
\Cref{appendix: approximation ratio of FFk and FFDk algorithms} and 
\Cref{subsec: Experiment: Heuristic algorithms for egalitarian allocation of watts}
)}}
We complement our theoretical analysis of the approximation ratios of $FFk$ and $FFDk$ bin-packing algorithms with
simulation experiments, which validate our conjectures in {\Cref{section:approx-algo:subsection:approx_algo:ffk}} and {\Cref{section:approx-algo:subsection:approx_algo:ffdk}} in \Cref{appendix: approximation ratio of FFk and FFDk algorithms}.
Further, more results that shows the effect of different values of $k$ and uncertainty in the agent's demand has on the number of connection hours, comfort and electricity supplied have been discussed in \Cref{APPENDIX: electricity-distribution-results}. 

The main objective of heuristics algorithms (HA1, HA2, HA3, and HA4) is to distribute electricity as equally as possible. 
The computational results of these heuristics algorithms in terms of electricity allocated in watts have been discussed in \Cref{subsec: Experiment: Heuristic algorithms for egalitarian allocation of watts}.
More results that compare these heuristics in terms of hours of connection to supply, and comfort delivered have been discussed in \Cref{subsec: Experiment: Heuristic algorithms for egalitarian allocation of supply}. 

\paragraph{Electricity distribution \textbf{( 
\Cref{sec: Experimental Results})}}
The fair electricity division problem was introduced by Oluwasuji, Malik, Zhang and Ramchurn
\cite{oluwasuji_algorithms_2018,oluwasuji_solving_2020} under the name of ``fair load-shedding''.
They presented several heuristic as well as ILP-based algorithms, and tested them on a dataset of $367$ households from Nigeria. 
We implement the $FFk$ and $FFDk$ algorithms for finding approximate solutions to $k$BP, and use the solutions to determine a fair electricity allocation. We test the performance of our allocations on the same dataset of Oluwasuji, Malik, Zhang and Ramchurn \cite{oluwasuji_algorithms_2018}.
We compare our results on the same metrics used by Oluwasuji, Malik, Zhang and Ramchurn \cite{oluwasuji_algorithms_2018}. These metrics are utilitarian and egalitarian social welfare and the maximum utility difference between agents. We compare our results in terms of hours of connection to supply on average, utility delivered to an agent on average, and electricity supplied on average, along with their standard deviation. We find that our results surmount their results in terms of the egalitarian allocation of connection time to the electricity 
$FFk$ and $FFDk$ run in time that is nearly linear in the number of agents. We conclude that using $k$BP can provide a practical, fair and efficient solution to the electricity division problem where the objective is to connect each agent as much as possible.

In the same \Cref{sec: Experimental Results}, we present computational results of heuristic algorithms developed for egalitarian allocation of watts. We find that heuristic algorithm HA1 with $FFDk$ outperforms all other heuristics in supplying more electricity (in terms of watts) to the worst-off agent.

In \textbf{\Cref{section:conclusion}}, we conclude with a summary and directions for future work. We defer most of the technical proofs and algorithms to the Appendix.

\section{Related Literature} \label{section:related-literature}
\paragraph{First-Fit.}

We have already defined the working of $FF$ in \Cref{using existing bin packing algorithms for kBP}.
Denote by $FF$ the number of bins used by the First-Fit algorithm, and by $OPT$ the number of bins in an optimal solution for a multiset $D$. An upper bound of $FF \leq 1.7OPT+3$ was first proved by Ullman in 1971 \cite{laboratory_performance_1971}. The additive term was first improved to 2 by Garey, Graham and Ullman \cite{10.1145/800152.804907} in 1972. In 1976, Garey, Graham, Johnson and Yao \cite{GAREY1976257}  improved the bound further to $FF \leq \lceil 1.7OPT \rceil$, equivalent to $FF  \leq 1.7OPT + 0.9$ due to the integrality of $FF$ and $OPT$. This additive term was further lowered to $FF \leq 1.7OPT + 0.7$ by Xia and Tan \cite{xia_tighter_2010}. Finally, in 2013 Dosa and Sgall \cite{Dsa2013}
settled this open problem and proved that $FF \leq \lfloor 1.7OPT \rfloor$, which is tight.

\paragraph{First-Fit Decreasing.} Algorithm \emph{First-Fit Decreasing} ($FFD$) first sorts the items in non-increasing order, and then implements $FF$ on them. 
In 1973, in his doctoral thesis \cite{Johnson1973}, D. S. Johnson proved that $ FFD\ \leq \frac{11}{9}OPT+4$. Unfortunately, his proof spanned more than 100 pages. In 1985, Baker \cite{baker_new_1985} simplified their proof and improved the additive term to $3$. In 1991 Minyi \cite{Minyi} further simplified the proof and showed that the additive term is $1$. Then, in 1997,  Li and Yue \cite{Li_Yue_1997} narrowed the additive constant to $7/9$ without formal proof. Finally, in 2007 Dosa \cite{dosa_tight_2007} proved that the additive constant is $6/9$. They also gave an example which achieves this bound. 

\paragraph{Next-fit.} The algorithm next-fit works as follows: It keeps the current bin (initially empty) to pack the current item. If the current item does not pack into the currently open bin then it closes the current bin and opens a new bin to pack the current item. Johnson in his doctoral thesis \cite{Johnson1973} proved that the asymptotic performance ratio of next-fit is 2.

\paragraph{Efficient approximation schemes.} In 1981, Fernandez de la Vega and Lueker \cite{fernandez_de_la_vega_bin_1981} presented a polynomial time approximation scheme to solve bin-packing. Their algorithm accepts as input an $\epsilon> 0$ and produces a packing of the items in $D$ of size at most $\left(1+\epsilon\right)OPT+1$. Their running time is polynomial in the size of $D$ and depends on $1/\epsilon$. They invented the \emph{adaptive rounding} method to reduce the problem size. In adaptive rounding, they initially organize the items into groups and then round them up to the maximum value in the group. This results in a problem with a small number of different item sizes, which can be solved optimally using the linear configuration program. Later, Karmarkar and Karp \cite{karmarkar-efficient-1982} devised several PTAS for the bin-packing problem. One of the Karmarkar-Karp algorithms solves bin-packing using at most $OPT + O({\log^2\,{ OPT}})$ bins. Other Karmarkar--Karp algorithms have different additive approximation guarantees, and they all run in polynomial time. This additive approximation was further improved to $O(\log\,{OPT} \cdot \log\,{\log\,{OPT}})$ by Rothvoss \cite{Rothvoss_2013}. They used a ``glueing" technique wherein they glued small items to get a single big item. In 2017, Hoberg and Rothvoss \cite{Hoberg_Rothvoss_2017}  further improved the additive approximation to a logarithmic term $O(\log\,{OPT})$.

\paragraph{Bin-packing with constraints.}
Jansen \cite{Jansen_1998} has proposed a $\mathrm{FPTAS}$ for the generalization of the bin-packing problem called \textit{bin-packing with conflicts}. 
The input instance for their algorithm is the conflict graph. Its vertices are the items and any two items are adjacent provided they cannot be packed into the same bin. In particular, $k$BP can be considered as the bin-packing with conflicts, where the conflict graph $D_k$ is a disjoint union of copies of a complete graph $K_k$. Their bin-packing problem with conflicts is restricted to $q-$inductive graphs. In a $q-$inductive graph the vertices are ordered from $1, \ldots, n$. Each vertex in the graph has at most $q$ adjacent lower numbered vertices. 
Since the degree of each vertex of $D_k$ equals $k-1$, $D_k$ is a $k-1$-inductive graph. 
In their method first they obtain an instance of large items from the given input instance. Let this instance be $J_k$. They apply the linear grouping method of Fernandez de la Vega and Lueker \cite{fernandez_de_la_vega_bin_1981} to obtain a constant number of different item sizes. Next they apply the Karmarkar and Karp algorithm \cite{karmarkar-efficient-1982} to obtain an approximate packing of the large items. The bins in this approximate packing may have conflicts, so they use the procedure called $\mathrm{COLOR}$ which places each conflicted item into a new bin. In the worst case it may happen that all the items in each bin have conflict and hence each one of them is packed into a separate bin. Finally, after removing the conflicts, they pack the small items into the existing bins, respecting conflicts among items. In doing so, new bins are opened if necessary. 

For the input instance $I$ that consists of a $q-$inductive graph their algorithm packs the items using at most 
$(1 + 2 \cdot \epsilon) OPT(I) + \frac{2 \cdot q + 1 }{\epsilon^2} + 3\cdot q + 4$
bins. In case of \kbp, $I$ equals $D_k$ and $q$ is $k-1$.

Their algorithm solves the $k$BP using at most $(1 + 2 \cdot \epsilon) OPT(D_k) + \frac{2 \cdot k -1 }{\epsilon^2} + 3\cdot k + 1$ bins. 

In this paper we focus on a special kind of conflicts, and for this special case, we present a better approximation ratio: our extension to Algorithm 1 and 2 of Karmarkar-Karp algorithms solves the $k$BP using at most $(1 + 2 \cdot \epsilon) OPT(D_k) + \frac{1 }{2 \cdot \epsilon^2} + 2\cdot k + 1$ and $OPT(D_k) + O(k \cdot \log^2 OPT(D))$ bins respectively. 

Gendreau, Laporte and Semet \cite{Gendreau_Laporte_Semet_2004} propose six heuristics named H1 to H6 for bin-packing with item-conflicts, where the input instance is 
represented by a general conflict graph.
The heuristic H1 is a variant of $FFD$ which  incorporates the conflicts, whereas H6 is a combination of a maximum-clique procedure and $FFD$. They show that H6 is better than H1 for conflict graphs with high density, whereas H1 performs marginally better for low density conflict graphs (where density is defined as the ratio of the number of edges to the number of possible edges).
\kbp{} can be represented by duplicating each item $k$ times, and constructing a conflict graph in which there are edges between each two copies of the same item. The density of this graph is $(k-1)/(kn-1)$, which becomes smaller for large $n$. This suggests that H1 is a better fit for \kbp{}. 

But if we use their method to solve 
\kbp{}, then the vertices of the conflict graphs are ordered as blocks corresponding to the items, in such a way that their sizes are non-decreasing. We have already seen in \Cref{using existing bin packing algorithms for kBP}
that when we change such item order we can obtain a better packing.

Ekici {\cite{Ekici_2021}} has given a heuristic solution based on linear programming for the bin packing problem with item-conflicts (BPPC) and item fragmentation (BPPIF {\cite{Casazza_Ceselli_2014}}{\cite{Casazza_Ceselli_2016}}). In their problem an item can be fragmented, and also an item can have size larger than the bin capacity. They have also represented the problem by a general conflict graph. In contrast to their problem, our problem have item sizes at most the bin capacity and also we do not allow an item to be fragmented. Hence, their solutions are not directly applicable in our case. Moreover, they do not consider the number of fragments an item can have in their solution.

Gupta and Ho {\cite{Gupta_Ho_1999}}, have developed the MBS (minimum bin slack) procedure that determines the total possible subset of items that can fill the bin and that leave the minimum space remaining in the bin. Their heuristic algorithm finds an optimal solution where the total sum of all items is at most twice the bin capacity. The complexity of their procedure is exponential in input size.
Fleszar and Hindi {\cite{Fleszar_Hindi_2002}}, suggested heuristic algorithms based on MBS to save computation. These improvements first pack the large item in a bin leaving small space for other items to pack, therefore, resulting in a small exponential-time search for the remaining items to fill the bin. In electricity distribution problem, the supply is very large as compared to the demand of the households. Therefore, adoption of such heuristic will be computationally expensive. 

Recently, Doron-Arad, Kulik and Shachnai \cite{Doron-Arad_Kulik_Shachnai_2022} have solved in polynomial time a more general variant of bin-packing, with partition matroid constraints. Their algorithm packs the items in $OPT + O\left(\frac{OPT}{(\ln {\ln OPT})^{1/17}}\right)$ bins. Their algorithm can be used to solve the $k$BP:
for each item in $D$, define a category that contains $k$ items with the same size.
Then, solve the bin-packing with the constraint that each bin can contain at most one item from each category (it is a special case of a partition-matroid constraint).
However, in the present paper we focus on the special case of $k$BP. This allows us to attain a better running-time (with $FFk$ and $FFDk$), and a better approximation ratio (with the de la Vega--Lueker and Karmarkar--Karp algorithms).

\paragraph{Cake cutting and electricity division} The electricity division problem can also be modeled as a classic resource allocation problem known as `cake cutting'. The problem was first proposed by Steinhaus {\cite{Society2016}}. A number of cake-cutting protocols have been discussed in {\cite{brams_taylor_1996,Webb_1998}}. In cake cutting, a cake is a metaphor for the resource. Like previous approaches, a time interval can be treated as a resource. A cake-cutting protocol then allocates this divisible resource among agents who have different valuation functions (or preferences) according to some fairness criteria. The solution to this problem differs from the classic cake-cutting problem in the sense that at any point in time, $t$, the sum of the demands of all the agents, should respect the supply constraint, and several agents may share the same piece.

\paragraph{Other solutions to fair load shedding.}
Load-shedding strategy at the bus level has been discussed in \cite{Shi2015}. They have dealt with the issue of cascading failure, which may occur in line contingencies. So when a line contingency happens, first they do load-shedding to avoid cascading failure. This load-shedding may reduce the loads beyond what is necessary (let's call this unnecessary load) to prevent cascading failure. Then, they recover as much unnecessary load as possible in the second step while maintaining the system's stability.
Shi and Liu \cite{Shi2015} presented a distributed algorithm for compensating agents for load-shedding based on the proportional-fairness criterion. In contrast, we present a centralized algorithm for computing load-shedding that attains a high egalitarian welfare.

In a smart grid, total distributed electricity combines utility supply and distributed renewable sources (DERs). \cite{Jainetal2022} proposes a fair, starvation-free mechanism allocating energy among microgrids so each receives a portion of its demand at lower cost. The method uses two phases: the first proportionally distributes DER energy; the second uses a Vickrey auction for remaining DER energy. If microgrid demand persists, and DER energy is fully allocated, any unmet need is addressed by purchasing from the utility grid. However, for households, when demand exceeds supply, fractional allotments may hinder critical activities.
Bus-level load shedding can be inequitable. An equity-aware model using the grid Gini coefficient ensures fairer load curtailment across buses, but may increase total load shed and cost \cite{Fangetal2024}. Each time an agent connects, recalculating the Gini coefficient is computationally expensive for large groups. 

Gerding et. al. \cite{Gerding2011} proposed a model-free mechanism for coordinating electric vehicle charging that ensures truthfulness by leaving some units unallocated, which return to the grid but are wasted when the purpose is to distribute electricity to households.
In \cite{Stein2012}, the mechanism pre-commits to fulfill each selected agent’s exact demand, assumes no additive value for extra resource, and allows allocations below the reported maximum, in contrast to our mandatory full-demand model.

Buermann et. al. \cite{Buermann2020} studied variable-energy allocation over time with agents’ linear satiable valuations, allowing partial fulfillment, while our model uses piecewise constant utility so agents gain only if their full demand is met; partial allocation is worthless.

Mechanisms in which a load is shifted from peak hours to non-peak hours has been studied in \cite{Ali2021,akasiadis2017mechanism} 
However, in this research, we assumed that an agent's utility for fractional allocation of their demand is $0$, and also shifting an agent's demand to other time-intervals may result in inhibiting an agent from performing some critical activities.

\section{Definitions and Notation} 
\label{section:model-assumptions}
\subsection{Electricity division problem}
The input to the Electricity Division problem consists of:
\begin{itemize}
    \item A number $S>0$ denoting the total amount of available supply (e.g. in kW);
    \item A number $n$ of households, and a list $D = D[1],\ldots,D[n]$ of positive numbers, where $D[i]$ represents the demand of households $i$ (in kW);
    \item An interval $[0,T]$ representing the time in which electricity should be supplied to the households.
\end{itemize}
 
 The desired output consists of:
 \begin{itemize}
     \item A partition $\mathcal{I}$ of the interval $[0,T]$ into sub-intervals, $I_1,\ldots,I_p$;
     \item For each interval $l\in [p]$, a set $A_l \subseteq [n]$ denoting the set of agents that are connected to electricity during interval $l$, such that
    $\sum_{i \in A_l} D[i] \leq S$ (the total demand is at most the total supply).
 \end{itemize}

In terms of connection time to electricity,
the utility of agent $i$ equals the total time  agent $i$ is connected:
$u_i^t(\mathcal{I}) = \sum_{l: i\in A_l}|I_l|$
(we will consider other utility functions in 
\Cref{sec: Experimental Results}
).

The optimization objective is 
\begin{align*}
\max_{\mathcal{I}} \min_{i\in [n]} u_i^t(\mathcal{I}),   
\end{align*}
where the maximum is  over all partitions that satisfy the demand constraints. 
This max-min value is called the \emph{egalitarian connection-time} of the given instance.

In terms of allocation of electricity supply, utility of an agent $i$ is defined as:
$u_i^s(\mathcal{I}) = \sum_{l: i\in A_l} D[i] \cdot |I_l|$.
The optimization objective is:
\begin{align*}
    \max_{\mathcal{I}} \min_{i\in [n]} u_i^s(\mathcal{I})
\end{align*}


\subsection{$k$-times bin packing}
\label{sec:model subsec:kbp}
We denote the bin capacity 
by $S>0$ and the multiset of $n$ items by $D$. 

Let $n(D)$ and $m(D)$ denote the number of items and the number of different item sizes in $D$, respectively. We denote these sizes by $c[1],\dots, c[m(D)]$. Moreover, for each natural $i\le m(D)$ let $n[i]$ be the number of items of size $c[i]$. The \emph{size} of a bin is defined as the sum of all the item sizes in that bin. Given a multiset $B$ of items, we denote by $V(B)$ its \emph{size}, defined asthe sum of the sizes of all items of $B$.

We denote $k$ copies  of $D$ by $D_k := DD\ldots D$. We denote the number of bins used to pack the items in $D_k$ by the optimal and the considered algorithm by $OPT(D_k)$ and $A(D_k)$, respectively, where $A$ is the algorithm used.

Note that each item in $D_k$ is present at most once in each bin, so it is present in exactly $k$ distinct bins. Consider the example in \Cref{sec:introduction}. There are three items $x,y,z$ with demand $2,1,1$ respectively. Let $k=2$ and $S=3$. 
Then, $\{x,y\}, \{y,z\}, \{z,x\}$ is a valid bin-packing. Note that each item is present twice overall, but at most once in each bin.
In contrast, the bin-packing $\{x,y\}, \{y,z,z\}, \{x\}$ is not valid, because there are two copies of $z$ in the same bin.

In the next section, we will show that for every electricity division problem, a finite value of $k$ can be found such that solving the \kbp problem optimally will also produce the optimal solution for the electricity division problem.

\section{On optimal \texorpdfstring{$k$}{k} for \texorpdfstring{$k$}{k}-times bin-packing} \label{section:existence of finite k for kBP}
In this section we prove that, for every electricity division instance, there exists an  integer $k$ such that $k$BP yields the optimal electricity division. Moreover, we give an upper bound on $k$ as a function of the number of agents.

\subsection{Upper bound}
We start with some useful definitions and lemmas from linear algebra.


Let $X$ be a nonempty set. 
We denote by $\mathbb{R}^X$ the linear space of all functions from $X$ to the real numbers $\mathbb{R}$. So elements of $\mathbb{R}^X$ have the form $(w_\alpha)_{\alpha \in X}$, where for each element $\alpha\in X$, $w_\alpha$ is the corresponding real number.

For each nonempty subset $Y$ of $X$, let $\pi_Y : \mathbb{R}^X \rightarrow \mathbb{R}^Y$ be the natural projection, which maps each element $(w_\alpha)_{\alpha \in X} \in \mathbb{R}^X$ to the element 
$(w_\alpha)_{\alpha \in Y} \in \mathbb{R}^Y$.

Let $W\subseteq \mathbb{R}^X$ be a linearly independent set,
and $Y\subseteq X$ a subset of the indices.
In general, the projection  $\pi_Y(W)$ might not be linearly independent. For example, if $X = \{1,2,3,4,5\}$ and $W = \{ [1,2,0,0,0], [2,4,1,0,0], [3,6,0,0,1]\}$
and $Y = \{1,2\}$, then $\pi_Y(W) = \{[1,2],[2,4],[3,6]\}$, which is linearly dependent.

We shall need the following lemmas.
\begin{lemma}
	\label{lem:Y}
	Let $W\subseteq \mathbb{R}^X$ be a nonempty finite linearly independent set. Then there exists a subset $Y$ of $X$ with $|Y|=|W|$ such that the set $\pi_Y(W)$ is linearly independent.
\end{lemma}

In the above example, as $|W|=3$, the lemma says that there exists a subset $Y$ containing $3$ indices, such that the projection of $W$ on these indices is still linearly independent. Indeed, in this case we can take $Y = \{1,3,5\}$, as $\pi_Y(W) = \{[1,0,0],[2,1,0],[3,0,1]\}$, which is linearly independent.\begin{proof}
	Let $q=|W|$. We prove the required claim by induction on $q$. The base case is $q = 1$. 
	Then $W$ consists of a single nonzero vector $w$.
	Therefore there exists $\alpha\in X$ such that the $\alpha$th entry of $w$ is non-zero. Put $Y=\{\alpha\}$. Then the vector $\pi_Y(w)$ is non-zero and so the set $\{\pi_Y(w)\}$ is linearly independent. 
	
	Now suppose that the required claim holds for $q-1$. Pick any vector $w'\in W$ and put $W'=W\setminus\{w'\}$. 
	By the induction assumption,
	there exists a subset $Y'$ of $X$ with $|Y'|=q-1$ such that the set $\pi_{Y'}(W')$ is linearly independent. 
	As $|W'|=q-1$, adding a single vector $\pi_{Y'}(w')$ to $\pi_{Y'} (W')$ makes it linearly dependent.
	Therefore there exists a unique vector of coefficients $(\lambda_v)_{v\in W'}$ such that $\pi_{Y'}(w')=\sum_{v\in W'} \lambda_v \pi_{Y'} (v)$. Since the set $W$ is linearly independent, 
	$\sum_{v\in W'} \lambda_v v \neq w'$,
	so there exists $\alpha\in X\setminus Y'$ such that $w'_\alpha\ne\sum_{v\in W'} \lambda_v v_\alpha$. Put $Y=Y'\cup\{\alpha\}$. 
	
	We claim that the set $\pi_Y(W)$ is linearly independent. Indeed, suppose for a contradiction that there exist coefficients $(\lambda'_v)_{v\in W}$ which are not all zeroes such that $\sum_{v\in W} \lambda'_v\pi_Y(v)=0$. Since the set $\pi_{Y'}(W')$ is linearly independent, the set $\pi_{Y}(W')$ is linearly independent too,
	so $\lambda'_{w'}\ne 0$. Then $\pi_Y(w')=\sum _{v\in W'} (-\lambda'_v/\lambda'_{w'})\pi_Y(v)$, and so 
	$\pi_{Y'}(w')=\sum _{v\in W'} (-\lambda'_v/\lambda'_{w'})\pi_{Y'}(v)$. The uniqueness of $(\lambda_v)_{v\in W'}$ ensures that $-\lambda'_v/\lambda'_{w'}=\lambda_v$ for each 
	$v\in W'$. But $w'_\alpha\ne \sum_{v\in W'} \lambda_v v_\alpha=\sum _{v\in W'} (-\lambda'_v/\lambda'_{w'})v_\alpha$, a contradiction. 
	
	Thus the required claim holds for $q$.
        \qed
\end{proof}

Let $X$ be a nonempty set and let $Z \subset \{0,1\}^X$ be a nonempty finite linearly-dependent set of nonzero vectors.
This means that there exist real coefficients $(x_w)_{w\in Z}$, not all zeros, such that $\sum_{w\in Z}x_w\cdot w = 0$.
The following lemma shows that we can choose these coefficients to be integers with a bounded magnitude.

We recall the following facts from the OEIS \cite{IntegerSequenceEncyclopedia}.
\begin{sloppypar}
For each natural $n$, let $a(n)$ be the maximal determinant of a matrix of order $n$ whose entries are $0$ or $1$. The values of $a(n)$ are known up to { $n=21$: $ \allowbreak 1, \allowbreak 1, \allowbreak 2, \allowbreak 3, \allowbreak 5, \allowbreak 9, \allowbreak 32, \allowbreak 56, \allowbreak 144, \allowbreak 320, \allowbreak 1458, \allowbreak 3645, \allowbreak 9477, \allowbreak 25515, \allowbreak 131072, \allowbreak 327680, \allowbreak 1114112, \allowbreak 3411968, \allowbreak 19531250, \allowbreak 56640625, \allowbreak 195312500,\dots$ }. Hadamard proved that $a(n) \le 2^{-n}(n+1)^{\frac{n+1}{2}}$, with equality iff a Hadamard matrix of order $n+1$ exists. It is believed that the latter holds iff $n+1 = 1, 2$ or a multiple of $4$.
\cite{clements1965sequence} provide a lower bound
$a(n) > 2^{-n}(\frac{3}{4}(n+1))^{\frac{n+1}{2}}$.

Here are two examples of $6\times 6$ matrices that attain the upper bound $a(6)=9$:%
\footnote{
The first is from OEIS;
the second is from Dietrich Burde in
\url{https://math.stackexchange.com/a/4965827}.
}
\begingroup
\setlength{\arraycolsep}{5pt}
\begin{align}
\label{eq:mat9}
\begin{bmatrix} 
  1 & 0 & 0 & 1 & 1 & 0
  \\
  0 & 0 & 1 & 1 & 1 & 1
  \\
  1 & 1 & 1 & 0 & 0 & 1
  \\
  0 & 1 & 0 & 1 & 0 & 1
  \\
  0 & 1 & 0 & 0 & 1 & 1
  \\
  0 & 1 & 1 & 1 & 1 & 0
\end{bmatrix}   
&&
\begin{bmatrix} 
  1 & 0 & 1 & 0 & 0 & 0
  \\
  1 & 1 & 0 & 1 & 0 & 0
  \\
  0 & 1 & 1 & 0 & 1 & 0
  \\
  0 & 0 & 1 & 1 & 0 & 1
  \\
  1 & 0 & 0 & 1 & 1 & 0
  \\
  1 & 1 & 0 & 0 & 1 & 1
\end{bmatrix}   
\end{align}
\endgroup

\end{sloppypar}

\begin{lemma}
	\label{integer coefficeints lemma}
	Let $Z \subset \{0,1\}^X$ be a nonempty finite linearly dependent set of nonzero vectors and
	 $p = |Z| - 1 $. Then there exist integers $(\Delta_w)_{w \in Z}$ which are not all zeros such that $|\Delta_w| \leq a(p)$ for each $w \in Z$ and $\sum_{w \in Z} \Delta_w \cdot w = 0$.
	That is if \emph{some} nontrivial linear combination of $Z$ equals $0$, then there exists such a linear combination in which the coefficients are all integers, and are all bounded by $a(p)$.
\end{lemma}
\begin{proof}
Let $W$ be a maximal linearly independent subset of $Z$ (note that $|W|\leq |Z|-1=p$). By \Cref{lem:Y},  there exists a subset $Y$ of $X$ with $|Y|= |W|$ such that the set $\pi_Y(W) \subset \mathbb{R}^Y$ is linearly independent.
	
Pick any vector $z \in Z \setminus W$. By the maximality of $W$, the set $W\cup\{z\}$ is linearly dependent, and so the set $\pi_Y(W)\cup\{\pi_Y(z)\}$ is linearly dependent too.
Therefore, the system $\sum_{v\ \in W} x_v \cdot \pi_Y(v)=\pi_Y(z)$ has a solution $\mathbf{x}$. Note that this is a square system, with $|Y|$ equations in $|W|=|Y|$ unknowns.

The solution $\mathbf{x}$ is unique, because $\pi_Y(W)$ is linearly independent.
Therefore, $\mathbf{x}$ can be computed by Cramer's rule.
It gives $x_v= \Delta_v/\Delta_z$ for each $ v \in W$, where $\Delta_z$ is the determinant of the matrix with columns $\pi_Y(v)$, and $\Delta_v$ 
is the determinant of the same matrix where column $v$ is replaced with $\pi_Y(z)$. 
All these matrices are square matrices with binary entries and size at most $|Y|\leq p$, so their determinants are integers with absolute value at most $a(p)$.

Since the set $W\cup \{z\}$ is linearly dependent and the system $ \sum_{v \in W} x_v \cdot \pi_Y(v)=\pi_Y(z)$ has a unique solution, the system $ \sum_{v \in W} x_v\cdot v=z$ has (the same) unique solution. 
Multiplying by $\Delta_z$ gives
$\sum_{v \in W} \Delta_v\cdot v - \Delta_z\cdot z = 0$
Putting $ \Delta_v=0$ for each $v\in Z \setminus (W\cup\{z\})$ yields the desired linear combination.
\qed
\end{proof}

We are now ready to prove the main theorem of this subsection.

Given an input set of items $D$, let $OPT(D_k)$ denote the optimal number of bins in $k$-times bin-packing of the items in $D$.

\begin{theorem} 
\label{theorem: finite-k-exists}
For any electricity division instance with 
demand-vector $D$  and supply $S$, 
there exists an integer $k \leq a(n)$ 
such that $\frac{k}{OPT(D_k)}$ is the egalitarian connection-time per agent.
This time can be attained by solving \kbp{} on $D$ and allocating a fraction $\frac{1}{OPT(D_k)}$ of the time to each bin in the optimal solution.
\end{theorem} \label{existence of finite k}
\begin{proof}
	The set $W = \{ (w_1, \ldots, w_n) \in \{  0,1 \}^n : \sum_{i=1}^{n} d_i w_i \leq S \}$ naturally represents all admissible ways to pack the agents into bins (all feasible ``configurations''). The convex hull $\mathrm{CH}(W)$ of $W$ naturally represents the set of all possible schedules, where a schedule is represented by the fraction of time allocated to each configuration.
	
	Denote $e=(1, \ldots, 1)\in\mathbb{R}^n$. Let $r_{\max}$ be the egalitarian connection time, defined by $r_{\max} = \mathrm{sup} \{ r \in [0,1] : r\cdot e \in \mathrm{CH}(W) \}$. Since the set $\mathrm{CH}(W)$ is compact, $r_{\max}$ is attained, that is $r_{\max}\cdot e \in \mathrm{CH}(W)$. The electricity division problem has a solution, so $r_{\max} > 0$. 
	
	Since $r_{\max}\cdot e \in \mathrm{CH}(W)$, by Caratheodory's theorem \cite{enwiki:1216835967_CaratheodoryThm}, there exists a simplex $\mathrm{CH}(W')$ with the set $W'\subset W$ of vertices such that $r_{\max}\cdot e \in \mathrm{CH}(W')$.	
	
	Consider the following illustrating example. Let the supply $S$ is $25$ and there are three agents $x$, $y$, and $z$ with the demands $d_x = 11$, $d_y = 12$, and $d_z = 13$, respectively. It is possible to connect each agent per $r_{\max} = 2/3$ of the time by connecting each of $\{x,y\}, \{y,z\}, \{x,z\}$ per $1/3$ of the time. This can be visually represented as shown in \Cref{fig:two-third-example}.
	\begin{figure}[ht]
		\centering
		\includegraphics[height=0.32\textheight, width=0.32\linewidth, keepaspectratio]{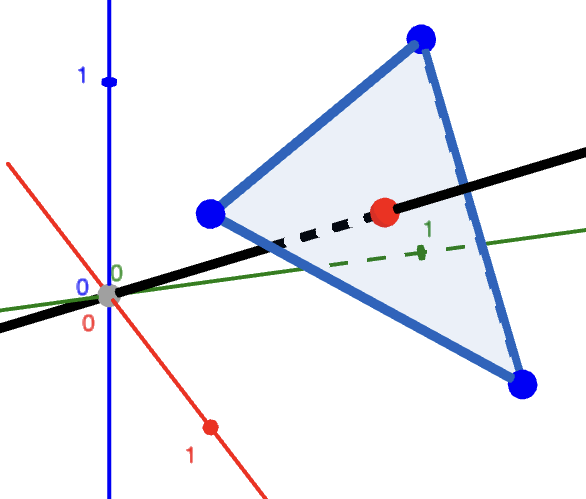}
            \caption{The three blue points are (1,1,0), (1,0,1) and (0,1,1). The black line that originates from the origin intersects with the triangle (convex hull of the three blue points) at red point (2/3,2/3,2/3).}
		\label{fig:two-third-example}
	\end{figure}
	The three blue points $(1,1,0)$, $(1,0,1)$, and $(0,1,1)$ represent elements of $W'$, and the white triangle with these vertices represents $\mathrm{CH}(W')$. The black ray originating from the origin $(0,0,0)$ is $e\cdot \{ r: r\geq 0 \}$. The red point $(2/3, 2/3, 2/3)$   represents the intersection of the ray with $\mathrm{CH}(W')$ and corresponds to $r_{\max} = 2/3$.

	Let $W''$ be the set of the vertices of the face of $\mathrm{CH}(W')$ of minimal dimension containing $r_{\max}\cdot e$. Clearly, $r_{\max} \cdot e$ is a boundary point of $\mathrm{CH}(W')$,  so $|W''| \leq n$. In other words, there exists an optimal electricity division schedule with at most $n$ different configurations. In the above example $W'' = W'$ and  $|W''| = 3 = n$.
	
	Let $(\lambda_w)_{w \in W''}$ be the barycentric coordinates of $r_{\max}\cdot e$, such that $\lambda_w > 0$ for each $w \in W''$ and $\sum_{w \in W''}\lambda_w = 1$, and
	\begin{align}
	\label{eq:rmaxe}
	r_{\max}\cdot e = \sum_{w \in W''} \lambda_w w.    
	\end{align}
	In the above example $\lambda_w = 1/3$ for each $w \in W''$.
	
	We claim that the set $W''$ is linearly independent. Indeed, suppose for a contradiction that there exist disjoint nonempty subsets $W''_1$ and $W''_2$ of $W''$ and positive numbers $(\mu_w)_{w \in W''_1 \cup W''_2}$ such that $ \sum_{w \in W''_1} \mu_w w = \sum_{w \in W''_2} \mu_w w $.
	Assume w.l.o.g. that $ \sum_{w \in W''_1} \mu_w \leq \sum_{w \in W''_2} \mu_w $. Let $ \nu = \min_{w \in W''_2} {\lambda_w / \mu_w}
	$. Then,
	\[
	r_{\max}\cdot e = \sum_{w \in W''} {\lambda_w w} = 
	\sum_{w \in W''} {\lambda_w w}  + \nu \left(\sum_{w \in W''_1} {\mu_w w} - \sum_{w \in W''_2} {\mu_w w}\right) = \sum_{w \in W''} {\lambda'_w w},
	\]
	where
	\begin{equation}
		\lambda'_w  = 
		\begin{cases}
			\lambda_w + \nu \mu_w	& \text{ if } w \in W''_1 \\	
			\lambda_w - \nu \mu_w	& \text{ if } w \in W''_2 \\
			\lambda_w 	& \text{ if } w \in W'' \setminus (W''_1 \cup W''_2)
		\end{cases} 
	\end{equation}
	Then, $ s := \sum_{w \in W''} {\lambda'_w} = \sum_{w \in W''} {\lambda_w}  + \nu (\sum_{w \in W''_1} {\mu_w} - \sum_{w \in W''_2} {\mu_w}) \leq 1, \lambda'_w \geq 0 $ for each $ w \in W'' $, and $ \lambda'_w = 0$ for some $ w \in W''_2$. Then $ s^{-1} \cdot r_{\max} \cdot e \in \mathrm{CH} (W'' \setminus \{w\})$, that contradicts the minimality of $W''$.

Let $Z:= W'' \cup \{e\}$.
As $W''$ contains a multiple of $e$, the set $Z$ is linearly dependent. 
\Cref{integer coefficeints lemma} can be applied to it with $p = |W''|\leq n$.
The lemma implies that there exist integers $(\Delta_v)_{v \in Z}$ which are not all zeroes such that $ |\Delta_v| \leq a(p)\leq a(n)$
 for each $v \in Z$ and $ \sum_{v \in Z} \Delta_v \cdot v = 0$. The set $W''$ is linearly independent, so $ \Delta_e \neq 0$ and 
\begin{align}
\label{eq:esum}
e = \sum_{w \in W''} {(-\Delta_w / \Delta_e)w}.     
\end{align}
But by \eqref{eq:rmaxe},
$e = \sum_{w \in W''} {(\lambda_w/r_{\max})}\cdot w $. As $W''$ is linearly independent, the coefficients in the expression for $e$ are unique, so we must have $-\Delta_w / \Delta_e = \lambda_w / r_{\max} > 0 $ for each $w \in W''$. 
Therefore, $-\Delta_w / \Delta_e = |\Delta_w| / |\Delta_e|$.

Construct a packing in which each configuration $w\in W''$ appears $|\Delta_w|$ times.
Then the vector 
$\sum_{w \in W''} {|\Delta_w|\cdot w}$ represents the number of times each item appears in the packing; but by \eqref{eq:esum}, this sum equals $|\Delta_e|\cdot e$. Therefore, it is a $k$-times bin-packing with $k := |\Delta_e|\leq a(n)$.

The total number of bins in the packing is $\sum_w |\Delta_w| = 
\sum_w |\Delta_e|\cdot \lambda_w/r_{\max}
=k/r_{\max}$.
Therefore, connecting each bin for an equal amount of time yields an allocation in which each agent is connected for a fraction $k/(k/r_{\max}) = r_{\max}$ of the time, as required.
\qed
\end{proof}

\subsection{Lower bound}
\label{subsec: lower bound on k}
For any electricity division instance $D$,
let $K(D)$ denote the smallest $k$ such that 
the egalitarian connection time equals $\frac{k}{OPT(D_k)}$.
For any $n$, let $K(n)$ denote the maximum $K(D)$ over all instances with $n$ households.
\Cref{theorem: finite-k-exists} implies $K(n)\leq a(n)$, which provides an exponential upper bound on $K(n)$.
This raises the question of whether there is a matching lower bound.
Currently, we only have two very loose lower bounds.

\begin{proposition}
\label{prop:lower:n-1}
For any $n\geq 2$, there is a lower bound $K(n)\geq n-1$
\end{proposition}
\begin{proof}
Consider an instance with $n$ items of size $1$, and let $S:=n-1$. 
Then we can construct $n$ bins, each of which contains a different subset of $n-1$ items. 
This packing is optimal, as all bins are full; it is an $(n-1)$-times bin-packing, as each item appears in exactly $n-1$ different bins.
Therefore, $r_{\max} = (n-1)/n$.
We show that any $k<n-1$ does not yield an optimal packing. For any $k$, the sum of all item sizes \dk{in $D_k$} is $k n$, so any $k$-times bin-packing must contain at least $k n / S = k\frac{n}{n-1}$ bins. But $n$ and $n-1$ are coprime, so for any $k<n-1$, the expression $k\frac{n}{n-1}$ is not an integer, which means that the packing must have non-full bins. \qed
\end{proof}

\begin{proposition}
\label{prop:lower:9}
$K(6) = 9$.
\end{proposition}
\begin{proof}
As $K(6)\leq a(6)=9$, it is sufficient to prove the lower bound $K(6)\geq 9$.

Let $A$ be the rightmost $6\times 6$ matrix in \eqref{eq:mat9}. Let $S := \det(A) =9$.

Let $B := \det(A)\cdot A^{-1}$, and note that $B$ is an integer matrix.

Let the vector of demands be $D := B\cdot e = [4,2,5,3,2,1]$.
By construction, Each row in $A$ corresponds to a configuration with sum exactly $9$. These configurations are: $4+5, 4+2+3, 2+5+2', 5+3+1, 4+3+2', 4+2+2'+1$ (where $2'$ denotes the second household with demand $2$).

Construct a packing by $\mathbf{x} := B^T \cdot e = [1,2,3,5,2,4]$.
Every element $x_i$ represents the number of times that configuration $i$ appears in the packing. In particular, there are ---
\begin{itemize}
\item $1$ bin with $4+5$;
\item $2$ bins with $4+2+3$;
\item $3$ bins with $2+5+2'$;
\item $5$ bins with $5+3+1$;
\item $2$ bins with $4+3+2'$;
\item $4$ bins with $4+2+2'+1$.
\end{itemize}
By construction, each item appears exactly $9$ times, so it is a valid $9$-times bin-packing. It has $17$ bins, and it is optimal since all bins are full. Therefore, $r_{\max} = 9/17$.

We now show that any $k<9$ does not yield an optimal packing. For any $k$, the sum of all bins in $k$BP is $17 k$, so any $k$BP must contain at least $17 k / 9$ bins. 
But $17$ and $9$ are coprime, Therefore, for any $k<9$, the expression $17 k / 9$ is not an integer, which means that the packing must have non-full bins. \qed
\end{proof}
It is open whether \Cref{prop:lower:9} can be generalized.
The following is our conjectured generalization of \Cref{prop:lower:9}.

\begin{conjecture}
\label{prop:coprime}
Let $A$ be an $n\times n$ binary matrix, 
let $B := \det(A)\cdot A^{-1}$,
and let $g$ be the sum of all elements in $B$ (the ``grand sum'' of $B$).
If $g$ and $\det(A)$ are coprime, then $K(n)\geq \det(A)$.
\end{conjecture}
Note that \Cref{prop:lower:9} is a special case with $\det(A)=9$ and $g=17$.
\section{Fast Approximation Algorithms} \label{section:approx_algo}
The results of the previous section are not immediately applicable to fair electricity division, as $k$BP is known to be an NP-hard problem. However, they do hint that good approximation algorithms for $k$BP can provide good approximation for electricity division.
Therefore, in this section, we study several fast approximation algorithms for $k$BP.

\subsection{\texorpdfstring{$FFk$}{FFk} --- First-Fit for \texorpdfstring{$k$}{k}BP} \label{section:approx-algo:subsection:approx_algo:ffk}
The $k$-times version of the First-Fit bin-packing algorithm packs each item of $D_k$ in order into the first bin where it fits and does not violate the constraint that each item should appear in a bin at most once. If the item to pack does not fit into any currently open bin, $FFk$ opens a new bin and packs the item into it. For example: consider $D = \{10,20,11\}, k=2, S=31$. $FFk$ will result the bin-packing $\{10,20\}, \{11,10\}, \{20,11\}$. 
It is known that the asymptotic approximation ratio of $FF$ is $1.7$ \cite{Dsa2013}.
Below, we prove that, for any fixed $k>1$, the asymptotic approximation ratio of $FFk$ for large instances (when 
\dk{$n \to \infty$}
) is better, and it improves when $k$ increases.

\begin{toappendix}
    \section{Fast Approximation Algorithms}
    \label{appendix: ffk-algorithm: ffk-general-case-proof}
\end{toappendix}

\begin{theoremrep} \label{ffk:theorem3/2}
	For every input $D$ and $k\geq 1$, $FFk(D_k) \leq \left(1.5+\frac{1}{5k}\right) \cdot OPT(D_k) + 3\cdot k$.
\end{theoremrep}
\begin{inlineproof}
	At a very high level, the proof works as follows:
	\begin{itemize}
		\item We define a \emph{weight} for each item, which depends on the item size, and may also depend on the instance to which the item belongs.
		\item We prove that the average weight of each bin in an optimal packing is at most some real number $Z$, so the total weight of all items is at most $Z\cdot OPT(D_k)$.
		\item We prove that the average weight of a bin in the $FFk$ packing (except some $3 k$ bins that we will exclude from the analysis) is at least some real number $Y$, so the total weight of all items is at least $Y\cdot (FFk(D_k) - 3 \cdot k)$.
		\item Since the total weight of all items is fixed, we get $FFk(D_k) - 3\cdot k \leq (Z/Y)\cdot OPT(D_k)$, which gives an asymptotic approximation ratio of $Z/Y$.
	\end{itemize}
	
	\paragraph{\textbf{Basic weighting scheme}}
	
	The basic weighting scheme we use follows \cite{Dsa2013}. The weight of any item of size $v$ is defined as
	\begin{align*}
		w(v) := v/S + r(v),
	\end{align*}
	where $r$ is a \emph{reward function}, computed as follows:
	\[
	r(v) := 
	\begin{cases}
		0 \quad & \text{if } v/S \leq \frac{1}{6}, \\
		\frac{1}{2}(v/S - \frac{1}{6}) \quad & \text{if } v/S \in (\frac{1}{6}, \frac{1}{3}), \\
		1/12 \quad & \text{if } v/S \in [\frac{1}{3}, \frac{1}{2}], \\
		4/12 \quad & \text{if } v/S > \frac{1}{2}.
	\end{cases}
	\]

	For $k>1$, we use a modified weighting scheme, which  gives different weights to items that belong to different instances. This allows us to get a better asymptotic ratio. We describe the modified weighting scheme below.
	
	\paragraph{\textbf{Associating bins with instances}}
	
	Recall that $FFk$ processes one instance of $D$ completely, and then starts to process the next instance. 
	We associate each bin in the $FFk$ packing with the instance in which it was opened. So for every $j \in\{1,\ldots,k\}$, the \emph{bins of instance $j$} are all the bins, whose first allocated item comes from the $j$-th instance of $D$.
	Note that bins of instance $j$ do not contain items of  instances $1,\ldots,j-1$, but may contain items of any instance $j,\ldots,k$. 
	
	For the analysis, we also need to associate some bins in the optimal packing with specific instances.  We consider some fixed optimal packing.
	For each bin $B$ in that packing:
	\begin{itemize}
		\item If $B$ contains an item $x$ with $V(x)>S/2$, and $x$ belongs to instance $j$, 
		then we associate bin $B$ with instance $j$. Clearly, there can be at most one item of size larger than $S/2$ in any feasible bin, so there is no ambiguity.
		Moreover, we ensure that all other items in $B$ belong to instance $j$ too: if some other item $x'\in B\setminus \{x\}$ belongs to a different instance $j'$, then we replace $x'$ with its copy from instance $j$. Note that the other copy of $x'$ must be located in a bin different than $B$, due to the restrictions of $k$BP. Since both copies of $x'$ have the same size, the size of all bins remains the same.
		\item If $B$ contains no item of size larger than $S/2$, then we do not associate $B$ with any instance. 
	\end{itemize}
	
	\paragraph{\textbf{Partitioning $FFk$ bins into groups}}
	For the analysis, we partition the bins of each instance $j\in[k]$ in the $FFk$ packing into five groups.
	
	\textbf{Group 1}.
	Bins with a single item of instance $j$, whose size is at most $S/2$. 
	There is at most one such bin.
	This is because, if there is one such bin $B$ of instance $j$,
	it means that all later items of instance $j$ do not fit into $B$, so their size must be greater than $S/2$. Therefore, any later bin $B'$ of instance $j$ must have size larger than $S/2$.
	
	\textbf{Group 2}.
	Bins with two or more items of instance $j$, whose total size is at most $2S/3$.
	There is at most one such bin.
	This is because, if there is one such bin $B$ of instance $j$,
	it means that all later items of instance $j$ do not fit into $B$, so their size must be greater than $S/3$. Therefore, in any later bin $B'$ with two or more items of instance $j$, their total size is larger than $2S/3$.
	
	\textbf{Group 3}.
	Bins with a single item of instance $j$, whose size is larger than $S/2$. 
	
	\textbf{Group 4}.
	Bins with two or more items of instance $j$, whose total size is at least $10 S/12$.
	
	\textbf{Group 5}.
	Bins with two or more items of instance $j$, whose total size is in $(2 S/3, 10 S/12)$.
	
	In the upcoming analysis, we will exclude from each instance, the at most one bin of group 1, at most one bin of group 2, and at most one bin of group 5. All in all, we will exclude at most $3k$ bins of the $FFk$ packing from the analysis.
	
	\paragraph{\textbf{Analysis using the basic weighting scheme}}
    As a warm up
    we prove that, with the basic weighting scheme, the total weight of each optimal bin $B$ is at most $17/12$. Since the total size of $B$ is at most $S$, we have $w(B)\leq 1+r(B)$, so it is sufficient to prove that the reward $r(B)\leq 5/12$. Indeed:
	\begin{itemize}
		\item If $B$ contains an item larger than $S/2$, then 
		this item gives $B$ a reward of $4/12$;
		the remaining room in $B$ is smaller than $S/2$. 
  This can accommodate either a single item of size at least $S/3$ and some items smaller than $S/6$,
  or two items of size between $S/6$ and $S/3$; in both cases, the total reward is at most $1/12$. 
            \item If $B$ does not contain an item larger than $S/2$, then there are at most $5$ items larger than $S/6$, so at most $5$ items with a positive reward. The reward of each item of size at most $S/2$ is at most $1/12$.
	\end{itemize}
 In both cases, the total reward is at most $5/12$.	
 
	We now prove that, with the basic weighting scheme, the average weight of each $FFk$ bin is at least $10/12$. 
	In fact, we prove a stronger claim: we prove that, for each instance $j$, the average weight of the items \emph{of instance $j$ only} in bins of instance $j$ (ignoring items of instances $j+1,\ldots k$ if any) is at least $10/12$.
	
	We use the above partition of the bins into 5 groups, excluding at most one bin of group 1 and at most one bin of group 2.
	
	The bins of group 3 (bins with a single item of instance $j$, which is larger than $S/2$) have a reward of $4/12$, so their weight is at least $1/2+4/12 = 10/12$.
	
	The bins of group 4 (bins with two or more items of instance $j$, which have a total size larger than $10 S/12$) already have weight at least $10/12$, regardless of their reward.
	
	We now consider the bins of group 5 (bins with two or more items of instance $j$, which have a total size in $(2S/3, 10S/12)$).
	Denote the bins in group 5 of instance $j$, in the order they are opened, by $B_1,\ldots,B_Q$.
	Consider a pair of consecutive bins, for instance, $B_t, B_{t+1}$.
	Let $x := $ the free space in bin $B_t$ during instance $j$.
	Note that $x\in (S/6, S/3)$. Let $c_1,c_2$ be some two items of instance $j$ which are packed into $B_{t+1}$.
	
	For each $c_i \in \{c_1,c_2\}$, as $FFk$ did not pack $c_i$ into $B_t$ during the processing of instance $j$, there are two options:
	either $c_i$ is too large (its size is larger than $x$), or $B_t$ already contains other copies of $c_i$ from other instances. But the second option cannot happen, because $B_t$ belongs to instance $j$, so it contains no items of instances $1,\ldots,j-1$; and while $c_i$ of instance $j$ was processed by $FFk$, items of instance $j+1,\ldots,k$ were not processed yet. Therefore, necessarily $V(c_i)>x$.
	
	Since $S/6<x<S/3$, the reward of each of $c_1,c_2$ is at least $(x/S-1/6)/2$, and the reward of both of them is at least $x/S-1/6$. Therefore, during instance $j$,
	\begin{align*}
		V(B_t)/S + r(B_{t+1}) \geq (S-x)/S + (x/S-1/6)
		= 1-1/6 = 10/12.
	\end{align*}
	In the sequence $B_1,\ldots,B_Q$, there are $Q-1$ consecutive pairs. Using the above inequality, we get that the total weight of these bins is at least $10/12\cdot (Q-1)$, which is equivalent to excluding one bin (in addition to the two bins excluded in groups 1 and 2).
	
	Summing up the weight of all non-excluded bins gives at least $(10/12)\cdot (FFk(D_k)-3k)$. Meanwhile, the total weight of optimal bins is at most $(17/12)\cdot OPT(D_k)$. This yields $FFk(D_k)\leq (17/10)\cdot OPT(D_k) + 3 k$, which corresponds to the known asymptotic approximation ratio of $1.7$ for $k=1$.
	
	We now present an improved approximation ratio for $k>1$. 
	To do this, we give different weights to items of different copies of $D$. 
	We do it in a way that the average weight of the $FFk$ bins (except the $3k$ excluded bins) will remain at least $10/12$. 
    So the total weight of all items is at least $(10/12)(FFk(D_k) -3k) \leq V(D)$.
	Meanwhile, the average weight of the optimal bins in some $k-1$ instances (as well as the unassociated bins) will decrease to $15/12$, whereas the average weight of the optimal bins in a single instance will remain $17/12$.
    So the total weight of all items is at most $ V(D) \leq \frac{(17/12) + (15/12)\cdot(k-1)}{k} \cdot OPT(D_k)$. Therefore, 
    \[
    FFk(D_k) \leq (1.5 + \frac{1}{5k})\cdot OPT(D_k) + 3k.
    \]
	This would lead to an asymptotic ratio of at most $1.5+\frac{1}{5k}$.
	
	\paragraph{\textbf{Warm-up: $k=2$}}
    To explain the main idea of our analysis,
	we will analyze the case $k=2$, and prove an asymptotic approximation ratio of $1.5+\frac{1}{2\cdot 5} = 1.6$.
	We give a detailed proof of the general case of any $k \geq 1$ in \Cref{appendix: ffk-algorithm: ffk-general-case-proof}. 
	
	We call a bin in the $FFk$ packing \emph{underfull} if it has only one item, and this item is larger than $S/2$ and smaller than $2S/3$. Note that all underfull bins belong to group 3. 
	We consider two cases.
	
	\paragraph{Case 1} No $FFk$ bin of instance $1$ is underfull (in instance $2$ there may or may not be  underfull bins). 
	In this case, we modify the weight of items in instance $1$ only: we reduce the reward of each item of size $v>S/2$ in instance $1$ from $4/12$ to $2/12$. 
	\begin{itemize}
		\item For any bin $B$ in the optimal packing, if $B$
		does not contain an item larger than $S/2$, then its maximum possible reward is still $3/12$ as with the basic weights. 
		If $B$ contains an item larger than $S/2$, then this item gives $B$ a reward of $2/12$ if it belongs to instance $1$, or $4/12$ if it belongs to instance $2$. The remaining room in $B$ is smaller than $S/2$, and the maximum total reward of items that can fit into this space is $1/12$ 
(as the remaining space can accommodate either one item of size at least $S/3$ and some items of size smaller than $S/6$, or two items of sizes in $(S/6,S/3)$). 
Overall, the reward of $B$ is at most $5/12$ if it contains an item larger than $S/2$ from instance $2$, and at most $3/12$ otherwise. As at most half the bins in the optimal packing contain an item larger than $S/2$ from instance $2$, at most half these bins have reward $5/12$, while the remaining bins have reward at most $3/12$, so the average reward is at most $4/12$.
		Adding the size of at most $1$ per bin leads to an average weight of at most $16/12$.
		\item In the $FFk$ packing, we note that the change in the reward affects only the bins of group $3$ (bins with a single item and $V(B)>S/2$). These bins now have a reward of at least $2/12$. Since none of them is underfull by assumption, their weight is at least $2/3+2/12 = 10/12$.
		The bins of groups $4$ and $5$ are not affected by the reduced reward: the same analysis as above can be used to deduce that their average weight (except one bin per instance excluded in group $5$) is at least $10/12$.
	\end{itemize}
	
	\paragraph{Case 2} At least one $FFk$ bin of instance $1$ is underfull.
	In this case, we modify the weight of items in instance $2$ only, as follows  (note that we reduce the \emph{weights} of these items, and not only their rewards):
	\[
	w(v) = 
	\begin{cases}
		0 \quad & \text{if } v/S < \frac{1}{3}, \\
		5/12 \quad & \text{if } v/S \in [\frac{1}{3}, \frac{1}{2}], \\
		10/12 \quad & \text{if } v/S > \frac{1}{2}.
	\end{cases}
	\]
	Note that all these weights are not higher than the basic weights of the same items. 
	\begin{itemize}
		\item In the optimal packing, every bin that contains an item larger than $S/2$ from instance $2$ is associated with instance $2$, and therefore contains only items of instance $2$ by construction. Therefore, the weight of any such bin is at most $15/12$ (the remaining space in such bin can pack only one item 
        of weight in the range $[S/3, S/2]$        ).
		The weight of every bin that contains no item larger than $S/2$ is still at most $15/12$, since its reward is at most $3/12$ as with the basic weights.
		The only bins that may have a larger weight (up to $17/12$) are those that contain an item larger than $S/2$ from instance $1$; at most half the bins in the optimal packing contain such an item. Therefore, the average  weight of a bin in the optimal packing is at most $16/12$.
		\item In the $FFk$ packing, 
		some bin $B^1$ in instance $1$ is underfull, that is, $B^1$ contains a single item of size $v\in(S/2, 2S/3)$, and the free space in $B^1$ is $S-v\in (S/3, S/2)$. 
		This means that, in the bins of instance $2$, there is no item of size at most $S/3$, since every such item would fit into $B^1$ (it is smaller than the available space in $B^1$, and it is not a copy of the single item in $B^1$).
		So for every bin $B^2$ in instance $2$ (except the excluded ones), there are only two options:
		\begin{itemize}
			\item $B^2$ contains one item larger than $S/2$ --- so its weight is $10/12$;
			\item $B^2$ contains two items, each of which is larger than $S/3$ --- so its weight is at least $2\cdot 5/12 = 10/12$.
		\end{itemize}
		The weight of items in instance $1$ does not change. So the average weight of bins in the $FFk$ packing (except the $3k$ excluded bins) remains at least $10/12$.
	\end{itemize}
	
	In all cases, for $k=2$, we get that the average weight of bins in the optimal packing is at most $16/12$, and the average weight of non-excluded bins in the $FFk$ packing is at least $10/12$. Therefore,
	$FFk -3k \leq (16/10)\cdot OPT = 1.6\cdot OPT$.
\qed	
\end{inlineproof}
\begin{toappendix}
    In this section we will give the proof of \Cref{ffk:theorem3/2} for the general case. Recall that we have already gave the proof for $k=2$ as a warm up case in \Cref{section:approx-algo:subsection:approx_algo:ffk}.    
\end{toappendix}
\begin{proof}
	Recall that we have a basic weighting scheme, in which each item of size $v$ is given a weight 
	\begin{align*}
		w(v) := v/S + r(v),
	\end{align*}
	where $r$ is a \emph{reward function}, computed as follows:
	\[
	r(v) := 
	\begin{cases}
		0 \quad & \text{if } v/S \leq \frac{1}{6}, \\
		\frac{1}{2}(v/S - \frac{1}{6}) \quad & \text{if } v/S \in (\frac{1}{6}, \frac{1}{3}), \\
		1/12 \quad & \text{if } v/S \in [\frac{1}{3}, \frac{1}{2}], \\
		4/12 \quad & \text{if } v/S > \frac{1}{2}.
	\end{cases}
	\]
	
	We modify this scheme later for some items, based on the instances to which they belong.
	
Denote by $u$, the first instance in the $FFk$ packing in which there are underfull bins (if there are no underfull bins at all, we set $u = k$). We modify the item weights as follows:
\begin{itemize}
\item For instances $1,\ldots, u-1$, we reduce the reward of each item of size $v>S/2$ from $4/12$ to $2/12$.
\item For instance $u$, we keep the weights unchanged.
\item For instances $u+1,\ldots, k$, 
we reduce the weights of items as in case $2$ of warm-up case for $k=2$ in the proof of \Cref{ffk:theorem3/2}), that is 
\[
w(v) = 
\begin{cases}
    0 \quad & \text{if } v/S < \frac{1}{3}, \\
    5/12 \quad & \text{if } v/S \in [\frac{1}{3}, \frac{1}{2}], \\
    10/12 \quad & \text{if } v/S > \frac{1}{2}.
\end{cases}
\]
\end{itemize}
Now we analyze the effects of these changes.
\begin{itemize}
\item In the optimal packing, 
the maximum reward of every  bin that contains an item larger than $S/2$ from instances $1,\ldots,u-1$ drops from at most $5/12$ to at most $3/12$, so their weight is at most $15/12$.
Additionally, every bin that contains an item larger than $S/2$ from some instance $j\in\{u+1,\ldots, k\}$, contains only items from instance $j$. Therefore, its weight can be at most $15/12$ (the remaining space in such bin can pack only one item from $[S/3,S/2]$ of positive weight. 
Additionally, the weight of every bin that contains no items larger than $S/2$ remains at most $15/12$.
The only bins that can have a larger weight (up to $17/12$) are bins with items larger than $S/2$ from instance $u$. At most $1/k$ bins in the optimal packing can contain such an item.
Therefore, the average weight of a bin in the optimal packing is at most $\frac{(17/12) + (15/12)\cdot(k-1)}{k}$.
		
		\item 
		In the $FFk$ packing, 
		instances $1,\ldots,u-1$ have no underfull bins. 
		The change in the reward affects only the bins of group 3
		(bins with a single item larger than $S/2$): the weight of each bin in group 3 (which must have one item larger than $2S/3$ since it is not underfull) is at least $2/3+2/12 = 10/12$. 
		Since we only reduce the reward of the item larger than $S/2$ for the instances $1, \ldots ,u-1$; the analysis of bins in groups 4 and 5 is not affected at all by this reduction in reward and therefore leads to their average weight (except excluding one bin per instance in groups 1, 2, and 5) of at least $10/12$.
		
		In instance $u$, the weights are unchanged, so the average bin weight (except excluding one bin per instance in groups 1, 2, and 5) is still at least $10/12$.
		
		Finally, the bins of instances $u+1,\ldots,k$ contain no item of size at most $S/3$ --- since any such item would fit into the underfull bin in instance $u$. Therefore, for each such bin, there are only two options:
		\begin{itemize}
			\item the bin contains one item larger than $S/2$ --- so its weight is $10/12$;
			\item the bin contains two items, each of which is larger than $S/3$ --- so its weight is at least $2\cdot 5/12 = 10/12$.
		\end{itemize}
	\end{itemize}
	Again, the average bin weight (except excluding one bin per instance in groups 1, 2, and 5) of instances $u+1, \ldots, k$ is at least 10/12.
	
	We conclude that the average weight of all bins in the $FFk$ packing (except the $3k$ excluded bins) is at least $10/12$. Therefore
	$$FFk \leq \left(\frac{\frac{(17/12) + (15/12)\cdot(k-1)}{k}}{10/12} \cdot OPT\right) + 3k=\left(1.5+\frac 1{5k}\right) \cdot OPT + 3k.$$	
 \qed
\end{proof}

\paragraph{Approximation ratio lower bound}
For $k=1$, the approximation ratio of $1.7$ is tight; an example is given in \cite{Dósa_Sgall_2014} and on page 306 in \cite{Johnson1974}.
In \Cref{appendix: ffk-worst-case-example-johnsonpaper}, we analyze these examples for $k=2$, and show that $FFk$ achieves a ratio of $1.35$. For the example given in \cite{Johnson1974} we observed that as $k$ increases, the approximation ratio continues to decrease. This raises the question of whether the bound of \Cref{ffk:theorem3/2} is tight.

Although we do not have a complete answer for this question, we show below an example in which $FFk$ attains an approximation ratio of $1.375$ for $k=2$, which is higher than the example of \cite{Johnson1974}.

Consider the example $D = \{371, 659, 113, 47, 485, 3, 228, 419, 468, 581, 626 \}$ and bin capacity ${S} = 1000$. Then, $FFk$ for $k=2$ will result in 11 bins $[371, 113, 47, 3, 228],$ $[659, 113, 47, 3],$ $[485, 419],$ $[468, 371],$ $[581, 228],$ $[626], [659],$ $[485, 419],$ $[468], [581], [626]$. The optimal packing for the items in $D_2$ is $[626, 371, 3],$ $[626, 371, 3]$, $[659, 228, 113],$ $[659, 228, 113],$ $[419, 581],$ $[419, 581],$ $[468, 485, 47],$ $[468, 485, 47]$. Clearly, $\frac{FFk(D_2)}{OPT(D_2)} = \frac{11}{8} = 1.375$. 
We observed that as $k$ increases, the approximation ratio continues to decrease.
Experimental results in support of this ratio are given in \Cref{section:results}.
Based on the example for which we have achieved a ratio of $1.375$ and the simulation of the algorithm on different datasets, we conjecture that for $k>1$ the \emph{absolute} approximation ratio for $FFk$ is $1.375$. 

\begin{toappendix}
    \subsection{Worst-case example}
    \label{appendix: ffk-worst-case-example-johnsonpaper}

    \begin{lemmarep} \label{lower bound 1}
		The approximation ratio of $FFk(D_k)$ for $k=2$ for the worst case examples given in \cite{Dósa_Sgall_2014,Johnson1974} is $1.35$. For the example given in \cite{Johnson1974}, the approximation ratio continues to decrease as $k$ increases.
	\end{lemmarep}
	\begin{proof}
		Dosa and Sgall \cite{Dósa_Sgall_2014} gave a simple example for the lower bound construction of $FF$ as following: For a very small $\epsilon$, first there are $10n$ \emph{small} items of size approximately $1/6$ (smaller and bigger), then comes $10n$ \emph{medium} items of size approximately $1/2$ (smaller and bigger), and finally $10n$ \emph{large} items of size $1/2 + \epsilon$ follows. They have proved that for this list $D$, $OPT(D)=10n$ and $FF(D)=17n$. Now, consider the packing of the items in $D_2$. After packing the first instance of $D$, $FFk$ packs the \emph{small} and \emph{medium} items into the existing bins. Only each of the \emph{large} items require a separate bin to pack. So it is easy to see that $FFk(D_2)=(17+10)n$ while $OPT(D_2)=(10+10)n$. This will give us $\frac{FFk(D_2)}{OPT(D_2)} = \frac{27}{20} = 1.35$. Note that here $OPT$ is arbitrarily big.
		
		Johnson, Demers, Ullman, Garey and Graham \cite{Johnson1974} provide the following example for $FF$. We now show that, on the same example, $FFk$ achieves an approximation ratio of $1.35$.
		
		All the item sizes in this example 
		$D$ are in non-decreasing order and are as follows
		$D = \{6 (7),10 (7),16 (3),34 (10),51 (10)\}$ and bin size $S = 101$.
		The number in the parenthesis denotes the occurrence of that item in $D$. In the above example $6 (7)$ expands to $\{6,6,6,6,6,6,6\}$. Before proceeding further, we describe some notation that we will use. For two bins $B_i, B_j$, $i<j$ implies bin $B_i$ has been created before bin $B_j$.

        \begin{figure}
\centering
    \begin{subfigure}[b]{0.3\textheight}
        \includegraphics[height=0.9\textheight, width=0.9\textwidth, keepaspectratio]{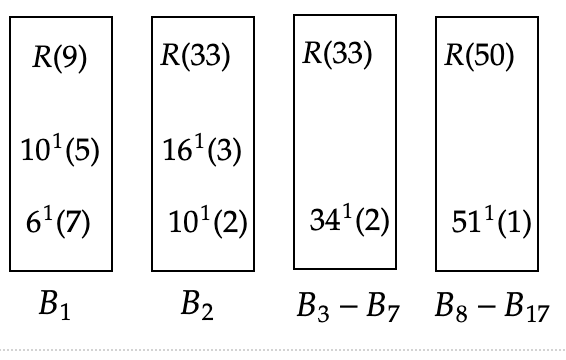}
        \caption{$FFk$ packing for $k=1$.}
        \label{ffk_lb_k1}
    \end{subfigure}

    \begin{subfigure}[b]{0.55\textheight}
        \includegraphics[height=0.9\textheight, width=0.9\textwidth, keepaspectratio]{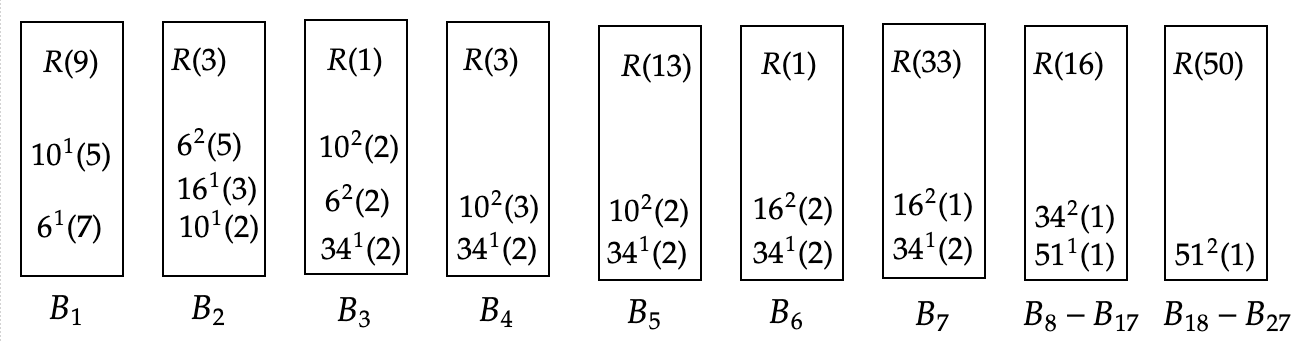}
        \caption{$FFk$ packing for $k=2$.}
        \label{ffk_lb_k2}
    \end{subfigure}

    \begin{subfigure}[b]{0.4\textheight}
        \includegraphics[height=0.9\textheight, width=0.9\textwidth, keepaspectratio]{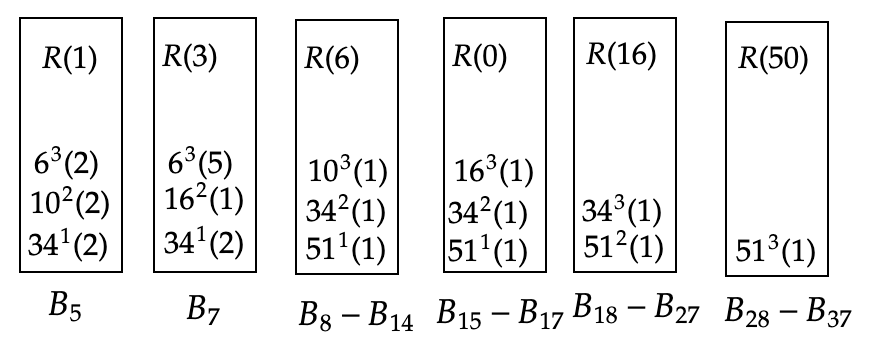}
        \caption{$FFk$ packing for $k=3$. Note that bins from $FFk(D_2)$ packing in which no item from future instances can be packed are not shown. }
        \label{ffk_lb_k3}
    \end{subfigure}

    \begin{subfigure}[b]{0.4\textheight}
        \includegraphics[height=0.9\textheight, width=0.9\textwidth, keepaspectratio]{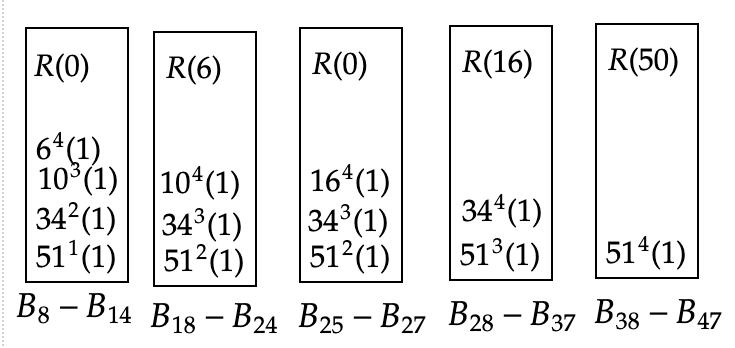}
        \caption{$FFk$ packing for $k=4$. Note that bins from $FFk(D_3)$ packing in which no item from future instances can be packed are not shown.}
        \label{ffk_lb_k4}
    \end{subfigure}

    \begin{subfigure}[b]{0.4\textheight}
        \centering
        \includegraphics[height=0.9\textheight, width=0.9\textwidth, keepaspectratio]{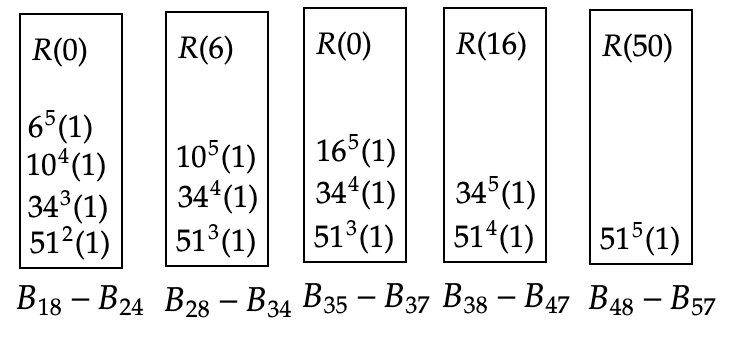}
        \caption{$FFk$ packing for $k=5$. Note that bins from $FFk(D_4)$ packing in which no item from future instances can be packed are not shown.}
        \label{ffk_lb_k5}
    \end{subfigure}
    \caption{ \Cref{ffk_lb_k1}-\Cref{ffk_lb_k5} show the $FFk$ packing for $D$ for different values of $k$. For the lack of space we are not showing the bins in the packing of $D_{k-1}$ in which no item from future instances can be packed. Note the pattern in packing in \Cref{ffk_lb_k4} and \Cref{ffk_lb_k5}.}
    \label{ffk lower bound example}
\end{figure}

		Note that $OPT$ for $k=1$ can result the following bin-packing 
		
		$O_1 = O_2 = O_3 = [51,34,16] $,
		
		$O_4 = O_5 = O_6 = O_7 = O_8 =O_9 = O_{10} = [51,34,10,6]$.
		
		It is easy to observe that optimal number of bins for packing $D_k$ is $OPT(D_k)=k \cdot OPT(D)$.
		
		In {\Cref{ffk lower bound example}} we have shown the packing of $D_1, D_2, D_3, D_4$, and $D_5$. Each box in the figure represents a bin. Each bin represents the items packed into the bin along with their count and the instance from which they belong. Each bin also represents the space remaining in the bin. For example, in {\Cref{ffk_lb_k1}}, $6^1(7)$, inside bin $B_1$ represents that there are 7 items of size 6 from first instance of $D$. $R(9)$ in bin $B_1$ denotes the remaining space in that bin.
		
		{\Cref{ffk_lb_k1}} and {\Cref{ffk_lb_k2}} show the $FFk$ packing of the items for instance $D_1$ and $D_2$ respectively. 
		Note that there are some bins in which the items from future bins cannot be packed (either there is no space or due to the violation of the $k$BP constraint). For the lack of space we do not show such bins. For example, 
		in {\Cref{ffk_lb_k3}} we have the shown the $FFk$ packing of the items for instance $D_3$. $FFk$ packs the items in the third instance of $D$ after packing the items in the first two instances. However, after packing the second instance, bins $B_1, B_2, B_3, B_4$ cannot pack any of the items from the third instance (in fact any of the items from all the future instances). Therefore, we do show these bins in {\Cref{ffk_lb_k3}}. 
		Note that there is a pattern in the $FFk$ packing of $D_4$ and $D_5$. It is easy to see that this pattern repeats itself in the packing of $D_6, D_7, \ldots, D_k$. Also note that while packing the items in $D_k$ we add 10 new bins in addition to the bins in the $FFk$ packing of $D_{k-1}$. Therefore, 
		$\frac{FFk(D_k)}{OPT(D_k)} = \frac{17 + 10(k-1)}{10k} \leq 1.35$ for $k>1$. One can observe that as $k$ increases, the approximation ratio continues to decrease.
	\end{proof}
	
	\begin{lemma} \label{lower bound for ffk}
		The lower bound of $FFk(D_k)$ for $k=2$ is 1.375.
	\end{lemma}
	\begin{proof}
		Consider the input $D= \{371, 659, 113, 47, 485, 3, 228, 419,468,581,626\}$ and $S=1000$.
		For $k=1$, $FFk$ will result in the following packing: \\
		$B_1 = [371^1, 113^1, 47^1, 3^1, 228^1], sum=762$, \\
		$B_2 = [659^1], sum=659 $, \\
		$B_3 = [485^1, 419^1], sum = 904$, \\
		$B_4 = [468^1], sum = 468$, \\
		$B_5= [581^1], sum = 581$, \\
		$B_6 = [626^1], sum = 626$ \\
		We use the word ``SAME" for the content of a bin if a bin cannot pack any of the items from the future instances of $D$.
		This happens because every item in a future instance is either too large to fit into that bin, or has the same type as an item already packed into it. 
		Note that bin $B_1$ is such bin. 
		For $k=2$, $FFk$ will result in the following packing: \\
		$B_1 = \text{SAME}$,\\
		$B_2 = [659^1, 113^2, 47^2, 3^2], sum=822$, \\
		$B_3 = [485^1, 419^1], sum=904$, \\
		$B_4 = [468^1, 371^2], sum=839$, \\
		$B_5= [581^1, 228^2], sum=809$, \\
		$B_6 = [626^1], sum=626$, \\
		$B_7 = [659^2], sum=659$, \\
		$B_8 = [485^2, 419^2], sum=904$, \\
		$B_9 = [468^2], sum=468$, \\
		$B_{10} = [581^2], sum=581$, \\
		$B_{11} = [626^2], sum=626$, \\
		
		The items in $D$ can be optimally  packed in bins $[626,371,3], [659,228,113], [419,581], [468,485,47]$. The size of each of these bins is exactly $1000$. Therefore, the optimal packing of $D_k$ requires $k \cdot OPT(D) = 4k$ bins. 
		Therefore, $\frac{FFk(D_k)}{OPT(D_k} = \frac{11}{8} = 1.375$ for $k = 2$. 
	\end{proof}
	
	\begin{remarkrep}
		The approximation ratio for the example instance given in \Cref{lower bound for ffk} continues to decrease as $k$ increases.
	\end{remarkrep}
	\begin{proof}
		The packing for $k=2$ is given in \Cref{lower bound for ffk}. We will continue from this packing.
		Note that bin $B_2$ cannot pack any of the items from the future instances of $D$.
		
		For $k=3$, $FFk$ will result in the following packing: \\
		$B_1 = \text{SAME}, B_2 = \text{SAME}$\\
		$B_3 = [485^1, 419^1, 47^3, 3^3], sum=954$, \\
		$B_4 = [468^1, 371^2, 113^3], sum=952$, \\
		$B_5= [581^1, 228^2], sum=809$, \\
		$B_6 = [626^1, 371^3], sum=997$, \\
		$B_7 = [659^2, 228^3], sum=887$, \\
		$B_8 = [485^2, 419^2], sum=904$, \\
		$B_9 = [468^2, 485^3], sum=953$, \\
		$B_{10} = [581^2, 419^3], sum=1000$, \\
		$B_{11} = [626^2], sum=626$, \\
		$B_{12} = [659^3], sum=659$, \\
		$B_{13} = [468^3], sum=468$, \\
		$B_{14} = [581^3], sum=581$, \\
		$B_{15} = [626^3], sum=626$ \\
		Note the bin $B_3$ and $B_{10}$ cannot pack any of the items from the future instances of $D$.
		
		For $k=4$, $FFk$ will result in the following packing: \\
		$B_1 = \text{SAME}, B_2 = \text{SAME}, B_3 = \text{SAME}, B_{10} = \text{SAME}$\\
		$B_4 = [468^1, 371^2, 113^3, 47^4], sum=999$, \\
		$B_5= [581^1, 228^2, 113^4, 3^4], sum=925$, \\
		$B_6 = [626^1, 371^3], sum=997$, \\
		$B_7 = [659^2, 228^3], sum=887$, \\
		$B_8 = [485^2, 419^2], sum=904$, \\
		$B_9 = [468^2, 485^3], sum=953$, \\
		$B_{11} = [626^2, 371^4], sum=997$, \\
		$B_{12} = [659^3, 228^4], sum=887$, \\
		$B_{13} = [468^3, 485^4], sum=953$, \\
		$B_{14} = [581^3, 419^4], sum=1000$, \\
		$B_{15} = [626^3], sum=626$, \\
		$B_{16} = [659^4], sum=659$, \\
		$B_{17} = [468^4], sum=468$, \\
		$B_{18} = [581^4], sum=581$, \\
		$B_{19} = [626^4], sum=626$ \\
		Note that bin $B_4$ and $B_{14}$ cannot pack any of the items from the future instances of $D$.
		
		For $k=5$, $FFk$ will result in the following packing: \\
		$B_1 = \text{SAME}, B_2 = \text{SAME}, B_3 = \text{SAME}, B_4 = \text{SAME}, B_{10} = \text{SAME}, B_{14}=\text{SAME}$\\
		$B_5= [581^1, 228^2, 113^4, 3^4, 47^5], sum=972$, \\
		$B_6 = [626^1, 371^3, 3^5], sum=1000$, \\
		$B_7 = [659^2, 228^3, 113^5], sum=1000$, \\
		$B_8 = [485^2, 419^2], sum=904$, \\
		$B_9 = [468^2, 485^3], sum=953$, \\
		$B_{11} = [626^2, 371^4], sum=997$, \\
		$B_{12} = [659^3, 228^4], sum=887$, \\
		$B_{13} = [468^3, 485^4], sum=953$, \\
		$B_{15} = [626^3, 371^5], sum=997$, \\
		$B_{16} = [659^4, 228^5], sum=887$, \\
		$B_{17} = [468^4, 485^5], sum=953$, \\
		$B_{18} = [581^4, 419^5], sum=1000$, \\
		$B_{19} = [626^4], sum=626$, \\
		$B_{20} = [659^5], sum=659$, \\
		$B_{21} = [468^5], sum=468$, \\
		$B_{22} = [581^5], sum=581$, \\
		$B_{23} = [626^5], sum=626$ \\
		Note that bin $B_5, B_6, B_7 B_{18}$ cannot pack any of the items from the future instances of $D$.
		
		For $k=6$, $FFk$ will result in the following packing: \\
		$B_1 = \text{SAME}, B_2 = \text{SAME}, B_3 = \text{SAME}, B_4 = \text{SAME}, B_{5} = \text{SAME}, B_{6} = \text{SAME}, B_{7} = \text{SAME}, B_{10} = \text{SAME}, B_{14}=\text{SAME}, B_{18} = \text{SAME}$\\
		$B_8 = [485^2, 419^2, 47^6, 3^6], sum=954$, \\
		$B_9 = [468^2, 485^3], sum=953$, \\
		$B_{11} = [626^2, 371^4], sum=997$, \\
		$B_{12} = [659^3, 228^4, 113^6], sum=1000$, \\
		$B_{13} = [468^3, 485^4], sum=953$, \\
		$B_{15} = [626^3, 371^5], sum=997$, \\
		$B_{16} = [659^4, 228^5], sum=887$, \\
		$B_{17} = [468^4, 485^5], sum=953$, \\
		$B_{19} = [626^4, 371^6], sum=997$, \\
		$B_{20} = [659^5, 228^6], sum=887$, \\
		$B_{21} = [468^5, 485^6], sum=953$, \\
		$B_{22} = [581^5, 419^6], sum=1000$, \\
		$B_{23} = [626^5], sum=626$, \\
		$B_{24} = [659^6], sum=659$, \\
		$B_{25} = [468^6], sum=468$, \\
		$B_{26} = [581^6], sum=581$, \\
		$B_{27} = [626^6], sum=626$ \\
		Note that bins $B_8, B_{12}, B_{22}$ cannot pack any of the items from the future instances of $D$.
		
		For $k=7$, $FFk$ will result in the following packing: \\
		$B_1 = \text{SAME}, B_2 = \text{SAME}, B_3 = \text{SAME}, B_4 = \text{SAME}, B_{5} = \text{SAME}, B_{6} = \text{SAME}, B_{7} = \text{SAME}, B_{8} = \text{SAME}, B_{10} = \text{SAME}, B_{12} = \text{SAME}, B_{14}=\text{SAME}, B_{18} = \text{SAME}, B_{22} = \text{SAME},$\\
		$B_9 = [468^2, 485^3, 47^7], sum=1000$, \\
		$B_{11} = [626^2, 371^4, 3^7], sum=1000$, \\
		$B_{13} = [468^3, 485^4], sum=953$, \\
		$B_{15} = [626^3, 371^5], sum=997$, \\
		$B_{16} = [659^4, 228^5, 113^7], sum=1000$, \\
		$B_{17} = [468^4, 485^5], sum=953$, \\
		$B_{19} = [626^4, 371^6], sum=997$, \\
		$B_{20} = [659^5, 228^6], sum=887$, \\
		$B_{21} = [468^5, 485^6], sum=953$, \\
		$B_{23} = [626^5, 371^7], sum=997$, \\
		$B_{24} = [659^6, 228^7], sum=887$, \\
		$B_{25} = [468^6, 485^7], sum=953$, \\
		$B_{26} = [581^6, 419^7], sum=1000$, \\
		$B_{27} = [626^6], sum=626$, \\
		$B_{28} = [659^7], sum=659$, \\
		$B_{29} = [468^7], sum=468$, \\
		$B_{30} = [581^7], sum=581$, \\
		$B_{31} = [626^7], sum=626$ \\
		Note that bins $B_9, B_{11}, B_{16}, B_{26}$ cannot pack any of the items from the future instances of $D$.
		
		For $k=8$, $FFk$ will result in the following packing: \\
		$B_1 = \text{SAME}, B_2 = \text{SAME}, B_3 = \text{SAME}, B_4 = \text{SAME}, B_{5} = \text{SAME}, B_{6} = \text{SAME}, B_{7} = \text{SAME}, B_{8} = \text{SAME}, B_{9} = \text{SAME}, B_{10} = \text{SAME}, B_{11} = \text{SAME}, B_{12} = \text{SAME}, B_{14}=\text{SAME}, B_{16} = \text{SAME}, B_{18} = \text{SAME}, B_{22} = \text{SAME}, B_{26} = \text{SAME},$\\
		$B_{13} = [468^3, 485^4, 47^8], sum=1000$, \\
		$B_{15} = [626^3, 371^5, 3^8], sum=1000$, \\
		$B_{17} = [468^4, 485^5], sum=953$, \\
		$B_{19} = [626^4, 371^6], sum=997$, \\
		$B_{20} = [659^5, 228^6, 113^8], sum=1000$, \\
		$B_{21} = [468^5, 485^6], sum=953$, \\
		$B_{23} = [626^5, 371^7], sum=997$, \\
		$B_{24} = [659^6, 228^7], sum=887$, \\
		$B_{25} = [468^6, 485^7], sum=953$, \\
		$B_{27} = [626^6, 371^8], sum=997$, \\
		$B_{28} = [659^7, 228^8], sum=887$, \\
		$B_{29} = [468^7, 485^8], sum=953$, \\
		$B_{30} = [581^7, 419^8], sum=1000$, \\
		$B_{31} = [626^7], sum=626$, \\
		$B_{32} = [659^8], sum=659$, \\
		$B_{33} = [468^8], sum=468$, \\
		$B_{34} = [581^8], sum=581$, \\
		$B_{35} = [626^8], sum=626$ \\
		Note that bins $B_{13}, B_{15}, B_{20}, B_{30}$ cannot pack any of the items from the future instances of $D$.
		
		At this point, we note a repeating pattern:
		
		\begin{itemize}
			\item Item $371^7, 371^8$ has been packed in bins $B_{23}, B_{27}$ respectively. These bins have the same content $[626, 371]$ (without superscript). For $k=9$ it will pack $371^9$ in bin $B_{31}$, and afterwards this bin will also have the same content $[626, 371]$. 
			\item 
			Items $659^7, 659^8$ have been packed in bins $B_{28}, B_{32}$ respectively.
			Item $659^9$ cannot be packed into any of the previous bins and hence will require a new bin $B_{36}$ to pack. 
			\item 
			Items $113^7, 113^8$ have been packed in bins $B_{16}, B_{20}$ respectively. These bins have the same content $[659, 228, 113]$ (without superscript). 
			For $k=9$ it will pack $113^9$ in bin $B_{24}$ and after packing this bin will also have the same content $[659, 228, 113]$.
			\item Items $47^7, 47^8$ have been packed in bins $B_{9}, B_{13}$ respectively. These bins have the same content $[468, 485, 47]$ (without superscript). For $k=9$ it will pack $47^9$ in bin $B_{17}$ and after packing this bin will also have the same content $[468, 485, 47]$.
			\item 
			Items $485^7, 485^8$ have been packed in bins $B_{25}, B_{29}$ respectively. These bins have the same content $[468, 485]$ (without superscript). For $k=9$ it will pack $485^9$ in bin $B_{33}$ and afterwards this bin will also have the same content $[468, 485]$.
			\item
			Items $3^7, 3^8$ have been packed in bins $B_{11}, B_{15}$ respectively. These bins have the same content $[626, 371, 3]$ (without superscript). For $k=9$ it will pack $3^9$ in bin $B_{19}$ and after packing this bin will also have the same content $[626, 371, 3]$.
			\item 
			Items $228^7, 228^8$ have been packed in bins $B_{24}, B_{28}$ respectively. These bins have the same content $[659, 228]$ (without superscript). For $k=9$ it will pack $228^9$ in bin $B_{32}$ and after packing this bin will also have the same content $[659, 228]$.
			\item 
			Items $419^7, 419^8$ have been packed in bins $B_{26}, B_{30}$ respectively. These bins have the same content $[581, 419]$ (without superscript). For $k=9$ it will pack  $419^9$ in bin $B_{34}$ and afterwards this bin will also have the same content $[581, 419]$.
			\item 
			Items $468^9, 581^9, 626^9$ cannot be packed into any of the previous bins and hence will require new bins $B_{37}, B_{38}, B_{39}$ respectively. 
			Note that for $k=8$ these items were packed in new bins $B_{29}, B_{30}, B_{31}$ and for $k=7$ these items were packed in new bins $B_{25}, B_{26}, B_{27}$.  
		\end{itemize}
		One can observe that the content of the bins (except the $\text{SAME}$ bins) remains same for $k=7,8,9$. This pattern continues for further values of $k$. After $k \geq 2$, packing of $D_k$ requires 4 new bins than the packing of $D_{k-1}$. 
		Therefore, $bins(D_k) = 11 + 4\cdot(k-2) = 4\cdot k + 3$.
		
		The items in $D$ can be optimally  packed in bins $[626,371,3], [659,228,113], [419,581], [468,485,47]$. The size of each of these bins is exactly $1000$. Therefore, the optimal packing of $D_k$ requires $k \cdot OPT(D) = 4k$ bins. 
		Therefore, $\frac{FFk(D_k)}{OPT(D_k} = \frac{11 + 4 \cdot (k-2)}{4 \cdot k} = 1 + \frac{3}{4k} \leq 1.375$ for $k \geq 2$. One can observe that as $k$ increases, the approximation ratio continues to decrease.
        \qed
	\end{proof}
\end{toappendix}

\subsection{\texorpdfstring{$FFDk$}{FFDk}} \label{section:approx-algo:subsection:approx_algo:ffdk}
The $k$-time version of the First-Fit Decreasing bin-packing algorithm first sorts $D$ in non-increasing order. Then it  constructs $D_k$ using $k$ consecutive copies of the sorted $D$, and then implements $FFk$ on  $D_k$.
In contrast to $FFk$, we could not prove an upper bound for $FFDk$ that is better than the upper bound for $FFD$; we only have a lower bound.

\begin{toappendix}
    \section{$FFDk$}
    \label{appendix: FFDk}
\end{toappendix}

\begin{lemmarep} \label{ffdk:lemma1}
	$FFDk(D_k) \geq \frac{7}{6} \cdot OPT(D_k) + 1$.
\end{lemmarep}
\begin{inlineproof}
	We use the following example from page $2$ of \cite{dosa_tight_2007}. Let $\delta$ be a sufficiently small positive number, and let $S=1$. Let $D = \{\frac{1}{2}+ \delta, \frac{1}{2}+ \delta, \frac{1}{2}+ \delta, \frac{1}{2}+ \delta, \frac{1}{4}+ 2\delta, \frac{1}{4}+ 2\delta, \frac{1}{4}+ 2\delta, \frac{1}{4}+ 2\delta, \frac{1}{4}+ \delta, \frac{1}{4}+ \delta, \frac{1}{4}+ \delta, \frac{1}{4}+ \delta, \frac{1}{4} - 2\delta, \frac{1}{4} - 2\delta, \frac{1}{4} - 2\delta, \frac{1}{4} - 2\delta, \frac{1}{4} - 2\delta, \frac{1}{4} - 2\delta, \frac{1}{4} - 2\delta, \frac{1}{4} - 2\delta\}$. 
 
    An optimal packing for $D$ contains 4 bins of type $\{ \frac{1}{2}+\delta, \frac{1}{4}+\delta, \frac{1}{4} - 2\delta \}$ and 2 bins of type $\{ \frac{1}{4}+ 2\delta, \frac{1}{4}+ 2\delta, \frac{1}{4} - 2\delta, \frac{1}{4} - 2\delta \}$. As all bin sizes are exactly $1$, this pattern is clearly optimal for any $k$. Therefore, for all $k\geq 1$, $OPT(D_k) = 6k$. 
	
	On applying $FFDk$ on $D_k$ the resulting number of bins are $8 + 7(k-1)$. 
    The detailed $FFDk$ packing of the items in $D_k$ is given in \Cref{appendix: FFDk}.
    This gives us a lower bound of $\frac{7}{6} \cdot OPT(D_k) + 1$.
 \qed
\end{inlineproof}
\begin{toappendix}
    Here, we give in detail the $FFDk$ packing of the items in $D_k$. Note that $D$ is $\{\frac{1}{2}+ \delta, \frac{1}{2}+ \delta, \frac{1}{2}+ \delta, \frac{1}{2}+ \delta, \frac{1}{4}+ 2\delta, \frac{1}{4}+ 2\delta, \frac{1}{4}+ 2\delta, \frac{1}{4}+ 2\delta, \frac{1}{4}+ \delta, \frac{1}{4}+ \delta, \frac{1}{4}+ \delta, \frac{1}{4}+ \delta, \frac{1}{4} - 2\delta, \frac{1}{4} - 2\delta, \frac{1}{4} - 2\delta, \frac{1}{4} - 2\delta, \frac{1}{4} - 2\delta, \frac{1}{4} - 2\delta, \frac{1}{4} - 2\delta, \frac{1}{4} - 2\delta\}$.
\end{toappendix}
\begin{proof}
    Note that all items in $D$ are in non-increasing order. We use the following notation: $x \pm c \cdot \delta^i$, where $x \in \{1/2, 1/4\}$ and $c \in \{1, 2\}$, denotes that the item belongs to the instance $i$. 
    Along with the bins in packing the items in $D$ we also show the sum of all the items in the bin along with an indicator to whether the bin can contain items from the future instances of $D$ or not. We use $\times$ to denote that the bin can not pack items from the future instances of $D$, otherwise we use symbol $\checkmark$. 
    For $k=1$, $FFDk$ packing for the items in $D$ is:\newline
    $B_1 = \{1/2 + \delta^1, 1/4 + 2\delta^1\}, V(B_1) = 3/4 + 3 \delta, \times \\ 
    B_2 = \{1/2 + \delta^1, 1/4 + 2\delta^1\}, V(B_2) = 3/4 + 3 \delta, \times \\
    B_3 = \{1/2 + \delta^1, 1/4 + 2\delta^1\}, V(B_3) = 3/4 + 3 \delta, \times\\
    B_4 = \{1/2 + \delta^1, 1/4 + 2\delta^1\}, V(B_4) = 3/4 + 3 \delta, \times\\
    B_5 = \{1/4 + \delta^1, 1/4 + \delta^1, 1/4 + \delta^1\}, V(B_5) = 3/4 + 3 \delta, \times\\
    B_6 = \{1/4 + \delta^1, 1/4 - 2\delta^1, 1/4 - 2\delta^1, 1/4 - 2\delta^1\}, V(B_1) = 1 - 5 \delta, \times\\
    B_7 = \{1/4 - 2\delta^1, 1/4 - 2\delta^1, 1/4 - 2\delta^1, 1/4 - 2\delta^1\}, V(B_1) = 1 - 8 \delta, \times\\
    B_8 = \{1/4 - 2\delta^1\}, V(B_1) = 1/4 - 2 \delta, \checkmark
    $\\
    Note that the bins $B_1 - B_7$ can not pack any of the items from the future instances of $D$. Therefore, for $k=2$, we show only the bins from $B_8$ onward. Therefore, $FFDk$ packing for the items in $D_2$ is as follows: \newline
    $
    B_8 = \{1/4 - 2\delta^1, 1/2 + \delta^2, 1/4 + \delta^2\}, V(B_8) = 1, \times\\
    B_9 = \{1/2 + \delta^2, 1/4 + 2\delta^2\}, V(B_9) = 3/4 + 3\delta, \times\\ 
    B_{10} = \{1/2 + \delta^2, 1/4 + 2\delta^2\}, V(B_{10}) = 3/4 + 3\delta, \times\\
    B_{11} = \{1/2 + \delta^2, 1/4 + 2\delta^2\}, V(B_{11}) = 3/4 + 3\delta, \times\\
    B_{12} = \{1/4 + 2\delta^2, 1/4 + \delta^2, 1/4 + \delta^2\}, V(B_{12}) = 3/4 + 4\delta, \times\\
    B_{13} = \{1/4 + \delta^2, 1/4 - 2\delta^2, 1/4 - 2\delta^2, 1/4 - 2\delta^2\}, V(B_{13}) = 1 - 5\delta, \times\\
    B_{14} = \{1/4 - 2\delta^2, 1/4 - 2\delta^2, 1/4 - 2\delta^2, 1/4 - 2\delta^2\}, V(B_{14}) = 1 - 8\delta, \times\\
    B_{15} = \{1/4 - 2\delta^2\}, V(B_{15}) = 1/4 - 2\delta, \checkmark
    $\\
    Note that the bins from $B_1$ to $B_{14}$ can not pack any of the items from the future instances of $D$, and the bin $B_{15}$ is same as the bin $B_8$ when $k=1$. Therefore, it is clear that the packing pattern will repeat for further values of $k$, and from above it is clearly evident that it will require only $7$ more bins.
    \qed
\end{proof}

We discuss the challenges in extending the existing proof for the $FFD$ in \Cref{appendix: FFDk}.
\begin{toappendix}
    \paragraph{Challenges in extending existing proof:} Existing proofs for the $FFD$ are based on the assumption that the last bin contains a single item, and no other item(s) are packed after that, i.e. the smallest item is the only item in the last bin of $FFD$ packing. We cannot say the same in the case of $k$BP when $k>1$. For example let $D= \{ 103,102,101\}$ and ${S}=205$. Then $FFDk$ packing when $k=1$ is $\{103,102\}, \{101\}$  whereas for $k=2$ the packing is $\{103, 102\}, \{101,103\}, \{102,101\}$. Hence, the existing proofs for $FFD$ cannot be extended to $FFDk$. 
\end{toappendix}
Based on the simulation of $FFDk$ on different datasets, we conjecture that the upper bound for $FFDk$ is $\frac{11}{9} \cdot OPT(D_k) + \frac{6}{9}$. Experimental results supporting this conjecture are provided in  \Cref{section:results}.

\subsection{\texorpdfstring{$NFk$}{NFk}}
\label{section:approx-algo:subsection:NFk}

\begin{toappendix}
    \section{$NFk$}
    \label{appendix: NFk}
\end{toappendix}

Given the input $D_k$, the algorithm $NFk$ works as follows: like ${NF}$, ${NFk}$ always keeps a single bin open to pack items. If the current item does not pack into the currently open bin then ${NFk}$ closes the current bin and opens a new bin to pack the item.

\begin{theoremrep} \label{NFk:theorem:asymptotic-ratio-2}
	For every input $D_k$ and $k \geq 1$,  the asymptotic ratio of ${NFk}(D_k)$ is 2.
\end{theoremrep}
\begin{proof}
	The idea behind the proof is due to \cite{Johnson1973}. 
		
We can assume that $V(D) > S$, otherwise there is a trivial solution with $k$ bins. 
While processing input $D_k$, $NFk$ holds only one open bin, and it cannot contain a copy of each item of $D$. In fact, the open bin always contains a part of some instance of $D$, and possibly a part of the next instance of $D$, with no overlap. Therefore, 
if the current item $x$ is not packed into the current open bin, the only reason is that $x$ does not fit, as there is no  previous copy of $x$ in the current bin (all previous copies, if any, are in already-closed bins).
	
	Let the number of bins in the ${NFk}$ packing of $D_k$ be ${NFk}(D_k)$. Let these bins be ordered in the sequence in which they are opened. Then, for any two bins $B_{i-1}$ and $B_{i}$ where $i\geq 2$, 
	\[
	V(B_{i-1}) + V(B_i) > S
	\]
	Therefore,
	\[
	V(D_k)=V(B_1) + \ldots + 	V(B_{{NFk}(D_k)}) > \left\lfloor \frac{{NFk}(D_k)}{2} \right\rfloor \cdot S
	\]
	
	Since $OPT(D_k) \geq V(D_k)/S$,
	\[
	{NFk}(D_k) \leq 2 \cdot OPT(D_k) + 1 
	\]
	For the lower bound we consider the example given in \cite{Zheng_Luo_Zhang_2015}. Let the bin capacity be $1$. Let $D$ be an input instance with $2 y$ items (for some large integer $y$) of sizes $1/2, \epsilon, 1/2, \epsilon, \ldots$ ($y$ pairs overall). Then ${NFk}$ will pack $D_k$ into $k\cdot y$ bins, whereas an optimal packing of $D_k$ consists of $k \cdot (y/2+1)$ bins. This gives an asymptotic ratio of $2$.
 \qed
\end{proof}
We give a detailed proof of the above Theorem in \Cref{appendix: NFk}.

\section{Polynomial-time Approximation Schemes} \label{section:ptas}
\subsection{Some general concepts and techniques} \label{section:EAA:subsection:general}

\begin{toappendix}
    \section{Polynomial-time Approximation Schemes}
    \label{appendix: Polynomial-time Approximation Schemes}
\end{toappendix}
The basic idea behind generalizing Fernandez de la Vega-Lueker and all the Karmarkar-Karp algorithms to solve $k$BP is similar. It consists of three steps: 
1. Keeping aside the small items, 2. Packing the remaining large items, and 3. Packing the small items in the bins that we get from step 2 (opening new bins if necessary) to get a solution to the original problem.

In step 3, the main difference from previous work is that, in $k$BP, we cannot pack two copies of the same small item into the same bin, so we may have to open a new bin even though there is still remaining room in some bins. The following lemma analyzes the approximation ratio of this step.
\begin{lemma} \label{general:lemma1}
	Let $D_k$ be an instance of the $k$BP problem, and $0 < \epsilon \leq 1/2$.
We say that the item is \emph{large}, if its size is bigger than $\epsilon \cdot S$ and \emph{small} otherwise. 
Assume that the large items are
packed into $L$ bins. Consider an algorithm which starts adding the small items into the $L$ bins respecting the constraint of $k$BP, but whenever required, the algorithm opens a new bin. Then the number of bins required for the algorithm to pack the items in $D_k$ is at most $\max \{ L, (1+ 2 \cdot \epsilon) \cdot OPT(D_k) + k \}$.
\end{lemma}

\begin{proof}
	 Let $I$ be the set of all small items in $D$ 
 and $I_k$ be the $k$ copies of $I$. Let $A(D_k)$ be the number of bins in the packing of $D_k$. If adding small items do not require a new bin, then $A(D_k) = L$. Otherwise, since  the first item in the last bin cannot be packed to $A(D_k)-k$ previous bins, each of these bins has less than $(1 - \epsilon) \cdot S$ free size. Thus 
	\begin{align} \label{general:lemma1:equation1}
		\sum {D_k [i]} &\geq (S - \epsilon \cdot S)(A(D_k)-k) \nonumber \\
		\sum {D_k [i]} &\leq OPT(D_k) \cdot S \nonumber \\
		A(D_k) &\leq \frac{1}{1 - \epsilon} \cdot OPT(D_k) + k \nonumber \\
		A(D_k) &\leq (1 + 2 \cdot \epsilon) \cdot  OPT(D_k) + k
	\end{align}
	Therefore,
	\begin{equation} \label{general:lemma1:equation2} A(D_k) \leq \max \{ L,  (1 + 2 \cdot \epsilon) \cdot OPT(D_k) + k \} 
	\end{equation}
\qed
\end{proof}


Step 2 is done using a linear program based on \emph{configurations}.

\begin{definition}
\label{definition; configuration}
	A \emph{configuration} (or a \emph{bin type}) is a collection of item sizes which sums to, at most, the bin capacity $S$.
\end{definition}
For example {\cite{enwiki:1139054649}}: suppose there are $7$ items of size $3$, $6$ items of size $4$, and $S=12$. Then, the possible configurations are $ [3,3,3,3], [3,3,3], [3,3], [3], $ $[4,4,4], [4,4], [4],$ $ [3,3,4],[3,4,4], [3,4]$.



Enumerate all possible configurations by the natural numbers from $1$ to $t$. Let $A=\|a_{ij}\|$ be a $m(D) \times t$ matrix, such that for each natural $i\le m(D)$ and $j\le t$ the entry $a_{ij}$ is the number of items of size $c[i]$ in the configuration $j$. Let $\mathbf{n}$ be a $m(D)$-dimensional vector such that for each natural $i\le m(D)$ its $i$th entry is $n[i]$ (the number of items of size $c[i]$). Let $\mathbf{x}$ be a $t$-dimensional vector such that for each natural $j\le t$ we have that $x[j]$ is the number of bins filled with configuration $j$, and $\textbf{1}$ be a $t$-dimensional vector whose each entry is $1$.
Consider the following linear program
\begin{align*}
&& \min  &\quad \mathbf{1 \cdot x}
\\
(C_1)
&&
\text{such that} &\quad A\mathbf{x} = \mathbf{n}
\\
&& &\mathbf{x} \geq 0 
\end{align*}
When $\mathbf{x}$ is restricted to integer entries ($\mathbf{x}\in \mathbb{Z}^t$),
the solution of this linear program defines a feasible bin-packing.
We denote by $F_1$ the fractional relaxation of the above program, where $\mathbf{x}\in \mathbb{R}^t$.

Recall that in $k$BP, each item of $D$ has to appear in $k$ distinct bins. \emph{One can observe that $k$BP uses the same configurations as in the bin-packing, to ensure that each bin contains at most one copy of each item.}
Therefore, the configuration linear program $C_k$ for $k$BP is as follows, 
where $A$, $\mathbf{n}$, and $\mathbf{x}$ are the same as in $C_1$ above
(for $k=1$ it is the same as in \cite{fernandez_de_la_vega_bin_1981}):
\begin{align}
&&	\min  \quad \mathbf{1 \cdot x} \label{eqn:dlvl:4}\\
(C_k)
&&
	\text{such that} \quad A\mathbf{x} &= k\mathbf{n} \label{eqn:dlvl:5} \\
&&	\mathbf{x} &\geq 0 \label{eqn:dlvl:6}
\end{align}

\begin{lemmarep}
Every integral solution of $C_k$ can be realised as a feasible solution of $k$BP.
\end{lemmarep}

\begin{inlineproof}
Since each bin can contain at most one copy of each item of $D$, for each natural $j\le t$ if the configuration $j$ can be realized then $a_{ij}\le n[i]$ for each natural $i\le m(D)$. We shall call such configurations \emph{feasible}. We can realise every sequence of feasible configurations as a solution of the $k$BP problem
as follows. 
For each natural $i\le m(D)$, let $d_1,\dots, d_{n[i]}$ be the items from $D$ of size $c[i]$. Let the queue $Q_i$ be arranged of $k$ copies of these items, beginning from  the first copies of $d_1,\dots, d_{n[i]}$ in this order, then of the second copies of these items in the same order, and so forth.
To realize the sequence, consider the first configuration, say, $j$, from the sequence, and for each natural $i\le m(D)$ move $a_{ij}$ items of size $c[i]$ from the queue $Q_i$ into the first bin (or just do nothing when $Q_i$ is already empty), then  similarly process the second configuration from the sequence and so forth. Since all configurations are feasible, we never put two copies of the same item in the same bin, so the above procedure constructs a feasible solution to $k$BP.

A detailed example to illustrate the realisation, as a feasible solution to \kbp, of the integral solution of $C_k$ is given in \Cref{appendix: Polynomial-time Approximation Schemes}.
\qed
\end{inlineproof}
\begin{toappendix}
    Here, we give an example using which we illustrate how an integral solution of $C_k$ can be realised as a feasible solution of \kbp. 

Consider the example given after \Cref{definition; configuration}
    In this example, there are $m(D) = 2$ distinct item sizes: $c[1] = 3$ and $c[2]=4$. There are $n[1] = 7$ items of size $c[1]=3$, and $n[2]=6$ items os size $c[2]=4$. 
    For $k=2$, an integral solution to $C_2$ will consist of $2$ copies of configuration $[3,3,3,3]$, $2$ copies of configuration $[3,3,4]$, $2$ copies of configuration $[3,4,4]$, and $2$ copies of configuration $[4,4,4]$. Then, this solution sequence can be realised as follows:\\
    -- $3,3,3,3,3,3,3$ are the items in $D$ of size $c[1] = 3$. Then, the queue $Q_1$ will consist of $2$ copies of the items $3,3,3,3,3,3,3$. First, there are first copies of these items, then there are second copies of the items. Therefore, finally $Q_1$ will be $\{3^1,3^1,3^1,3^1,3^1,3^1,3^1, 3^2,3^2,3^2,3^2,3^2,3^2,3^2\}$. Superscript $i$ denotes the $i$th copy of an item. \\
    -- Similarly, $Q_2$ will be $\{4^1,4^1,4^1,4^1,4^1,4^1, 4^2,4^2,4^2,4^2,4^2,4^2\}$.\\
    -- Now, consider the first configuration $[3,3,3,3]$ in the solution sequence. To realise this configuration we move $4$ items of size $c[1] = 3$ from $Q_1$ into the first bin. Since the next configuration in the solution sequence is same as the first one, we make the same movement of items from $Q_1$ to the second bin.
    After these movements, $Q_1$ will be $\{3^2, 3^2, 3^2, 3^2, 3^2, 3^2\}$, and $Q_2$ will be $\{4^1,4^1,4^1,4^1,4^1,4^1, 4^2,4^2,4^2,4^2,4^2,4^2\}$.\\
    -- Third configuration in the solution sequence is $[3,3,4]$. To realise this configuration we move $2$ items of size $c[1]=3$ from $Q_1$ to the third bin, and $1$ item of size $c[2]=4$ from $Q_2$ to the third bin.
    Since the next configuration in the solution sequence is same as third configuration, we make the same movements of items from $Q_1$ and $Q_2$ to the fourth bin.
    After these movements, $Q_1$ will be $\{3^2, 3^2\}$, and $Q_2$ will be $\{4^1,4^1,4^1,4^1, 4^2,4^2,4^2,4^2,4^2,4^2\}$.\\
    -- Fifth configuration in the solution sequence is $[3,4,4]$. To realise this configuration, we move one item of size $c[1] = 3$ from $Q_1$ to the fifth bin, and two items of size $c[2]=4$ from $Q_2$ to this fifth bin.
    Since the next configuration in the solution sequence is same as fifth configuration, we make the same movements of items from $Q_1$ and $Q_2$ to sixth bin.
    After these movements $Q_1$ will be empty, and $Q_2$ will be $\{ 4^2, 4^2, 4^2, 4^2, 4^2, 4^2\}$.\\
    -- Seventh configuration in the solution sequence is $[4,4,4]$. To realise this configuration we move three items of size $c[2]=4$ from $Q_2$ to seventh bin.
    Since the next configuration is same as seventh configuration, therefore we make the same movements of items from $Q_2$ to eighth bin. 
    After, these movement both $Q_1$ and $Q_2$ will become empty and also there are no more configurations to realise.\\
    Since all the configuration in the solutions sequence are feasible, the above procedure has constructed a feasible solution to \kbp.
    \qed
\end{toappendix}

Let $F_k$ be the fractional bin-packing problem corresponding to $C_k$. Step $2$ involves grouping. Grouping reduces the number of different item sizes, and thus reduces the number of constraints and configurations in the fractional linear program $F_k$.

To solve the configuration linear program 
efficiently, both Fernandez de la Vega-Lueker algorithm and Algorithm 1 of Karmarkar Karp use a \emph{linear grouping} technique. In linear grouping, items are divided into groups (of fixed cardinality, except possibly the last group), and each item size (in each group) increases to the maximum item size in that group. 
See \Cref{appendix:linear grouping} for more detail.

\begin{toappendix}
    \section{Fernandez de la Vega-Lueker algorithm to \kbp}
    \label{appendix: Fernandez de la Vega-Lueker algorithm to kbp}
\end{toappendix}
\begin{toappendix}
	\paragraph{Linear Grouping: } \label{appendix:linear grouping} 
    Let $D$ be some instance of the bin-packing problem and $g > 1$ be some integer parameter. Order the items in $D$ in a non-increasing order. Let $U$ be the instance obtained by making groups of the items in $D$ of cardinality $g$ and then rounding up the items in each group by the maximum item size in that group. Let $U'$ be the group of the $g$ largest items and $U''$ be the instance consisting of groups from the second to the last group in $U$. Then 
$OPT(U'') \leq OPT(D)$ \cite{karmarkar-efficient-1982} and 
$OPT(U') \leq g$, because each item in $U'$ can be packed into a single bin. So
$$OPT(D) \leq OPT(U'' \cup U') \leq OPT(U'') + OPT(U') \leq OPT(U'') + g.$$
It implies the below \Cref{general:lemma2} due to \cite{karmarkar-efficient-1982},
where 
 $LIN(D)$ denotes the value of the fractional bin-packing problem $F_1$ associated with the instance $D$:
\begin{lemma} \label{general:lemma2}
    1. $OPT(U'') \leq OPT(D) \leq OPT(U'') + g$. $\quad$
	2. $LIN(U'') \leq LIN(D) \leq LIN(U'') + g$. $\quad$
	3. $V(U'') \leq V(D) \leq V(U'') + g$.
\end{lemma}
\end{toappendix}

Our extension of the 
Fernandez de la Vega-Lueker and
the Karmarkar-Karp algorithms to $k$BP differs from their original counterparts in mainly two directions. First, in the 
configuration linear program (see the constraint {\Cref{eqn:dlvl:5}} in $C_k$), and hence the obtained solution to this configuration linear program is not necessarily the $k$ times copy of the original solution of BP. Second, in greedily adding the small items, see {\Cref{general:lemma1}}. In extension of Karmarkar-Karp algorithm 1 to $k$BP we have also shown that getting an integer solution from $\mathbf{x}$ by rounding method may require at most $(k-1)/2$ additional bins. 
We discuss extensions to the Fernandez de la Vega-Lueker and Karmarkar-Karp algorithms and their analyses in \Cref{SEC: PTAS: SUBSECTION: dlvl} and \Cref{SEC: PTAS: SUBSECTION: kkalgorithms}, respectively.

The inputs to the extension of the algorithms by Fernandez de la Vega-Lueker and Algorithm 1 and Algorithm 2 of Karmarkar-Karp are an input set of items $D$, a natural number $k$, and an 
approximation parameter $\epsilon \in (0,1/2]$. Algorithm 2 of Karmarkar-Karp, in addition, accepts an integer parameter $g > 0$.


\subsection{Fernandez de la Vega-Lueker algorithm to \texorpdfstring{$k$}{k}BP} \label{SEC: PTAS: SUBSECTION: dlvl}

Fernandez de la Vega and Lueker \cite{fernandez_de_la_vega_bin_1981} published a PTAS which, given an input instance $D$ and $\epsilon \in (0,1/2] $, solves a bin-packing problem 
with, at most, $(1+\epsilon) \cdot OPT(D) + 1$ bins. They devised a method called ``adaptive rounding'' for this algorithm. In this method, the given items are put into groups and rounded to the largest item size in that group. This resulting instance will have fewer different item-sizes. This resulting instance can be solved efficiently using a configuration linear program $C_k$ (see 
\Cref{appendix:linear grouping}
for more detail).

\begin{toappendix}
    \paragraph{Solving the configuration linear program $C_k$:} \label{appendix: solving the configuration lp}
    Configuration linear program $C_k$ can be solved as follows:
    
    \textit{Using exhaustive search:} Since there are $n(D)$  items, a configuration can repeat at most $n(D)$  times. Therefore, $(x_\tau)_{\tau \in T} \in \{0,1, \ldots,n(D)\}^{\lvert T \rvert}$ 
    where $T$ is the set of all possible configurations, $\tau$ is some configuration in $T$, and $x_\tau$ is the number of bins filled with the configuration $j$.
    Each configuration contains at most $1/\epsilon$ items, and there are $m(D)$ different item sizes. Therefore, the number of possible configurations is at most $m(D)^{1/\epsilon}$. Now, we can check if enough slots are available for each size $c[i]$. Finally, output the solution with the minimum number of bins. Running time then is $m(D) \times n(D)^{m(D)^{1/\epsilon}}$, which is polynomial in $n(D)$ when $m(D)$ and $\epsilon$ are fixed.
\end{toappendix}

\paragraph{A high-level description of the extension of Fernandez de la Vega-Lueker algorithm to $k$BP.}
Let $I$ and $J$ be multisets of small and large items in $D$, respectively. After applying linear grouping in $J$, let $U$ be the resulting instance and $C_k$ be the corresponding configuration linear program. An optimal solution to $C_k$ will give us an optimal solution to the corresponding $k$BP instance $U_k$. Ungrouping the items in $U_k$ will give us a solution to $k$BP instance $J_k$. Finally, adding the items in $I_k$, by respecting the constraints of $k$BP, to the solution of $J_k$ (and possibly opening new bins if required) will give us a packing of $D_k$.

\Cref{dlvl:lemma:opt<(1+epsilon)opt} (see \Cref{appendix: Fernandez de la Vega-Lueker algorithm to kbp} for more detail) 
bounds the number of bins in an optimal packing of $U_k$. The extension of the algorithm by Fernandez de la Vega-Lueker to $k$BP is given in 
{\Cref{appendix: Fernandez de la Vega-Lueker algorithm to kbp}}.

\begin{toappendix} 
	
    If the items in $D$ are arranged in a non-decreasing order, then $U''$ will be the instance consisting of groups from first to the second from the last group, and $U'$  be the last group consisting of the $\leq g$ large items. The above same results hold in this case as well.
	
	We extend the algorithm by de la Vega and Lueker to solve $k$BP as follows:
	
	\begin{algorithm}[H]
		\DontPrintSemicolon
		\KwIn{A set $D$ of items, an integer $k$, and $\epsilon \in (0, 1/2]$.}
		\KwOut{A bin-packing of $D_k$.}
		Let $I$ and $J$ be multisets of small (of size $\le\epsilon \cdot S$) and large (of size $>\epsilon \cdot S$) items in $D$, respectively.  \;
		Sort the items in $J$ in a non-decreasing order of their sizes. \;
		Construct an instance $U$ from $J$ by 
		the linear grouping with $g = n(J) \cdot {\epsilon}^2$ and round up the sizes in each group to the maximum size in that group. \;
		Optimally solve the configuration linear program $C_k$ corresponding to the rounded problem $U$. 
		This will give us an optimal solution to the $k$BP instance $U_k$.  \;
		To get a solution for $J_k$, ``ungroup'' the items, that is replace the items in groups with the original items in that group, as in step $3$, prior to rounding up.  \;
		Greedily add small items in $I_k$ by respecting constraint of $k$BP to get a solution for $D_k$. \;
		\caption{Fernandez de la Vega-Lueker Algorithm to $k$BP}
		\label{delaVegaLueker}
	\end{algorithm}	
	The instances $I_k$ and $J_k$ consist of $k$ copies of instances $I$ and $J$, respectively.
	\begin{lemma} \label{dlvl:lemma:opt<(1+epsilon)opt}
		Let  $OPT(U_k)$ denote the optimal number of bins in solving $C_k$. Then, $OPT(U_k) < (1 + \epsilon) \cdot OPT(J_k) $.
	\end{lemma}
	\begin{proof}
		Let $L$ be the problem resulting from rounding down the items in each group but the first by the maximum of the previous group. Then, the instance $L_k$ consists of $k$ copies of instance $L$. The items in the first group are  rounded down to $0$. Since bin-packing  is monotone, 
		\begin{align} \label{equation:dlvlp:1}
			OPT(L_k) \leq OPT(J_k) \leq OPT(U_k)
		\end{align} 
		Note that $L$ and $U$ differ by a group of at most 
		$n(J) \cdot {\epsilon}^2$ items. Therefore, the difference between $L_k$ and $U_k$ is, at most, $k \cdot n(J) \cdot {\epsilon}^2$. Hence,
		\begin{align} \label{equation:dlvlp:2}
			OPT(U_k) - OPT(L_k) \leq k\cdot n(J)\cdot \epsilon^2 \nonumber \\
			OPT(U_k) \leq OPT(J_k) + k\cdot n(J)\cdot \epsilon^2
		\end{align}
		Since all items in $U_k$ 
		(and hence in $J_k$) have a size larger than $\epsilon \cdot S$, the number of items in a bin is at most $1/\epsilon$.  Hence, the minimum number of bins required to pack 
		the items in $J_k$
		are at least
		$k \cdot n(J) \cdot \epsilon$. Hence,
		\begin{equation} \label{equation:dlvlp:3}
			k \cdot n(J) \cdot \epsilon^2 \leq \epsilon \cdot OPT(J_k).
		\end{equation}
		By \Cref{equation:dlvlp:2} and \Cref{equation:dlvlp:3},  
		\begin{equation} \label{dlvlres1}
			OPT(U_k) \leq (1 + \epsilon)OPT(J_k)
		\end{equation}	
        \qed
	\end{proof}
	
	\begin{corollary} \label{dlvl:corollary:bins<(1+epsilon)opt}
		$ A(J_k) \leq (1+\epsilon) \cdot OPT(J_k)$.
	\end{corollary}
	\begin{proof}
		From steps $3$ and $5$ of the algorithm, it is clear that from a packing of $U_k$ we can obtain a packing of $J_k$ with the same number of bins. Therefore, $A(J_k) \leq (1+\epsilon) \cdot OPT(J_k)$.
        \qed
	\end{proof}
\end{toappendix}

\begin{theoremrep} \label{dlvl:theorem:bins<=(1+2epsilon)opt+k}
    Let $A$ be the generalization of the Fernandez de la Vega-Lueker algorithm to solve \kbp. Then, 
    $A(D_k) \leq (1 +  2 \cdot \epsilon ) \cdot OPT(D_k)	 + k$.
\end{theoremrep}
\begin{proof}
	Let $I$ be the set of all items of size at most $\epsilon \cdot S$ in $D$, and $I_k$ be $k$ copies of the instance $I$. In the step $6$, while packing the items in $I_k$ if fewer than $k$ new bins are opened, then from $\Cref{dlvl:corollary:bins<(1+epsilon)opt}$ it is clear that $A(D_k) \leq A(J_k) + k \leq (1 + \epsilon) \cdot OPT(J_k) + k$. Since bin-packing is monotone, $OPT(J_k) \leq OPT(D_k)$. Therefore, $A(D_k) \leq (1 + \epsilon) \cdot OPT(D_k) + k$.
	
	Now assume that at least $k$ new bins are opened at the step $6$ of the algorithm. Then, 
    From \Cref{general:lemma1},
	\begin{align} \label{dlvl:theorem1:equation1}
		A(D_k) &\leq \max \{(1+\epsilon) \cdot OPT(D_k) + k, (1+2 \cdot \epsilon) \cdot OPT(D_k) + k \} \nonumber \\
		&\leq (1+2 \cdot \epsilon) \cdot OPT(D_k) + k
	\end{align}
\qed
\end{proof}

\begin{toappendix}
    \paragraph{Running time of \Cref{delaVegaLueker}: }
    \label{dlvl:algorithmruntimeanalysis} 
    Time $O(n(D)\log{n(D)})$ is sufficient for all the steps except the step 4 in \Cref{delaVegaLueker}. Step 4 takes time $O\left(m(D) \cdot {n(D)} ^ {{m(D)} ^ {1/\epsilon}}\right)$. Therefore, the time taken by the algorithm is $O\left(m(D) \cdot {n(D)} ^ {{m(D)} ^ {1/\epsilon}}\right)$, which is polynomial in $n(D)$ given $m(D)$ and $\epsilon$.    
\end{toappendix}

We give a detailed proof of \Cref{dlvl:theorem:bins<=(1+2epsilon)opt+k}, along with the runtime analysis of the Fernandez de la Vega-Lueker algorithm to $k$BP, in 
\Cref{appendix: Fernandez de la Vega-Lueker algorithm to kbp}.

\subsection{Karmarkar-Karp Algorithms to \texorpdfstring{$k$}{k}BP} \label{SEC: PTAS: SUBSECTION: kkalgorithms}

Karmarkar and Karp \cite{karmarkar-efficient-1982} improved the work done by Fernandez de la Vega and Lueker \cite{fernandez_de_la_vega_bin_1981} mainly in two directions: (1) Solving the linear programming relaxation of $C_1$ 
using a variant of the GLS method \cite{grotschel_ellipsoid_1981} and (2) using a different grouping technique. These improvements led to the development of three algorithms. Their algorithm $3$ is a particular case of the algorithm $2$; we will discuss the generalization of algorithms $1$ and $2$ of Karmarkar-Karp algorithms to solve $k$BP.
Let $LIN(F_k)$ denote the optimal solution to the fractional linear program $F_k$ for $k$BP. We will discuss helpful results relevant to analyzing the generalized version of their algorithms. These results are an extension of the results in \cite{karmarkar-efficient-1982}.
 {\Cref{kkalgorithms: opt-dk<=2V(Dk)/s+k}} bounds from above the number of bins needed to pack the items in an optimal packing of some instance $D_k$. {\Cref{kkalgorithms: V(Dk)<=LIN(Dk)S+S(m(Dk)+k)/2}} concerns  obtaining an integer solution from a basic feasible solution of the fractional linear program.
We discuss \Cref{kkalgorithms: opt-dk<=2V(Dk)/s+k} and \Cref{kkalgorithms: V(Dk)<=LIN(Dk)S+S(m(Dk)+k)/2} in appendix \Cref{kkalgorithms2kBP}.

\begin{toappendix}
	\section{Karmarkar-Karp Algorithms to \kbp} \label{kkalgorithms2kBP}
\end{toappendix}

\begin{toappendix}
	\begin{lemma} 
		\label{kkalgorithms: opt-dk<=2V(Dk)/s+k}
		$OPT(D_k) \leq 2 \cdot V(D_k)/S + k$ 
	\end{lemma}
	\begin{proof}
		When $k=1$ then from \mbox{\cite{karmarkar-efficient-1982}} we know that $OPT(D) \leq \frac{2}{S} \cdot V(D)+ 1$. Therefore $$OPT(D_k) \leq k \cdot OPT(D) \leq \frac{2}{S}\cdot k \cdot V(D)+ k = \frac{2}{S} \cdot V(D_k)+ k.$$
        \qed
	\end{proof}
	
	\begin{lemma} 
		\label{kkalgorithms: V(Dk)<=LIN(Dk)S+S(m(Dk)+k)/2}
		$V(D_k) \leq S \cdot LIN(D_k) \leq S \cdot OPT(D_k) \leq S \cdot LIN(D_k) + \frac{S \cdot (m(D_k) + k)}{2}$
	\end{lemma}
	\begin{proof}
		Let $\mathbf{c}$ and $\mathbf{n}$ be the vectors of dimension $m(D)$ such that for each natural $i\le m(D)$ the $i$th entry of $\mathbf{c}$ is the size $c[i]$ and the $i$th entry of $\mathbf{n}$ is the number $n[i]$ of items of size $c[i]$. 
		For each natural $j\le t$ let the size of the configuration $j$ be the sum of all sizes of the items of $a_j$, which is equal to the $j$th entry of the vector $\mathbf{c^T}A$. For instance, let $m(D)=2$, $n[1]=2$, $n[2]=2$, $c[1]=3$, $c[2]=4$, and $S=10$. Then, the number $t$ of feasible  configurations is $6$, and let the matrix $A$ be 
		$ \begin{bmatrix}
			0 & 0 & 1 & 2 &1 & 2 \\
			1 & 2 & 0 & 0 &1 & 1
		\end{bmatrix} $. Then $\mathbf{c^T}A=\begin{bmatrix}
			4 & 8 & 3 & 6 & 7 & 10
		\end{bmatrix}$. Note that the size of each configuration is at most the bin size $S$. So if $\mathbf{S}$ is a vector of dimension $t$ whose each entry equals $S$ then
		\begin{equation} \label{kk:lemma2:equation1}
			\mathbf{S} \geq \mathbf{c^T}A
		\end{equation}
		Let $\mathbf{x}$ be an optimal basic feasible solution of the fractional linear program $F_k$, which represents the $k$BP problem, for instance, $D_k$. 
		Each element $x[j]$ in $\mathbf{x}$ represents the (fractional) number of bins filled with configuration $j$.
		Then, $\mathbf{S} \cdot \mathbf{x}$ is the maximum possible sum of all items  
		from
		all configurations in the solution vector $\mathbf{x}$.
		Then,
		\begin{align}
			S \cdot LIN(D_k) = S \cdot LIN(F_k) &= \mathbf{S} \cdot \mathbf{x}  \label{kk:lemma2:equation2}\\
			&\geq (\mathbf{c}^T A) \cdot \mathbf{x} && \text{by \Cref{kk:lemma2:equation1}} \label{kk:lemma2:equation3} \\
			&\geq  \mathbf{c}^T(k\mathbf{n}) && \text{by \Cref{eqn:dlvl:5}} 
		\label{kk:lemma2:equation4}
	\end{align}
	
	Since for each natural $i\le m(D)$, the $i$th entry of $\mathbf{c^T}$ is the size $c[i]$ and the $i$th entry of $\mathbf{n}$ is the number of items of size $c[i]$, $\mathbf{c^T n}$ is the sum of sizes of all items of $D$, so by \Cref{kk:lemma2:equation4} we have  $S \cdot LIN(D_k) = k \cdot V(D) = V(D_k)$.
	
	\par Since for any natural $j\le t$ the variable $x_j$ in the problem for $LIN(D_k)$ is the fractional number of occurrences of the configuration $j$, we have $LIN(D_k) \leq OPT(D_k)$, and the equality holds iff the problem for $LIN(D_k)$ has an integer solution.
	
	\par 
	
	It is well-known that the linear program {\Cref{eqn:dlvl:4}} - {\Cref{eqn:dlvl:6}} has 
	an optimal solution $\mathbf{x}$ which is a \emph{basic feasible solution},
	where the number of
	non-zero variables is at most the number $m(D)$ of constraints. 
	For each natural $j\le t$ such that $x_j>0$ we have $x_j = \lfloor x_j \rfloor + r_j$, where $ \lfloor x_j \rfloor$ is the integer part of $x_j$, and $r_j$ is the fractional part of $x_j$.
	For each natural $j \le t$ such that $r_j \in (0,1)$, we take configuration $j$. Let $D'$ be the resulting constructed instance. Clearly, all items in $D'$ will have a single copy.
	Then,
	\begin{equation} \label{kk:lemma2:equation6}
		V(D') \leq S \cdot LIN(D')  = S \cdot \sum_j r_j 
	\end{equation}
	For the residual part, we can construct two packings. 
	First, for each non-zero $r_j$, let's take a bin of configuration $j$. We can then remove extra items.
	Therefore,
	\begin{equation} \label{eqn:kk:lemma2:1} OPT(D') \leq m(D_k) \text{ ( the instances } D \text{ and } D_k \text{ have the same configurations)}\end{equation}
	
	Second, from \Cref{kkalgorithms: opt-dk<=2V(Dk)/s+k} we have	 
	\begin{equation} \label{eqn:kk:lemma2:2} OPT(D') \leq 2V(D')/S + k  \end{equation}
	By \Cref{eqn:kk:lemma2:1} and \Cref{eqn:kk:lemma2:2},$$OPT(D')\le  \min \{m(D_k), 2V(D')/S + k\}\le\frac {m(D_k) +  2V(D')/S + k}2= V(D')/S + \frac{m(D_k)+k}{2}.$$.
	
	Therefore,
	\begin{align} \label{kklemm2e4}
		OPT(D_k) &\leq \text{the principal part} + OPT(D') \nonumber \\
		OPT(D_k) &\leq \text{the principal part} + V(D')/S + \frac{m(D_k)+k}{2} \nonumber \\
		OPT(D_k) &\leq LIN(D_k) + \frac{m(D_k)+k}{2} && \text{by \Cref{kk:lemma2:equation2}}
	\end{align}
    \qed
	\end{proof}
	
	\begin{corollary} \label{kk:lemm2cor1}
	There is an algorithm \textbf{A} that solves the $k$BP by solving the corresponding fractional bin-packing problem $F_k$. It produces at most $\mathbf{1 \cdot x} + \frac{m(D_k)+k}{2}$ bins, where $\mathbf{x}$ is 
	an optimal solution of $F_k$.
	\end{corollary}
	Note that the number of different item sizes in the instance is $m(D_k) = m(D)$.
\end{toappendix}

Before moving further, we would like to mention that if we use some instance (or group) without subscript $k$, we are talking about the instance when $k=1$. 

All Karmarkar-Karp algorithms use a variant of the ellipsoid method to solve the fractional linear program. So, we will talk about adapting this method to $k$BP.

\paragraph{Solving the fractional linear program: } \label{kk:solving-flp}
Solving the fractional linear program \hypertarget{kk:flp}{$F_k$}   involves a variable for each configuration. This results in a large number of variables. The fractional linear program $F_k$ has the following dual $D_F$.
\begin{align} \label{kk:dflp}
	\max  \quad k \cdot \mathbf{n \cdot y} \\
	\text{such that} \quad A^T \mathbf{y} &\leq \mathbf{1} \\
	\mathbf{y} &\geq 0 
\end{align}
The above dual linear program can be solved to any given tolerance $h$ by using a variant of the ellipsoid method that uses an approximate separation oracle \cite{karmarkar-efficient-1982}. The running time of the algorithm is $T(m(D_k),n(D_k)) = O\left(m(D)^8 \cdot  {\ln{m(D)}} \cdot {\ln^2\left(\frac{m(D) \cdot n(D)}{\epsilon \cdot S \cdot h}\right)} + \frac{m(D)^4 \cdot k \cdot n(D) \cdot \ln{m(D)}}{h} \ln{\frac{m(D) \cdot n(D)}{\epsilon \cdot S \cdot h}}\right)$.	We give a description of this variant of the ellipsoid method and its running time in \Cref{kkalgorithms: solving-flp}

\begin{toappendix}
	\subsection{Solving the fractional linear program: } \label{kkalgorithms: solving-flp}
	Solving the fractional linear program \hyperlink{kk:flp}{$F_k$}   involves a variable for each configuration. This results in a large number of variables. The fractional linear program $F_k$ has the following dual $D_F$.
	\begin{align} \label{kk:dflp:appendix}
		\max  \quad k \cdot \mathbf{n \cdot y} \\
		\text{such that} \quad A^T \mathbf{y} &\leq \mathbf{1} \\
		\mathbf{y} &\geq 0 
	\end{align}
	The above dual linear program can be solved to any given tolerance $h$ by using a variant of the ellipsoid method that uses an approximate separation oracle \cite{karmarkar-efficient-1982}.
	This separation oracle accepts as input a vector $\mathbf{y}$ (each element $y[i]$ of $\mathbf{y}$ represents the price of item of size $c[i]$) and determines:
	\begin{itemize}
		\item whether $\mathbf{y}$ is feasible, or
		\item if not, then return a constraint that is violated, i.e. a vector $\mathbf{a}$ such that $\mathbf{a} \cdot \mathbf{y} > 1$.
	\end{itemize}
	Let $\mathbf{c}$ be the $m(D)$-vector of the item sizes. The separation oracle solves the above problem by solving the following knapsack problem:
	\begin{align} \label{kk:dflp:knapsack}
		\max \quad \mathbf{a} \cdot \mathbf{y} \\
		\text{such that} \quad \mathbf{a} \cdot \mathbf{c} &\leq S \\
		\mathbf{a} &\geq 0 
	\end{align}
	Each entry $a[i]$ of $\mathbf{a}$ represents $a[i]$ pieces of size $c[i]$ 
	, and hence $\mathbf{a}$ is an integer vector.   If $\mathbf{y}$ is feasible then the optimal value of the above knapsack problem is at most $1$. Otherwise, $\mathbf{a}$ correspond to a configuration that violates the constraint $\mathbf{a} \cdot \mathbf{y} \leq 1$.
	
	Suppose we want a solution of $D_F$ 
	within a specified tolerance $\delta$. Then, the above knapsack problem can be solved in polynomial time by rounding down each component of $\mathbf{y}$ to the closest multiple of  {$\frac{\delta}{n(D)}$}. The runtime of the above algorithm is  {$O\left(\frac{m(D) \cdot n(D)}{\delta}\right)$} \cite{karmarkar-efficient-1982}.  Alternatively, we can round down each component of $\mathbf{y}$ to the nearest multiple of {$\frac{\delta}{k \cdot n(D)}$}, and then the runtime of the above algorithm is {$O\left(\frac{k \cdot m(D) \cdot n(D)}{\delta}\right)$}.
	
	The modified ellipsoid method uses the above approximate separation oracle as follows. Given the current ellipsoid center $\mathbf{y_c}$, let $\widetilde{\mathbf{y_c}}$ be the rounded down solution corresponding to $\mathbf{y_c}$. The modified ellipsoid method performs either an \textit{feasibility cut} (if $\widetilde{\mathbf{y_c}}$ is infeasible) or an \textit{optimality cut} (if $\widetilde{\mathbf{y_c}}$ is feasible). Since $\mathbf{\widetilde{y_c}}$ is the rounded down solution corresponding to $\mathbf{y_c}$, if $\mathbf{\widetilde{y_c}}$ is infeasible then $\mathbf{y_c}$ is not feasible too. In a feasibility cut, the modified ellipsoid method cuts from the ellipsoid all points such that $\mathbf{a} \cdot \mathbf{y} > 1$. If $\mathbf{\widetilde{y_c}}$ is feasible, then $\mathbf{y_c}$ may or may not be feasible.  
	If the rounding-down is done to a multiple of {$\frac{\delta}{k \cdot n(D)}$}, then, by definition of the rounding, we have 
	$$k \cdot \mathbf{n} \cdot \widetilde{\mathbf{y_c}} \geq k \cdot \mathbf{n} \cdot \mathbf{y_c} - k \cdot \mathbf{n} \cdot \mathbf{1} \cdot \frac{\delta}{k \cdot n} \nonumber  k \cdot \mathbf{n} \cdot \mathbf{y_c} - \delta.$$
	In an optimality cut, the method preserves all the vectors whose value exceeds $\widetilde{\mathbf{y_c}}$ by more than $\delta$ \cite{karmarkar-efficient-1982} i.e. it removes all points that satisfy $k \cdot \mathbf{n} \cdot \mathbf{y} < k \cdot \mathbf{n} \cdot \widetilde{\mathbf{y_c}} + \delta $ .
	
	The number of iterations in the modified ellipsoid algorithm is the same as that of the modified ellipsoid method in \cite{karmarkar-efficient-1982} that is, at most, {$Q = 4 \cdot {m(D)}^2 \cdot \ln{\left(\frac{m(D) \cdot n(D)}{\epsilon \cdot S \cdot \delta}\right)}$} iterations. The total execution time of the modified algorithm is  {$O\left(Q \cdot m(D) \cdot \left(Q \cdot m(D) + \frac{k \cdot n(D)}{\delta}\right)\right)$}. 
	
	Let $t$ be the number of configurations in $k$BP, which is the same as the number of configurations in BP. There are at most $t$ constraints of the form $\mathbf{a} \cdot \mathbf{y} \leq 1$ during the ellipsoid method. It is well known that in any bounded linear program with {$m(D)$}  variables, at most {$m(D)$} constraints are sufficient to determine the optimal solution. 
	We can use the same constraint elimination procedure as in \cite{karmarkar-efficient-1982} to come up with this critical set of {$m(D)$}  constraints. It results in having a reduced dual linear program with {$m(D)$}  variables and {$m(D)$}  constraints. Let us call this reduced dual linear program $R_{D_F}$. By taking the dual of $R_{D_F}$, we get a reduced primal linear program, say $R_F$. This reduced primal linear program will have only {$m(D)$}  variables corresponding to {$m(D)$}  constraints. Let us denote the optimal solutions of the dual linear program $D_F$ and the reduced dual linear program $R_{D_F}$ as $LIN(D_F)$ and $LIN(R_{D_F})$ respectively. Since $R_{D_F}$ is a relaxation of $D_F$ therefore, $LIN(D_F) \leq LIN(R_{D_F})$. From the modified GLS algorithm, we know that $LIN(R_{D_F}) \leq LIN(D_F) + \delta$ if rounding down is done in a multiple of  {$\frac{\delta}{k \cdot n(D)}$}. From the duality theorem of the linear program, both primal and dual have the same optimal solution. Therefore, the optimal solution of the reduced primal linear program, $R_F$, is at most $LIN(D_F) + \delta$ using the rounding down scheme in the modified GLS method.
	
	The constraint elimination procedure will require to execute the modified GLS algorithm to run at most
	{$(m(D)+1)\left(\left\lceil \frac{R}{m(D)+1} \right\rceil + (m(D)+1) \ln\left(\left\lceil \frac{R}{m(D)+1} \right\rceil\right) + 1\right)$} times, where  {$R \leq Q + 2 \cdot m(D)$} \cite{karmarkar-efficient-1982}. In the constraint elimination procedure, in the worst case, one out of every {$m(D)+1$} executions succeeds \cite{karmarkar-efficient-1982}. Therefore, the total accumulated error is  {$\delta \left(\left\lceil \frac{R}{m(D)+1} \right\rceil + (m(D)+1) \ln\left(\left\lceil \frac{R}{m(D)+1} \right\rceil\right) + 1\right)$}. If the tolerance $\delta$ is set to 
	$$\frac{h}{\left\lceil \frac{R}{m(D)+1} \right\rceil + (m(D)+1) \ln\left(\left\lceil \frac{R}{m(D)+1} \right\rceil\right ) + 2}$$ then the optimal solution of the reduced primal linear program is at most $LIN(D_F) + h$ using the rounding down mechanism in the modified GLS method.
	
	The total running time of the algorithm is
	$$O\left(\left(Q  m(D) \left(Q  m(D) + \frac{k \cdot n(D)}{\delta}\right)\right)(m(D)+1)\left(\left\lceil \frac{R}{m(D)+1} \right\rceil + (m(D)+1) \ln\left\lceil \frac{R}{m(D)+1} \right\rceil + 1\right)\right) $$ $$\approx O\left(m(D)^8 \cdot  {\ln{m(D)}} \cdot {\ln^2\left(\frac{m(D) \cdot n(D)}{\epsilon \cdot S \cdot h}\right)} + \frac{m(D)^4 \cdot k \cdot n(D) \cdot \ln{m(D)}}{h} \ln{\frac{m(D) \cdot n(D)}{\epsilon \cdot S \cdot h}}\right).$$ {Since $m(D) = m(D_k)$ and $k \cdot n(D) = n(D_k)$}, in our analysis, we denote this running time as {$T(m(D_k),n(D_k))$}.	 
\end{toappendix}

\paragraph{A high-level description of the extensions of the Karmarkar-Karp algorithms.} We will give a high-level description behind the extension of Karmarkar-Karp algorithms to solve $k$BP. Let $I$ and $J$ be multisets of small and large items in $D$, respectively. Let $U''$ be the instance constructed from $J$ by applying the grouping technique. Construct the configuration linear program $C_k$ for $U''_k$ and solve the corresponding fractional linear program $F_k$. Let $\mathbf{x}$ be the resulting solution. From $\mathbf{x}$, obtain an integral solution for $U''_k$. From this solution, get a solution for $J_k$ by ungrouping the items. Add the items in $I_k$ by respecting the constraints of $k$BP to get a solution for $D_k$.

\subsubsection{Karmarkar-Karp Algorithm 1 extension to \texorpdfstring{$k$}{k}BP:} \label{SUBSUBSECTION: KKAlgorithm1}
Algorithm 1 of the Karmarkar-Karp algorithms uses the linear grouping technique as illustrated in \Cref{section:EAA:subsection:general}. We give the extension of the Karmarkar-Karp Algorithm 1 to $k$BP in \Cref{kkalgorithms: algorithm1}

\begin{toappendix}
	\subsection{Karmarkar-Karp Algorithm 1 extension to \texorpdfstring{$k$}{k}BP}
	\label{kkalgorithms: algorithm1}
	Recall that the algorithm 1 of Karmarkar-Karp algorithms uses the linear-grouping technique \Cref{section:EAA:subsection:general}. 
	
	\begin{algorithm}[H] 
		\DontPrintSemicolon
		\KwIn{A set $D$ of items, a number $\epsilon \in (0,1/2]$, and an integer $k$.}
		\KwOut{A bin-packing of $D_k$.}
		Let an item of $D$ is \emph{small} if it has size at  most $\max \{1/n, \epsilon\} \cdot S$ and \emph{large}, otherwise. Let $I$ and $J$ be multisets of small and large items of $D$, respectively.  Let $I_k$ be the $k$ copies of the instance $I$. \;
		Perform the linear grouping on the instance $J$ by making groups of cardinality $g = \lceil n(J) \cdot {\epsilon}^2 \rceil $. Denote the resulting instances as ${U'}$  and  $U''$, as defined in the linear grouping technique in 
        \Cref{appendix:linear grouping} \;
		Pack each item in ${U'_k}$ into a bin. There can be at most $g \cdot k$ bins, because each item of $D$ has $k$ copies in $D_k$. \;
		Solve ${U''}_k$ using a fractional linear program with tolerance $h=1$. Denote the resulting solution as $\mathbf{x}$. \;
		Construct the integer solution from $\mathbf{x}$ by the rounding method using no more than $1 \cdot \mathbf{x} + \frac{m(U'')+k}{2}$ bins (see \Cref{kk:lemm2cor1}). \;
		Reduce the item sizes as necessary and obtain a packing for $J_k$ along with the addition of bins as in step 3. \;
		Obtain a packing for $D_k$ by adding the items in $I_k$ which are kept aside in step 1 by respecting constraint of $k$BP, creating new bin(s) when necessary. \;
		\caption{Karmarkar-Karp Algorithm 1 extension to $k$BP}
		\label{kk:algo1}
	\end{algorithm}
\end{toappendix}

\begin{theoremrep} \label{kkalgorithms:algorithm1:theorem:bin(Dk)<=(1+2kepsilon)OPT+additiveterms}
	Let $A(D_k)$ denote the number of bins produced by Karmarkar-Karp Algorithm 1 extension to solve \kbp. Then,
	$A(D_k) \leq (1 + 2 \cdot k \cdot \epsilon)OPT(D_k) + \frac{1}{2 \cdot \epsilon^2} + (2 \cdot k+1)$.
\end{theoremrep}
\begin{appendixproof}
	Since bin-packing is monotone.
	\begin{equation} \label{eqn:kk:algorithm1:theorem1:1}
		OPT({U''}_k) \leq OPT(J_k) \leq OPT(D_k)  
	\end{equation}
	The size of each item in $J_k$ is at least $\frac{\epsilon}{2} \cdot S$. Therefore, from \Cref{kkalgorithms: V(Dk)<=LIN(Dk)S+S(m(Dk)+k)/2},
	\begin{align} \label{eqn:kk:algorithm1:theorem1:2}
		S\cdot OPT(J_k) &\geq V(J_k) \geq \frac{\epsilon}{2} \cdot S \cdot n(J_k) \nonumber \\
		OPT(J_k) &\geq \frac{\epsilon}{2} \cdot n(J_k).
	\end{align}
	From algorithm step 2 and from \Cref{eqn:kk:algorithm1:theorem1:1} and \Cref{eqn:kk:algorithm1:theorem1:2},
	\[ g = \lceil n(J_k) \epsilon^2 \rceil \leq 2 \cdot \epsilon \cdot OPT(J_k) + 1 \leq 2 \cdot \epsilon \cdot OPT(D_k) + 1  \]
	\[ m({U''}_k) = n({U''}_k)/g = n({U''}_k)/{ \lceil n(J_k)\epsilon^2 \rceil } \leq n(J_k)/ { \lceil n(J_k)\epsilon^2 \rceil } \leq 1/\epsilon^2.  \]
	We have tolerance $h=1$, so
	\[ \mathbf{1 \cdot x} \leq LIN({U''}_k) + 1 \leq OPT({U''}_k) + 1 \leq OPT(D_k) +1 \]
	Step $6$ will result in the number of bins which is at most
	\[ \mathbf{1 \cdot x} + \frac{m({U''}_k) + k}{2} + g \cdot k  \leq\]
	\[   OPT(D_k) +1 + \frac{k + \frac{1}{\epsilon^2}}{2} + k \cdot (2 \cdot \epsilon \cdot OPT(D_k) + 1) \leq\]
	\[ \ (1 + 2 \cdot k \cdot \epsilon)OPT(D_k) + \frac{1}{2 \cdot \epsilon^2} + (2k+1). \]
	
	By \Cref{general:lemma1}, the number $bins(D_k)$ of bins required at step $7$ of the algorithm is at most 
	$$\max \left\{ (1 + 2 \cdot k \cdot \epsilon)OPT(D_k) + \frac{1}{2 \cdot \epsilon^2} + (2k+1), (1 + 2\ \cdot \epsilon)OPT(D_k) + k \right\}=$$ $$(1 + 2 \cdot k \cdot \epsilon)OPT(D_k) + \frac{1}{2 \cdot \epsilon^2} + (2k+1). \notag \nonumber $$
\end{appendixproof}

We give the proof of the above Theorem and the running time of the Karmarkar-Karp Algorithm 1 extension to $k$BP in \Cref{kkalgorithms: algorithm1}.

\begin{toappendix}	
	\paragraph{Running time of \Cref{kk:algo1}.}  \label{kk:algorithm1:running-time}
	The time taken by fractional linear program is at most $T(m(U''_k),n(U''_k) ) \leq T\left(\frac{1}{\epsilon^2}, n(D_k)\right)$. The rest of the steps in Algorithm 1 will take at most $O(n(D_k) \log {n(D_k)})$ time. Therefore, the execution time of above Algorithm 1 is at most $O\left(n(D_k) \log {n(D_k)} + T(\frac{1}{\epsilon^2}, n(D_k)\right)$.
\end{toappendix}


\subsubsection{Karmarkar-Karp Algorithm 2 extension to \texorpdfstring{$k$}{k}BP.} \label{SUBSUBSECTION: KKAlgorithm2}
Algorithm 2 of Karmarkar and Karp uses the \textit{alternative geometric grouping technique}. Let $J$ be some instance and $g > 1$ be some integer parameter, then, alternative geometric grouping partitions the items in $J$ into groups such that each group contains the necessary number of items so that the size of each group but the last (i.e. the sum of the item sizes in that group) is at least $g \cdot S$. See \Cref{kkalgorithms:algorithm2} for more details
on alternative geometric grouping and the extension of the Karmarkar-karp Algorithm 2 to \kbp.

\begin{toappendix}
	\subsection{Karmarkar-Karp Algorithm 2 extension to \texorpdfstring{$k$}{k}BP} \label{kkalgorithms:algorithm2}
	
	\paragraph{Alternative Geometric Grouping.} Let $J$ be an instance and $g>1$ be an integer parameter. Sort the items in $J$ in a non-increasing order of their size. Now, we partition $J$ into groups $G_1,\dots G_r$ 
	each containing necessary number of items so that the size of each but the last 
	group (the sum of item sizes in that group) is at least $g \cdot S$.
	
	For natural $i\le r$ let $l_j$ be the  cardinality of the group $G_i$, that is the number of items in the group.
	Let $l_r$ be the cardinality of $G_r$. Note that the size of $G_r$ may be less than $g \cdot S$.
	
	Since each item size is less than $S$, the number of items in each but the last
	group is at least $g+1$. Let $\epsilon \cdot S$ be the smallest item size in $J$. Then, the number of items in each but the last 
	group is at most $\frac{g}{\epsilon}$. 
	Therefore,
	$
	g+1 \leq l_1 \leq \ldots \leq l_r \leq \frac{g}{\epsilon}
	$. 
	For each natural $i<r$ let ${G'}_{i+1}$ be the group obtained by increasing the size of each item in ${G}_{i+1}$ to the maximum item size in that group except for the smallest $l_{i+1} - l_{i}$ items in that group. Let $U' = \bigcup_{i=2}^r \Delta {G_i} \cup G_1$ and $U'' = \bigcup_{i=2}^r G_i^{'}$, where $\Delta G_i$ is the multiset of smallest $l_i - l_{i-1}$ items in $G_i$ for each natural $i$ with $2\le i\le r$ .
	
	Since bin-packing is monotone, $OPT(J) \leq OPT(U') + OPT(U'')$. By \Cref{kkalgorithms: opt-dk<=2V(Dk)/s+k}, $OPT(U') \leq 2 \cdot V(U')/S + 1$ (note that $k=1$). It holds that $V(U') \leq S \cdot g \cdot (1 + \ln {\frac{1}{\epsilon \cdot S}}) + S$ \cite[page 315]{karmarkar-efficient-1982}. Therefore,
	\begin{equation} \label{kk:equation4}
		OPT(J) \leq OPT(U'') + 2\cdot g \cdot \left(2 + \ln {\frac{1}{\epsilon \cdot S}} \right)
	\end{equation}
	Also, $V(U'') \geq {\sum_{i=2}^r} g \cdot S \cdot \frac{l_{i-1}}{l_i} \geq g \cdot S \cdot \left( (r-1) - \ln \frac{1}{\epsilon \cdot S} \right)$. 
	Therefore, by \cite{karmarkar-efficient-1982},  
	\begin{equation} \label{kk:equation5}
		m(U'') = r-1 \leq \frac {V(U'')}{g \cdot S} + \ln {\frac{1}{\epsilon \cdot S}}.
	\end{equation} 
	
	\begin{algorithm}[H]
		\DontPrintSemicolon
		\KwIn{A set $D$ of items, $\epsilon \in (0,1/2]$, and integers $g > 1$ and $k$.}
		\KwOut{A bin-packing of $D_k$.}
		Let $J$ be the instance obtained after removing all items of size at most $ \epsilon \cdot S$ from $D$.  \;
		\While{$V(J) > S \cdot (1 + \frac{g}{g-1} \cdot \ln{\frac{1}{\epsilon}})$}{
			Do the alternative geometric grouping with the parameter $g$. Let the resulting instances be $U'$ and $U''$.  \;
			Solve ${U''}_k$ using fractional bin-packing with the tolerance $h=1$. Denote the resuting basic feasible solution as $\textbf{x}$. \;
			Construct an integer solution from $\textbf{x}$. Create $\lfloor x_i \rfloor$ bins for each $x_i$ in $\textbf{x}$. Remove the packed items from $J$. \;
			Pack ${U'}_k$ in  at most $2 \cdot k \cdot g \cdot \left[ 2 + \ln \frac{1}{\epsilon} \right]$ bins. \;
		}
		When $V(J) \leq S \cdot (1 + \frac{g}{g-1} \cdot \ln{\frac{1}{\epsilon}})$, pack the remaining items greedily in at most $(2 + \frac{2 \cdot g}{g-1} \cdot \ln {\frac{1}{\epsilon}})$ bins and create $k$ copies of this packing. \;
		Greedily pack the $k$ copies of the items of size at most $\epsilon \cdot S$ respecting $k$BP constraint to get a solution for $D_k$. \;
		\caption{Karmarkar-Karp Algorithm 2 extension to $k$BP}
		\label{kk:algorithm2}
	\end{algorithm}	
\end{toappendix}

\begin{theoremrep} \label{kkalgorithm2:theorem:bin(Dk)<=opt+O(klog2opt)}
    Let $A(D_k)$ denote the number of bins produced by Karmarkar-Karp Algorithm 2 extension to solve \kbp. Then,
    $A(D_k) \leq OPT(D_k) + O(k \cdot \log^2 {OPT(D)})$.
\end{theoremrep}
\begin{appendixproof}
		Let ${U''}_{k,i}$ 
		be the instance of $U''_k$ 
		at the beginning of the $i$th iteration of the loop in step $2$. Then the fractional part of $\mathbf{x}$ contains at most $m({U''}_{k,i})$ 
		non-zero variables $x_j$ for $j\le t$. Therefore, from 
		\eqref{kk:equation5} we have
		\begin{equation} \label{kk:algorithm2: analysis:equation1}
			m({U''}_{k,i})  = m({U''}_{1,i}) \leq \frac{V({U''}_{1,i})}{g\cdot S} + \ln{\frac{1}{\epsilon \cdot S}}
		\end{equation}
		\emph{The number of iterations in step $6$ of algorithm.}
		As shown in \cite{karmarkar-efficient-1982} the number of iterations is at most $\frac{\ln {V({D})}}{\ln {g}}+ 1$.
		
		\emph{Number of bins produced using Fractional Bin-Packing solution in steps $3-5$.}
		Let $J_i$ denote the item set 
		$J$ 
		at the beginning of $i$th iteration. On applying alternative geometric grouping with the parameter $g$, we get 
		the instances ${U'}_{1,i}$ and ${U''}_{1,i}$ . Let 
		${U'}_{k,i}$ and ${U''}_{k,i}$ be the corresponding $k$BP instances.
		Let $B_i$ be the number of bins obtained in solving ${U''}_{k,i}$ by applying fractional bin-packing at the iteration $i$ with the tolerance $h=1$. Then
		\[
		B_i \leq LIN({U''}_{k,i}) + 1.
		\]
		Summing across all iterations
		\begin{align} \label{kk:algorithm2: analysis:equation2}
			\sum B_i &\leq LIN({U''}_{k}) + \frac{\ln {V(D) }}{\ln {g}}+ 1 \nonumber \\
			&\leq OPT({U''}_{k}) + \frac{\ln {V({D}) }}{\ln {g}}+ 1
		\end{align}
		\emph{The total number of bins used.}
		Let $A(D_k)$ denote the total number of used bins
		Then $A(D_k)$ is the sum of  the bins used 
		in steps $3-5$, in step $6$, and in step $8$. So 
		\begin{equation}
			\begin{split}
				A(D_k) \leq OPT(D_k) + \left(\frac{\ln{V(D)}}{\ln{g}} + 1\right)\left(1 + 4 \cdot k \cdot g + 2 \cdot k \cdot g \ln{\frac{1}{\epsilon}}\right) + \\ 
				k\cdot \left(2 + \frac{2 \cdot g}{g-1} \ln{\frac{1}{\epsilon}}\right).	
			\end{split}
		\end{equation}
		By \Cref{general:lemma1}, 
		\begin{equation}
			\begin{split}
				A(D_k) \leq \max 
				\bigg\{
				OPT(D_k) + 
				\left(
				\frac{\ln{V(D)}}{\ln{g}} + 1
				\right)
				\left(
				1 + 4 \cdot k \cdot g + 2 \cdot k \cdot g \ln{\frac{1}{\epsilon}}
				\right) + 
				\\ 
				k\cdot
				\left(
				2 + \frac{2 \cdot g}{g-1} \ln{\frac{1}{\epsilon}}
				\right), 
				(1+2 \cdot \epsilon)\cdot OPT(D_k) + k  
				\bigg\}
			\end{split}
		\end{equation}
		Choosing $g=2$ and $\epsilon=\frac{1}{V(D)}$ will result in 
		\begin{align}
			A(D_k) &\leq OPT(D_k) + O(k \cdot \log^2 {OPT(D)})
		\end{align}
		number of bins.
        \qed
\end{appendixproof}
	
\begin{toappendix}
	\paragraph{Running time of \Cref{kk:algorithm2}:}
	By \Cref{kk:equation5}, at each iteration of the while loop, $m({U''}_{1,i})$ decreases by a factor of $g \cdot S$. During the $i$th iteration of the while loop, the time taken by fractional bin-packing is at most $T(m({U''}_{k,i}), n({U''}_{k,i}))$. The remaining steps will take at most $O(n(D_k) \cdot \log {n(D_k)})$ time. Therefore, the total time taken by the algorithm is $O\left(T\left(\frac{n(D) \cdot S}{g \cdot S} + \ln {\frac{1}{\epsilon}},n(D_k)\right) + n(D_k) \cdot \log{n(D_k)}\right)$. For the choice $g=2$ and $\epsilon=\frac{1}{V(D)}$ the running time will be $O\left(T\left(\frac{n(D)}{2} + \ln {V(D)},n(D_k)\right) + n(D_k) \cdot \log{n(D_k)}\right) \in O\left(T\left(\frac{n(D)}{2},n(D_k)\right) + n(D_k) \cdot \log{n(D_k)}\right)$.
\end{toappendix}

We give the proof of the above Theorem and the running time of the Karmarkar-Karp Algorithm 2 extension to $k$BP in \Cref{kkalgorithms:algorithm2}.


\section{Egalitarian allocation of watts}
\label{sec: Egalitarian allocation of supply}
So far, our aim was to maximize the connection time during which a  household is connected to electricity. Another possible direction of research in the egalitarian allocation of the electricity distribution problem is to maximize the minimum allocation of electricity in watts.
In terms of allocation of electricity in watts, the utility of an agent $i$ is defined as 
\begin{align*}
    u_i(\mathcal{I}) = \sum_{l : i \in A_l} D[i] \cdot | I_l |
\end{align*}
Then, the optimization objective is
\begin{align*}
    \max_{\mathcal{I}} 
    \min_{i \in [n]} 
    u_i(\mathcal{I})
\end{align*}
where the maximum is over all partitions that satisfy the demand constraints. This max-min value is called the \emph{egalitarian watts allocation} of the given instance.

\subsection{Examples}
\label{chap:equal supply algorithms-sec:impossibility result}
We illustrate the difference between the two variants of the egalitarian allocation problem using several examples. Throughout the examples, the demands are given in kW, and the time-interval for allocation is 1 hour.

\begin{example}
\label{exm:watts-1}
There are two agents, with demands $D[1]=3$ and $D[2]=5$ kW. The supply is $5$ kW. Egalitarian allocation of time (as discussed in the previous section) would connect each agent exactly $1/2$ of an hour; the agents get $3/2$ and $5/2$ kWh respectively. Egalitarian allocation of watts would connect agent 1 for $5/8$ of an hour and agent 2 for $3/8$ of an hour, so that each agent gets $15/8$ kWh.
\end{example}

With egalitarian watts allocation, agents with small demands will necessarily get small allocation of watts, even if they are connected all the time. If we look only on the smallest amount of watt allocation, the problem might become uninteresting.

\begin{example}
\label{exm:watts-2}
There are three agents with demand $D[1] = M, D[2] = M-1, D[3] = 1$ for some large integer $M >1$. The supply is $S = M+1$.
	Agent 3 can get at most $1$ kWh. Hence, an algorithm that only consider the minimum amount can connect agents $\{1,3\}$ for $1/M$ of the time and agents $\{2,3\}$ for $(M-1)/M$ of the time; then agents 1 and 3 get 1 kWh each, and agent 2 gets $(M-1)^2/M$ kWh. This allocation is clearly unfair towards agent 1.

A fairer solution would be to connect 	agents $\{1,3\}$ for a fraction $\frac{M-1}{2M-1}$ of the time, and the agents $\{2,3\}$ for a fraction $\frac{M}{2M-1}$ of the time. Then, agent 3 gets 1 kWh, whereas agents 1 and 2 both get $\frac{M(M-1)}{2M-1}$ kWh.
\end{example}

One way to get to the fairer solution is to use the \emph{leximin criterion}. A simple way to apply this criterion is using the following two-step allocation process:
\begin{enumerate}
	\item Choose an integer $g\geq 0$, and choose a subset $G$ of the $g$ agents with smallest demand; these agents will be connected all the time.
	\item Compute an egalitarian watts allocation for a reduced instance, in which the agents are $D\setminus G$ and the supply is $S - V(G)$ (the original supply minus the total demand of agents in $G$).
\end{enumerate}
Run this procedure for different values of $g$, and pick the solution with the highest \emph{leximin vector}.

In \Cref{exm:watts-1} above, the optimal $g$ is $0$; in \Cref{exm:watts-2}, the optimal $g$ is $1$.

In \Cref{section:existence of finite k for kBP}, we proved that egalitarian time allocation can always be solved using \kbp. Formally, there exists a function $K(n)$ such that, any egalitarian time allocation problem with $n$ agents can be solved by computing an optimal \kbp solution for some $k\leq K(n)$, and connecting each bin an equal amount of time.

This procedure clearly does not yield an egalitarian watts allocation, as agents with a smaller demand should be connected for a longer time. 
One could hope that this approach could be extended by having a different $k$ for different agents. For example, in \Cref{exm:watts-1} we could choose $k_1=5$ and $k_2=3$. The solution to the generalized \kbp instance would have $8$ bins: $5$ bins with agent 1 alone and $3$ bins with agent 2 alone. Connecting each bin for $1/8$ of the time would yield the egalitarian watts allocation.

One could think that setting $k_i$ inversely proportional to $D[i]$, such that $k_i \cdot D[i]$ would be a constant for all agents, would lead to an egalitarian watts allocation. But this is not the case:
\begin{example}
There are three agents with demands $D[1]=1, D[2]=2, D[3]=3$, and $S=5$.
	To make the product $k_i \cdot D[i]$ constant, we need 	$k_1=6, k_2=3, k_3=2$. The optimal bin allocation satisfying these requirements has 6 bins: $\{1,2\},\{1,2\},\{1,2\},\{1,3\},\{1,3\},\{1\}$. Connecting each bin $1/6$ of an hour results in each agent receiving $1$ kWh. 
	
But if we consider only the following $5$ bins:  $\{1,2\},\{1,2\},\{1,2\},\{1,3\},\{1,3\}$, and connect each bin for $1/5$ of the time, agent 1 still gets 1 kWh, whereas agents 2 and 3 get more: each of them gets $6/5$ kWh.
\end{example}
This issue can be solved by the same procedure mentioned after \Cref{exm:watts-2}: taking $g=1$ leads to connecting agent 1 all the time, and solving the remaining instance with $k_2=3$ and $k_3=2$.

Unfortunately, we cannot get a universal bound on the required $k_i$'s (that depends only on $n$).
\begin{example}
\label{exm:watts-3}
There are two agents with demand $D[1]$ and $D[2]$, which are co-prime integers. Suppose $S < D[1] + D[2]$ but $S \geq \max \{D[1], D[2] \}$. Then, the only feasible configurations are $\{D[1]\}$ and $\{D[2]\}$. For an egalitarian allocation, $D[1]$ should be connected $\frac{D[2]}{D[1] + D[2]}$ fraction of the time, and $D[2]$ should be connected $\frac{D[1]}{D[1] + D[2]}$ fraction of the time.
This means that, if each bin is connected the same amount of time, then we should have $D[2]$ bins with configuration $\{D[1]\}$ and $D[1]$ bins with configuration $\{D[2]\}$. So we have $k_1 = D[2]$ and $k_2 = D[1]$. These numbers can be arbitrarily high, even though $n=2$ is fixed.   
\end{example}

Another potential problem with the above approach
arises when 
there are some items with very small item sizes compared to the largest item size. Imagine the situation where the smallest item in the instance is, say, $0.05$ and the maximum item size in the instance is, say, $14$. In this case, there are $280$ copies of this smallest item w.r.t a single copy of the maximum item size. The final derived instance will be a large instance in terms of the number of items. Moreover, in this case, the equal watts allocation is limited by the small item size, i.e., the agent with max demand $14$ will not get electricity more than $0.05$(kW). 
We can determine such items $D[i]$, and $g$ is initialized with the count of such items.

\emph{Assumptions:} --- We assume that $V(D) > S$. Otherwise, the problem is trivial, as we can connect all agents simultaneously, resulting in allocating the watts equal to their demand. \newline
--- When we say ``original'' \kbp it means that we are talking about the \kbp problem where each item appears exactly $k$ times in an instance (as mentioned in 
\Cref{sec:intro subsec:kbp,sec:model subsec:kbp}).

Below, we present a number of heuristic algorithms to tackle the problem of the electricity distribution of watts as equally as possible.

    In all heuristics, we first determine a set of $g$ agents with smallest demands that will be connected all the time. Then, we apply heuristics to the remaining supply and the remaining agents.
    
    We denote $d_{\max} := \max(D) = \text{ the largest demand in } D$. Here, we use the term ``item'' and ``demand'' interchangeably.
    
    \subsection{Heuristic Algorithm 1 (HA1)}
    \label{subsec: heuristic algorithm 1}

        Let $D^{\uparrow}$ be the instance $D$ sorted in non-decreasing order. Let $d_l$ be the maximum demand in $D^{\uparrow}$ such that the sum of all the item sizes $\leq d_l$ is at most $S - d_{\max}$ i.e. they all can be packed in a bin whose capacity is  $S - d_{\max}$, and adding the next item larger than $d_l$ violates the bin capacity (with bin size $S - d_{\max})$ constraint. 
    Clearly, such an item exists; otherwise, all the items can be packed into a single bin and connected all the time.

    \begin{example} \label{example: example1 for HA1}
        Let $D = \{0.2, 0.22, 0.4, 0.42, 0.8, 0.82, 1.7, 1.7, 3, 3.2, 6.5, 6.7, 14, 14.2\}$ and supply is $S=21$.
        Then, $d_{\max} = 14.2$. $D^{\uparrow}$ will be $D^{\uparrow} = \{ 0.2, 0.22, 0.4, 0.42, 0.8, 0.82, 1.7, 1.7, 3, 3.2, 6.5, 6.7, 14, 14.2 \}$. Note that $S - d_{\max} = 21 - 14.2 = 6.8$.  $d_l$ will be $1.7$ because the sum of all the items $\leq d_l = 1.7$ (that is the items in the set $\{0.2, 0.22, 0.4, 0.42, 0.8, 0.82, 1.7, 1.7 \}$) will be $6.26$ that is $\leq S - d_{\max} = 6.8$. Clearly, the next item larger than $d_l$ cannot be added to it, as then their sum will be $9.26$ that is, $> S - d_{\max} = 6.8$.
    \end{example}

    Let $G \subseteq D$ be some subset of items
    such that all the items in $G$ have sizes $\leq d_l$. The updated set of remaining agents is $D' = D\setminus G$, and the updated supply is $S' = S - V(G)$. We determine $G$ as follows:
    \begin{itemize}
        

        \item For some natural $g$, $G$ contains the first $g$ items in the set $D^{\uparrow}$.
        In \Cref{example: example1 for HA1}, there are no smallest items for which the number of copies is too large w.r.t a single copy of the maximum item size. Hence, initially $g$ will be $0$.
        such items. 
        Now, there are $9$ possible values of $g$ viz, $0,1,2,3,4,5,6,7,8$.
        As an example, for $g=2$, the set $G$ will be $G = \{0.2, 0.22\}$.
    \end{itemize}

    Our heuristic algorithm 1 connects agents in $G$ all the time. Hence, these agents get their required demand.
    We derive an instance $D''$ as follows:
    \begin{itemize}
	   \item[--] for each item $D'[i]$ in $D'$ we compute the number of copies that needs to be created for (nearly) equal allocation of watts as:
        $\frac{k \cdot d_{\max}}{{D'}[i]}$  rounded to the nearest integer, where $k$ denotes the number of copies of $d_{\max}$.
	   \item[--] Now,
	   \begin{enumerate}
		\item for each item 
        $D'[i]$ in $D'$, if the number of copies of demand ${D'}[i]$ are non-zero then take a copy of it and insert it into the instance $D''$. Reduce the number of copies of demand $D'[i]$ by $1$. 
		\item repeat the above procedure until the number of copies of each item becomes $0$. 
	   \end{enumerate}
    \end{itemize}

    Now, we can apply $FFk$ or $FFDk$ to the derived instance $D''$. Let $B$ be the resulting solution.
    Solution $B$ is
    \emph{incomplete} in the sense that they do not contain items in $G$.
    To each bin $B_i \in B$, we add the items in $G$, these are the agents that are connected all the time.
    Doing so will \emph{complete} the solution $B$.
    After that, compute the watts allocation vector for this computed \emph{complete} solution $B$.
    Different such solutions can be computed for all possible values of $g$. Among these solutions, the one that is leximin-preferred over other solutions (this is done by comparing the watts allocation vector) is chosen as the best solution. Ties are broken arbitrarily.

    The above approach can be time-consuming as there can be different possible values of $g$ for each instance.
    To run it faster, we use the ternary search
    \footnote{We are thankful for the following stackoverflow answer \url{https://stackoverflow.com/a/40695871}.}
    to determine the possible solution that is leximin-preferred over other computed solutions.

    The rationale behind using the ternary search is as follows: if $G$ contains too few demands then the watts allocation may be low as the remaining agents may contain small demands;
    if $G$ contains too many agents then again the watts allocation for the remaining agents may be low because the remaining capacity may not be sufficient to pack as many demands as possible from the remaining set of demands.
    Therefore, the peak lies somewhere in between them.
    Examples to support this claim have been given in \Cref{Example support for the applicability of ternary search}.

    \begin{toappendix} 
        \section{Example support for the applicability of ternary search}
        \label{Example support for the applicability of ternary search}
        To demonstrate this we have taken some examples from the dataset (same dataset used by \cite{oluwasuji_algorithms_2018,oluwasuji_solving_2020}) used in this paper. Each of these examples has aggregate demand greater than the supply. It makes sense as otherwise every agent gets their stipulated demand. \Cref{figure:ternary-search-example-parta} and \Cref{figure:ternary-search-example-partb} shows that peak for egalitarian allocation of electricity supply lies between the small and large values of $G$. This suggests the use of ternary search.
        \begin{figure}[htbp]
            \centering
            \begin{subfigure}[b]{\textwidth}
                \centering
                \includegraphics[width=0.82\linewidth, height=0.82\textheight, keepaspectratio]{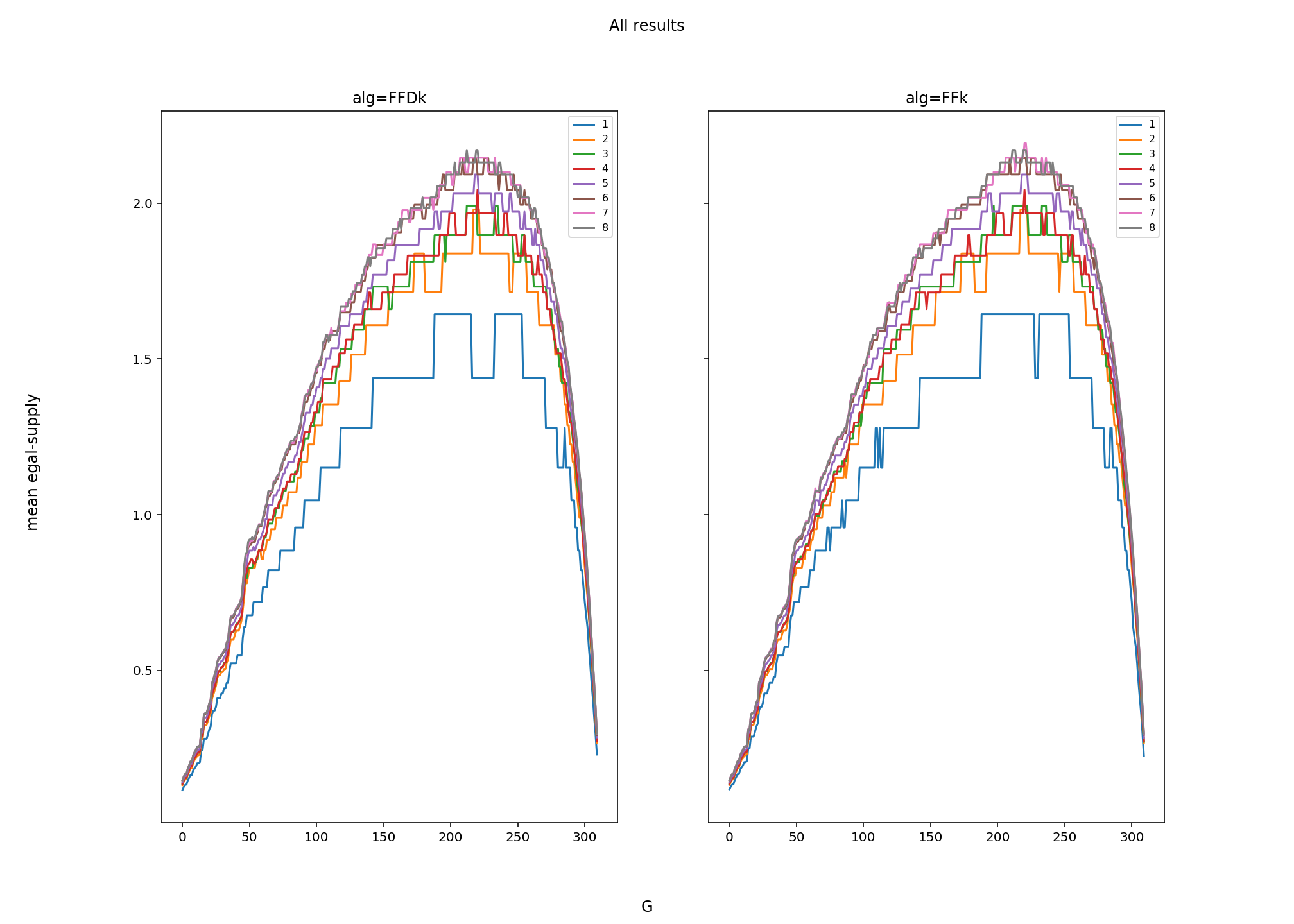}
                \caption{}
                \label{figure:example1-ternarysearch}
            \end{subfigure}
            \hfill
            \begin{subfigure}[b]{\textwidth}
                \centering
                \includegraphics[width=0.82\linewidth, height=0.82\textheight, keepaspectratio]{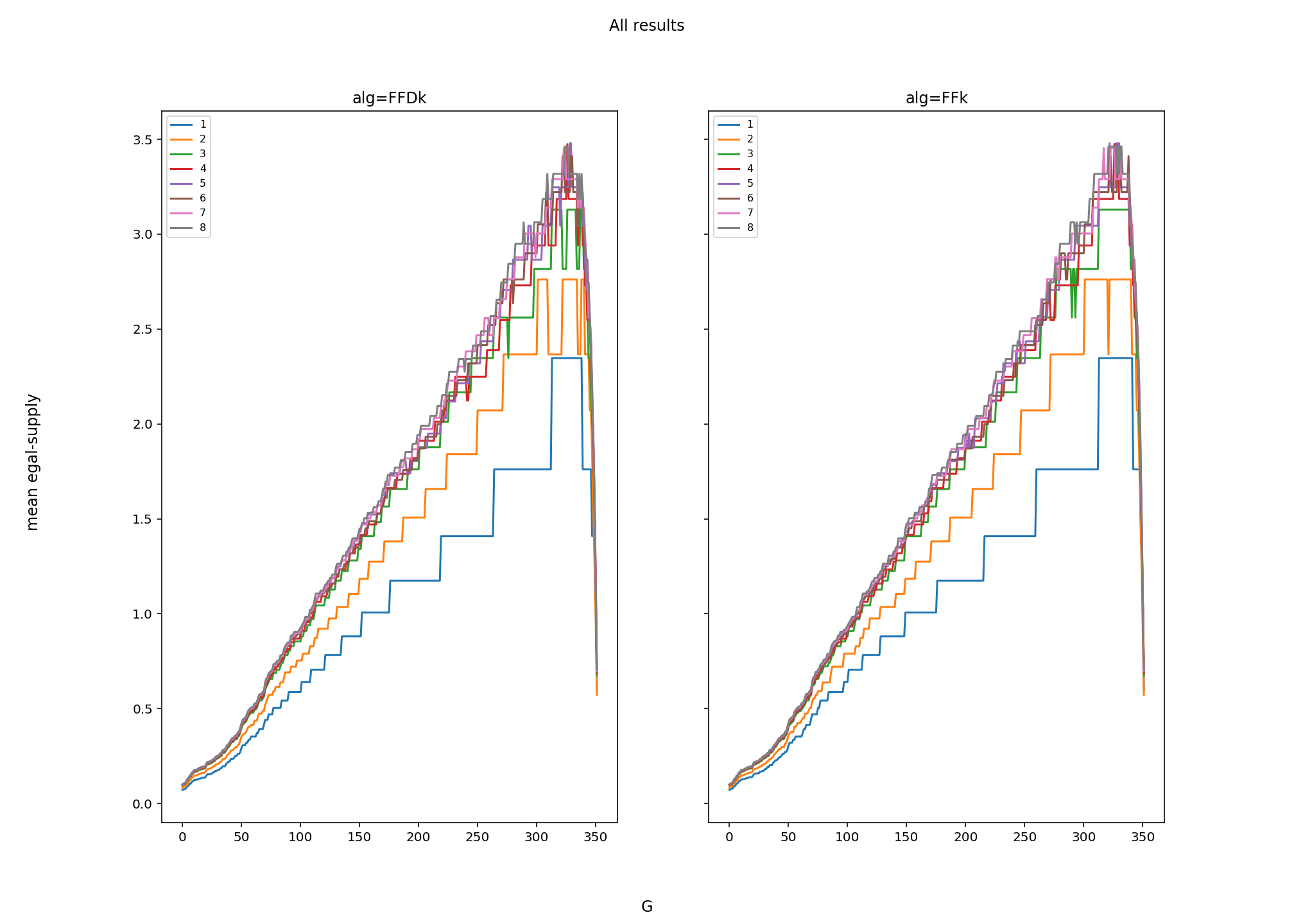}
                \caption{}
                \label{figure:example2-ternarysearch}
            \end{subfigure}
            \caption{Values of $G$ and mean egalitarian value of electricity supply are shown at the $x$ and $y-$ axis respectively. \Cref{figure:example1-ternarysearch} and \Cref{figure:example2-ternarysearch} corresponds to different examples. In each figure we see that the peak lies in between the smallest and highest value of $G$. Hence, ternary search can be used. Different lines denote a different value of $k$.}
            \label{figure:ternary-search-example-parta}
        \end{figure}
    
        \begin{figure}[htbp]
            \begin{subfigure}[b]{\textwidth}
                \centering
                \includegraphics[height=0.82\textheight, width=0.82\linewidth, keepaspectratio]{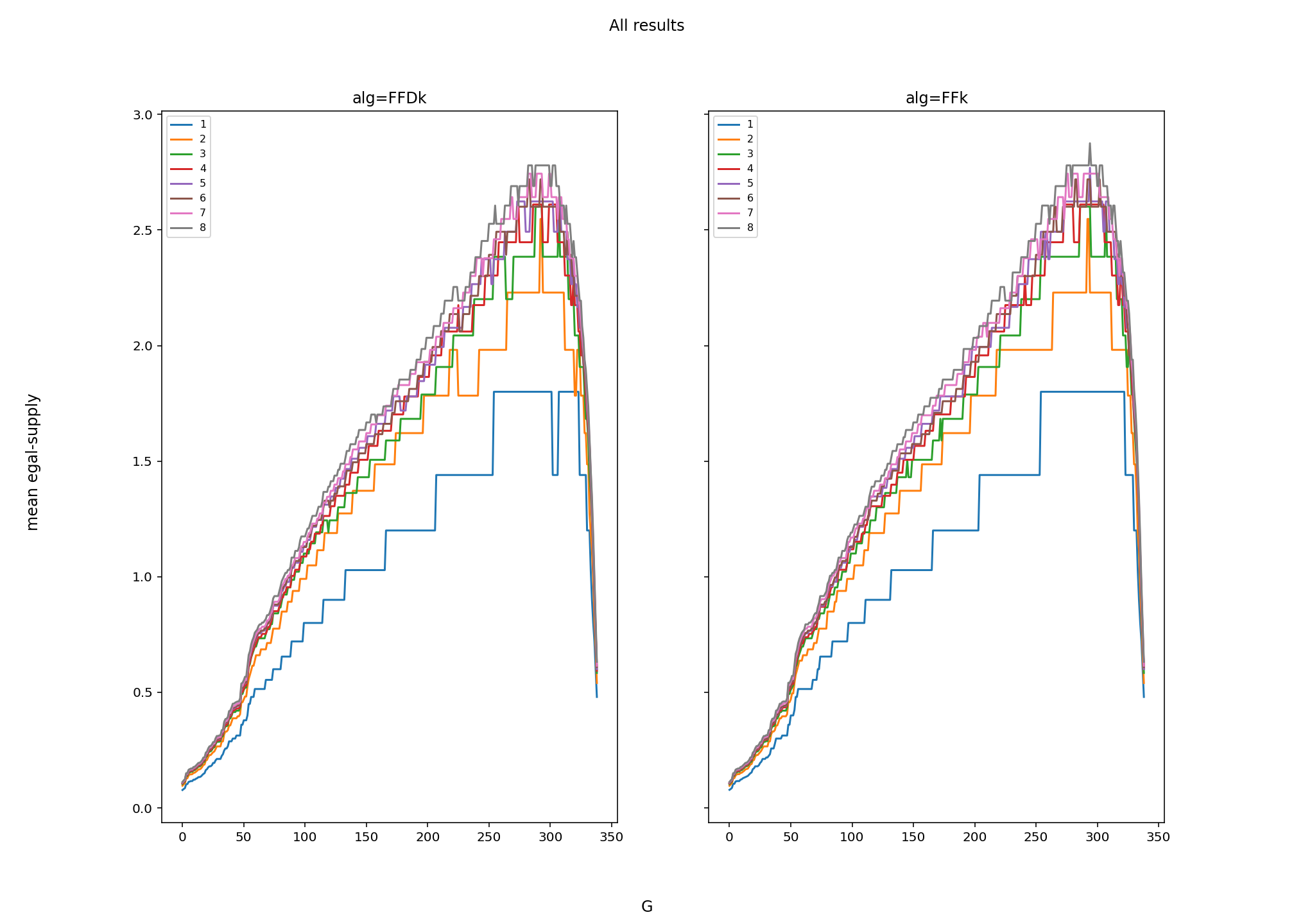}
                \caption{}
                \label{figure:example3-ternarysearch}
            \end{subfigure}
            \hfill
            \begin{subfigure}[b]{\textwidth}
                \centering
                \includegraphics[height=0.82\textheight, width=0.82\linewidth, keepaspectratio]{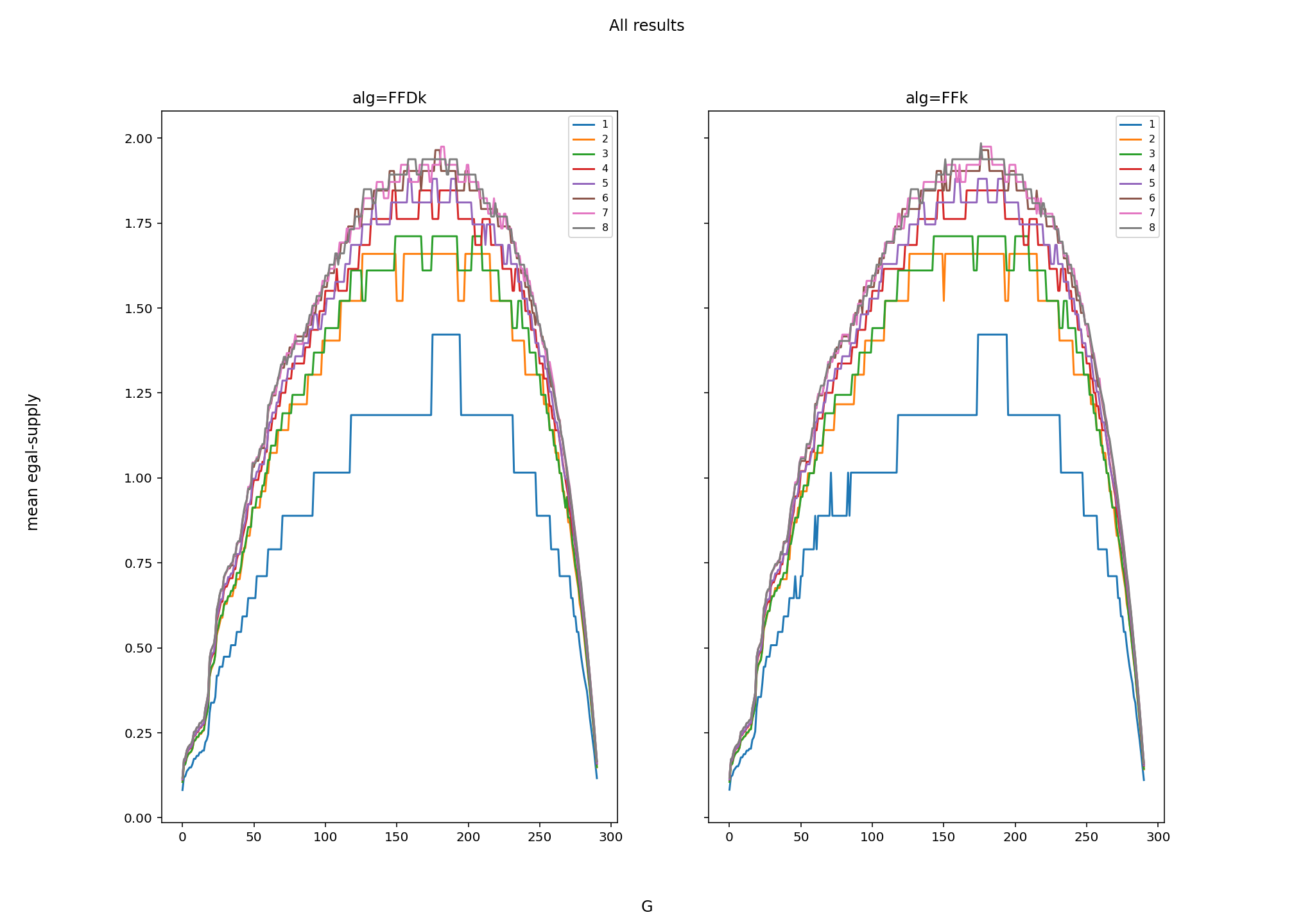}
                \caption{}
                \label{figure:example4-ternarysearch}
            \end{subfigure}
            \caption{Values of $G$ and mean egalitarian value of electricity supply are shown at the $x$ and $y-$ axis respectively. \Cref{figure:example3-ternarysearch} and \Cref{figure:example4-ternarysearch} corresponds to different examples. In each figure we see that the peak lies in between the smallest and highest value of $G$. Hence, ternary search can be used. Different lines denote a different value of $k$.}
            \label{figure:ternary-search-example-partb}
        \end{figure}
    
        The ternary approach is as follows:
    
        \begin{algorithm}[H]
        	\label{ternary search}
        	\DontPrintSemicolon
        	\KwIn{
        		$P$: function which computes the solution; \\
        		$g_{begin}$:= an integer parameter. Initially set to $0$; \\
        		$g_{end}$:= an integer parameter. Initially set to the cardinilaity of the set that contains all the items $\leq d_l$ \;
        	}
        	\SetKw{kwto}{to}
        	
        	\tcc{NOTE: $\ldots$ in $P$ denotes other parameters that function $P$ accepts.}
        	$\mathtt{bestsolution} = None$ \;
        	
        	\While{$g_{end} - g_{begin} > 3$}{
        		Compute $\mathtt{solutionbegin} = P(\ldots, g_{begin})$ \;
        		Compute $\mathtt{solutionend} = P(\ldots, g_{end})$ \;
        		Set $\mathtt{bestsolution}$ to the one which is best among $\mathtt{bestsolution}$, $\mathtt{solutionbegin}$ and $\mathtt{solutionend}$. \;
        		$g_{begin}:= round((g_{begin} \cdot 2 + g_{end})/3,0)$ \;
        		$g_{end}:= round((g_{begin} + g_{end} \cdot 2)/3,0)$ \;
        	}
        	
        	\For{$g = g_{begin}$ \kwto $g = g_{end}$}{
        		Compute $\mathtt{solutiong} = P(\ldots, g)$ \;
        		Compare $\mathtt{bestsolution}$ and $\mathtt{solutiong}$ and set $\mathtt{bestsolution}$ accordingly. \;
        	}
        	\Return $\mathtt{bestsolution}$ \;
        	\caption{Ternary Search}
        \end{algorithm}
    
        Function $P$ involves applying algorithm $FFk$ or $FFDk$ (passed as an argument to $P$) on the instance ${D''}$. Finally, after adding the items in $G$ to each bin of the solution obtained, the function $P$ returns this solution.
    
    \end{toappendix}

    \begin{example} \label{ha1:example2}
        First, we consider \Cref{exm:watts-2}
        In this example there are three agents with demands $M, M-1, 1$ respectively, and the supply is $M+1$.
        There are two possibilities for $G$: for $g=0$ we get $G=\{\}$, and for $g=1$ we get $G=\{1\}$.
In both of the cases for $g$ there is enough remaining space to pack item $d_{\max} = M$. However, for $g=2$ we get $G = \{1, M-1\}$ and there is not enough remaining space to pack item $d_{\max} = M$.

        -- When $G = \{\}$,  HA1 will have $\frac{k \cdot M}{1}$ copies of item $1$, and $\frac{k \cdot M}{M-1}$ (rounded to $0$ decimal points)copies of item $M-1$ where $k$ is the number of copies of item $M$. Then, the rest of the algorithm computes the egalitarian allocation of watts.
        
        To determine the final allocation let's assume that $\frac{k \cdot M}{M-1}$ is an integer. Since, $k \cdot M > \frac{k \cdot M}{M-1} + k = \frac{k \cdot (2M-1)}{M-1}$ for large $M >1$ and also, the item $M-1$ and $M$ can not be packed into the same bin, we can say that the number of bins that contain item $1$ are sufficient to pack the items $M-1$ and $M$. In fact, to pack these items we require no more than $k\cdot M$ bins. Therefore, 
        \begin{itemize}
        \item Item $1$ will get $1 \cdot \frac{k\cdot M}{k \cdot M} = 1$kW, and
        \item item $M-1$ will get $(M-1) \cdot \frac{\frac{k \cdot M}{M-1}}{k \cdot M} = 1$kW, and 
        \item item $M$ will get $M \cdot \frac{k}{k \cdot M} = 1$kW.
        \end{itemize}
        The allocation vector of watts in this case is $(1,1,1)$.

        -- When $G = \{1\}$,  HA1 connects the agent $1$ all the time. For the remaining agents $\{M-1, M\}$, HA1 will have  $\frac{k \cdot M}{M-1}$ copies of item $M-1$ where $k$ is the number of copies of item $M$. 
        Note that the items $M-1$ and $M$ can not fit into a single bin. We require $\frac{k \cdot M}{M-1} + k = \frac{k \cdot (2M-1)}{M-1}$ bins to pack the item $M-1$ and $M$, and each of these bins can accommodate item $1$. Therefore,
        \begin{itemize}
        \item Item $1$ will get $1 \cdot \frac{\frac{k \cdot (2M-1)}{M-1}}{\frac{k \cdot (2M-1)}{M-1}} = 1$kW of electricity, and
        \item item $M-1$ will get
        $(M-1) \cdot \frac{\frac{k \cdot M}{M-1}}{\frac{k \cdot (2M-1)}{M-1}} = \frac{(M-1) \cdot M}{2M-1}>1$kW of electricity, and
        \item item $M$ will get $M \cdot \frac{k}{\frac{k \cdot (2M-1)}{M-1}} = \frac{M \cdot (M-1)}{2M-1} > 1$kW of electricity.
        \end{itemize}
        The allocation vector of watts in this case is $1, \frac{M \cdot (M-1)}{2M-1}, \frac{M \cdot (M-1)}{2M-1}$ which is better than the allocation vector $(1,1,1)$. HA1 outputs this leximin preferred solution.
        

        We can understand above description with the help of an example. Let's take $M=5$. The three agents with demand $M, M-1, 1$ are $5,4,1$ respectively, and the supply is $S = M+1  = 6$. Now, the feasible values for $g$ are $0,1$ (Note that we have already argued why $g=2$ is not feasible). \\
        -- For $g=0$ we get $G = \{\}$. Assume that $k=4$, so there are $4$ copies of the item $d_{\max} = 5$. In this case HA1 will have $\frac{4 \cdot 5=10}{1} = 20$ copies of item $1$, and $\frac{4 \cdot 5 = 10}{4} = 5 \text{ (after rounding to $0$ decimal points)}$ copies of item $4$.
        Now the derived instance $D''$ is 
        $\{5,4,1,5,4,1,5,4,1,5,4,1,4,1,1,1,1,1,1,1,1,1,1,1,1,1,1,1,1\}$ 
        and the updated supply is $S' = S - V(G) = 6 - 0 = 6$.
        Applying $FFk$ to instance $D''$ results in the following packing:
        \begin{align*}
            \{5,1\}, \{4,1\}, \{5,1\}, \{4,1\}, \{5,1\}, \{4,1\}, \{5,1\}, \{4,1\}, \\
            \{4,1\}, \{1\}, \{1\}, \{1\}, \{1\}, \{1\}, \{1\}, \{1\}, \{1\}, \{1\}, \{1\}, \{1\}   
        \end{align*}
        To each bin in this solution we add all the items in $G$. Here $G$ is empty so the final solution is same as computed above. 
        The final allocation of supply in this will be $\frac{4}{20}\cdot 5 = 1$kW for agent with demand $5$kW and $\frac{5}{20}\cdot 4 = 1$kW for agent with demand $4$kW and $1$kW for agent with demand $1$kW. 
        Final allocation vector in this case is $(1,1,1)$ \\
        -- For $g=1$ we get $G = \{1\}$. Assume that $k=4$, so there are $4$ copies of the item $d_{\max} = 5$. In this case HA1 will have $\frac{4 \cdot 5}{4} = 5 \text{ (after rounding to 0 decimal points)}$ copies of item $4$.
        Now the derived instance $D''$ is 
        $\{5,4,5,4,5,4,5,4,4\}$
        and the updated supply is $S' = S - V(G) = 6 - 1 = 5$.
        Applying $FFk$ to instance $D''$ results in the following packing:
        $\{5\}, \{4\}, \{5\}, \{4\}, \{5\}, \{4\}, \{5\}, \{4\}, \{4\}$.
        To each bin in this solution we add all the items in $G$. After adding, the final solution is
        $\{5,1\}, \{4,1\}, \{5,1\}, \{4,1\}, \{5,1\}, \{4,1\}, \{5,1\}, \{4,1\}, \{4,1\}$. 
        The final allocation of supply in this will be $\frac{4}{9}\cdot 5 =\approx 2.2$kW for agent with demand $5$kW and $\frac{5}{9}\cdot 4 \approx 2.2$kW for agent with demand $4$kW and $1$kW for agent with demand $1$kW.
        Final allocation vector in this case will be 
        $(1,2.2,2.2)$. \\
        -- Clearly, allocation $(1,2.2,2.2)$ is preferred over allocation $(1,1,1)$.
    \end{example}

    \begin{example} \label{ha1:example1}
        $D = \{0.2,0.22, 0.4,0.42, 0.8,0.82, 1.7,1.7, 3,3.2, 6.5,6.7, 14,14.2\}, S = 21$.
        
        -- The feasible values for $g$ in this case are: $0,1,2,3,4,5,6,7,8$ ( $9$ is not a feasible value for $g$. See explanation of the HA1 algorithm). 

        -- For $g=0$ we get $G = \{\}$. Assume that $k=3$ so there are $3$ copies of the item $d_\max = 14.2$.
        In this case HA1 will have $\frac{3 \cdot 14.2 = 42.6}{0.2} = 213$ copies of the item with size $0.2$. Similarly, the number of copies of other items can be computed after rounding to $0$ decimal points. 
        The below table shows the item along with their number of copies:
        \begin{center}
            \begin{tabular}{|c|c|c|c|c|c|c|c|c|c|c|c|c|c|c|}
                \hline
                 item & 0.2 & 0.22 & 0.4 & 0.42 & 0.8 & 0.82 & 1.7 & 1.7 & 3 & 3.2 & 6.5 & 6.7 & 14 & 14.2  \\ \hline
                 \#copies & 213 & 194 & 106 & 101 & 53 & 52 & 25 & 25 & 14 & 13 & 7 & 6 & 3 & 3 \\
                 \hline
            \end{tabular}
        \end{center}
        HA1 creates the instance $D''$ as follows:
        \begin{itemize}
            \item[*] first it picks the first item in $D'$ and checks if the number of copies are non-zero. Then picks up this item and inserts it into $D''$. Then it picks the second item in $D$ and does the same. Following this, $D''$ contains
            the $3$ copies of the instance $\{0.2, 0.22, 0.4,0.42, 0.8,0.82, 1.7,1.7, 3,3.2, 6.5,6.7, 14,14.2\}$,
            \item[*] After doing this the updated list of items and their corresponding number of copies are:
            \begin{center}
            \begin{tabular}{|c|c|c|c|c|c|c|c|c|c|c|c|c|c|c|}
                \hline
                 item & 0.2 & 0.22 & 0.4 & 0.42 & 0.8 & 0.82 & 1.7 & 1.7 & 3 & 3.2 & 6.5 & 6.7 & 14 & 14.2  \\ \hline
                 \#copies & 210 & 191 & 103 & 98 & 50 & 49 & 22 & 22 & 11 & 10 & 4 & 3 & 0 & 0 \\
                 \hline
            \end{tabular}
            \end{center}
            Note that the item $14$ and $14.2$ do not have any copies left. Upon repeating the same procedure as explained in above step, the instance $D''$ will contain $3$ copies of the following instance:
            $\{0.2, 0.22, 0.4,0.42, 0.8,0.82, 1.7,1.7, 3,3.2, 6.5,6.7\}$.
            \item[*] HA1 repeats the above procedure until the number of copies of each item becomes $0$. After that it is not possible to make any addition to the instance $D''$ since we are left with no items.
        \end{itemize}

        Now, HA1 applies \kbp to instance $D''$. To each bin in the resulting solution HA1 adds all the items in $G$ to get the final solution. The minimum supply delivered in this case is $0.18613$kW.

        -- For $g=8$ we get $G=\{0.2,0.22, 0.4,0.42, 0.8,0.82, 1.7,1.7\}$. These are the set of items that are connected all the time. The remaining set of items is $D' = \{3,3.2,6.5,6.7,14,14.2\}$.
        There are $k=3$ copies of item $d_\max = 14.2$. 
        HA1 will have $\frac{14.2 \cdot 3 = 42.6}{3} \simeq 14 $ copies of item $3$. Similarly, HA1 computes the number of copies of other items in the remaining set of items. The below table shows the item along with their number of copies:
        \begin{center}
            \begin{tabular}{|c|c|c|c|c|c|c|}
                \hline
                \text{item} & 3 & 3.2 & 6.5 & 6.7 & 14 & 14.2 \\ \hline
                 \text{\#copies} & 14 & 13 & 7 & 6 & 3 & 3 \\ \hline 
            \end{tabular}
        \end{center}
        HA1 creates the instance $D''$ as follows:
        \begin{itemize}
            \item[*] first it picks the first item in $D'$ whose remaining number of copies are non-zero. It then picks up this item and inserts it into $D''$. In doing so HA1 reduces the number of copies of this item by $1$. Then, it picks the second item in $D'$ with non-zero number of copies and inserts it at the end of $D''$. Again, HA1 reduces by $1$ the numbere of copies of this item. Following this procedure $D''$ contains the consecutive $3$ copies of the following instance:
            $\{3, 3.2, 6.5, 6.7, 14, 14.2\}$. 
            The updated table showing the items and their corresponding number of copies is:
            \begin{center}
                \begin{tabular}{|c|c|c|c|c|c|c|}
                    \hline
                    \text{item} & 3 & 3.2 & 6.5 & 6.7 & 14 & 14.2 \\ \hline
                    \text{\#copies} & 11 & 10 & 4 & 3 & 0 & 0 \\ \hline
                \end{tabular}
            \end{center}

            \item[*] HA1 repeats the above procedure until the remaining number of copies of each item becomes $0$. After which no further insertion of items is possible.
        \end{itemize}
        Now, HA1 applies \kbp to this created instance $D''$, and to each bin in the resulting solution HA1 adds all the items in $G$. The minimum supply delivered in this case is $1.74783$kW. 

        -- HA1 computes different solutions in this way for different values of $g$ (HA1 uses ternary search Algorithm \ref{ternary search} to select only a few values of $g$ for which the solution is to be computed). Finally, it outputs a solution that is leximin preferred over other solutions. In this example, HA1 outputs the solution corresponding to $g=8$ in which items in set $\{0.2,0.22, 0.4,0.42, 0.8,0.82, 1.7,1.7\}$ are connected all the time.

    \end{example}

    \subsection{Heuristic Algorithm 2 (HA2)}
    \label{heuristic algorithm 2}
    Our next heuristic algorithm uses the geometric grouping technique of \cite{karmarkar-efficient-1982}. The idea behind this heuristic is as follows:
    
    First, we determine the set of small items using the same technique described in \Cref{subsec: heuristic algorithm 1} 
    ($G$ is initialized to this set of items; For more details, see \Cref{subsec: heuristic algorithm 1}).
    Let $D'$ be the set of remaining demands and $S'$ be the updated supply.
    We pack the items in $D'$ using the approach described below, and to each bin in this solution, we add the items in $G$.
    
    Recall that the geometric grouping works as follows:
    Let $D'$ be some instance to pack the items. Let $d_{\min}$ and $d_{\max}$ be the minimum and maximum item in $D'$. We define 
    $r:= \lfloor \log_2{\frac{d_{\max}}{d_{\min}}} \rfloor$. 
    Now, we divide the instance into groups $L_i$ such that each group $L_i$ contains all the items whose sizes lie in the interval 
    $(d_{\max} \cdot 2^{-(i+1)}, d_{\max} \cdot 2^{-i}]$ for $i = 0, \ldots, r$. 
    Let $L$ be the set of all groups formed using geometric grouping.
    Now, depending on the number of groups formed using geometric grouping, we do the following computation:
    \begin{itemize}
    	\item[--] If there is only one group $L_1 = D'$, then we compute the ``original'' \kbp solution for this group for some value of $k$. 
        The watts allocation for each agent in $L_1$ then is $\frac{k}{B(L_1)}$ where $B(L_1)$ is the number of bins in the ``original'' \kbp solution to $L_1$.
    	
    	\item[--] \textbf{Warm-up case:} Let us assume that there are only two groups $L_0,L_1$. 
    	Note that, 
    	$ L_1[i] \in L_1 \geq \frac{d_{\max}}{2^2} = \frac{1}{2} \cdot \frac{d_{\max}}{2^1}$,
    	and 
    	$L_0[j] \in L_0 \geq \frac{d_{\max}}{2^1}$
    	Therefore, connecting the agents in group $L_{1}$ twice the time allocated to group $L_0$ will result in a near equal allocation of supply to the agents in group $L_{1}, L_0$.
    	
    	\textbf{General case: }Let $\#L$ be the number of groups resulting from geometric grouping. For each group $L_i \in L$ we compute the ``original'' \kbp solution and let $B_i$ denote both the number of bins and also the set of bins in the ``original'' \kbp solution to group $L_i$. To compute the equal watts allocation we generalize the idea as mentioned in the warm-up case to compute the time that each bin is connected to so as to result in the equal watts allocation.
    	
    	Let us denote by $T_{L_i}$ as the final allocation of time to each bin in group $L_i$. Then,
    	\begin{equation} \label{ha2:eqn1} 
    		\sum_{i=0}^{\#L - 1} T_{L_i} \cdot B_i = 1
    	\end{equation}
    	Let us denote by $\tau_{L_i}$ the final time allocated to each agent in the group $L_i$. It may happen that $k$ is different for different groups. Let $k_{L_i}$ be the $k$ for group $L_i$. Since there are $k_{L_i}$ copies of each agent, therefore the total time allocated to each agent in group $L_i$ is
    	\begin{equation} \label{ha2:eqn2}
    		\tau_{L_i} = T_{L_i} \cdot k_{L_i}
    	\end{equation}
    	for group $L_{i}$ we compute the time allocated to each agent in $L_{i}$ as
    	\begin{equation} \label{ha2:eqn3}
    		\tau_{L_{i}} = 2 \cdot \tau_{L_{i-1}}
    	\end{equation}
    	substituting the value of $\tau_{L_i}$ and $\tau_{L_{i-1}}$ from \Cref{ha2:eqn2} in \Cref{ha2:eqn3} we get that
    	\begin{align} \label{ha2:eqn4}
    		T_{L_i} \cdot k_{L_i} = 2 \cdot T_{L_{i-1}} \cdot k_{L_{i-1}} \nonumber \\
    		T_{L_i} = 2 \cdot T_{L_{i-1}} \cdot \frac{k_{L_{i-1}}}{k_{L_i}}
    	\end{align}
    	
    	from \Cref{ha2:eqn1} and \Cref{ha2:eqn4} we get that
    	\begin{align} \label{ha2:eqn5}
    		T_{L_0} = \frac{1}{\sum\limits_{i=0}^{\#L - 1}{B_i \cdot 2^i \cdot \frac{k_{L_0}}{k_{L_i}}}}
    	\end{align}
    	
    	From this, we can compute the time for each bin in other groups using \Cref{ha2:eqn4} and hence the time for each agent using \Cref{ha2:eqn2}.
    \end{itemize}
    
    Consider the following example:
    \begin{example} \label{ha2:example-1}
    	D = \{0.2, 0.4, 0.8, 1.7, 3, 6.5, 14\}, S = 21
    \end{example}
    In the above example, applying geometric grouping will result in $7$ groups $L_0= \{14\}, L_1= \{6.5\}, L_2= \{3\}, L_3= \{1.7\}, L_4= \{0.8\}, L_5= \{0.4\}, L_6= \{0.2\}$, one item in each group. Following the above procedure for computing the resulting allocation of time will result in nearly $\sim 0.1$kW of supply to each agent, which is very small.
    
    To alleviate the problem highlighted by the above example, we use the same approach of 
    \Cref{subsec: heuristic algorithm 1}
    in which items are grouped together to connect all the time ($G$ is augmented with this set), and for the remaining items, we compute the equal allocation of supply.
    Here, we determine the different groups that can be connected all the time 
    and for the remaining groups we compute the (possible) equal allocation of watts. Following this, \Cref{ha2:example-1} results in always connecting the groups
    $L_2, L_3, L_4, L_5, L_6$, and agents in the group $L_1, L_0$ will get the supply $4.33$ kW and $4.67$ kW, respectively.
    
    This allocation of watts to agents in the group $L_0, L_1$ is as follows:
    
    -- The remaining space in bin to pack the items in group $L_0, L_1$ is $14.9$. 
    
    -- For each group, only one bin is needed to pack the items. 
    
    -- Assume that $k$ is the same for each bin. Using \Cref{ha2:eqn5} and \Cref{ha2:eqn2} we can determine that the time allocated to each item in $L_0$ is $1/3$. 
    
    -- Similarly, using \Cref{ha2:eqn4} and \Cref{ha2:eqn2}, we can determine that the time allocated to each agent in $L_1$ is $2/3$. 
    
    -- Therefore, the agent in group $L_0$ will get $14 \cdot \frac{1}{3} \sim 4.67$kW, and the agent in group $L_1$ will get $6.5 \cdot \frac{2}{3} \sim 4.33$kW of electricity.
    
    In a realistic scenario, it is highly unlikely that each group resulting from the geometric grouping contains a single item as the demand of an agent is very small compared to the supply. 
    There will be less different possible values of $g$ (to determine the groups that are connected all the time).
    Hence, keeping this in mind, we are not using the ternary search here. 
    
    \begin{example} \label{ha2:example2}
        First, we consider \Cref{exm:watts-2}.
        To repeat, in that example there are three agents with demand $1, M-1, M$ and supply $M+1$. 
        The geometric grouping of HA2 forms two groups: $L_0 = \{M-1,M\}$, and $L_l = \{1\}$ where $l$ is such that agent with demand $1$ lies in $(M \cdot 2^{-(l+1)}, M \cdot 2^{-l}]$.
        It then applies ``original'' \kbp solution to pack the items in each of the groups.
        
        Note that two bins are needed to pack the two items in group $L_0$. Also, it holds that $2^l \leq M < 2^{l+1}$.
        
        -- Now, there are two possibilities, either no group is connected all the time, or group $L_l$ is connected all the time.\\
        -- In the first case, HA2 computes the final time allocated to each bin in the packing of items in $L_0$ using \Cref{ha2:eqn5} (from that we can easily determine the time for each agent using \Cref{ha2:eqn2}).

        For the sake of simplicity, we assume that $k=1$.
        In the first case, the computation is as follows: using \Cref{ha2:eqn5}, \Cref{ha2:eqn4} and \Cref{ha2:eqn2}, the time allocated to item $M, M-1$ and $1$ is $\frac{1}{2 + 2^l}, \frac{1}{2 + 2^l}$ and $\frac{2^l}{2 + 2^l}$ respectively. Agents $M, M-1, 1$ are allocated $\frac{M}{2 + 2^l}, \frac{M-1}{2 + 2^l}$ and $\frac{2^l}{2 + 2^l}$kW of electricity respectively. The minimum allocation of watts in this case will always be $\frac{2^l}{2 + 2^l} < 1$.
        
        -- In the second case, HA2 determines only the time allocated to each bin in group $L_0$ (from that we can easily determine the time for each agent using \Cref{ha2:eqn2}). Finally, to each bin in $B_0$, HA2 adds the agents in $L_l$.
        
        Using \Cref{ha2:eqn5} and \Cref{ha2:eqn2}, the time allocated to each agent in $\{M, M-1\}$ is $\frac{1}{2}$. Therefore, the watts allocation for agents $M, M-1, 1$ are $\frac{M}{2}, \frac{M-1}{2}, 1$ respectively. It is easy to observe that the minimum allocation of watts in this case is $1$.
        
        -- Finally, from among the above two solutions, 
        one can observe that the allocation of watts in the second case is always leximin-preferred over the allocation of watts in the first case.
    \end{example}
    
    \begin{example} \label{ha2:example1}
        Applying geometric grouping on the same \Cref{ha1:example1} results in $7$ groups $L_0= \{14, 14.2\}, L_1= \{6.5, 6.7\}, L_2= \{3, 3.2\}, L_3= \{1.7, 1.7\}, L_4= \{0.8, 0.82\}, L_5= \{0.4, 0.42\}, L_6= \{0.2, 0.22\}$, two items in each group.
        The groups $L_6, L_5, L_4, L_3$ are connected all the time. 
        
        The remaining bin capacity is $14.74$. Since \kbp is applied to each group with this remaining bin capacity, one can see that the items in $L_1, L_2$ can be packed in a single bin whereas packing the items in $L_0$ requires two bins.
        The time allocated to 
        each bin in group $L_0$ is $0.125$, and the time allocated to group $L_1$ is $0.25$, and the time allocated to group $L_2$ is $0.5$.
        Therefore, the resulting egalitarian allocation is $1.5$ kW.
    \end{example}

    \subsection{Heuristic Algorithm 3 (HA3)}
    \label{heuristic algorithm 3}
   Our next heuristic uses the alternative geometric grouping as in \cite{karmarkar-efficient-1982}. 
    Among all the groups $V_{\max}$ denotes the maximum group size. The idea behind the alternative geometric grouping adopted for our algorithm is as follows:

    First, the input instance $D$ is sorted in a non-increasing order of item sizes. Then, starting from the largest item, items are grouped such that the sum of the item sizes in the group is sufficient to equal or exceed $\mathfrak{u} \cdot d_{\max}$, for some positive $\mathfrak{u}$. In \cite{karmarkar-efficient-1982}, $\mathfrak{u}$ required to be a positive integer; but in our case, $\mathfrak{u}$ can be a positive real number.

    Note that the last group may have a size less than $\mathfrak{u} \cdot d_{\max}$. In this case, we treat the group as if its (pseudo) group size is $\mathfrak{u} \cdot d_{\max}$,
    {the reason being it may be the case that this last group may contain a very small item and may result in the creation of a large number of copies of this group (creating several  copies of the group will be apparent in the later part of this algorithm).}
    Let $E$ be the set of all the groups and $\#E$ the cardinality of this set. We treat each group as if it is a single agent with demand equal to the group size (and equal to the pseudo group size for the last group or alternatively, we only consider the pseudo group size of each group. The pseudo group size of a group can be defined as the actual group size if it is not the last group. For the last group, it is computed as explained earlier).
    Now, we determine the groups that are connected all the time. For that we use the same approach as discussed in \Cref{subsec: heuristic algorithm 1}.
    
    Now, for each remaining group $E_i \in E$ (that are not connected all the time) we compute the number of copies that need to be created for as equally as possible allocation of watts (to the group) as: $\frac{k \cdot V_{\max}}{V(E_i)}$, where $k$ denotes the number of copies of the group with group size ${V}_{\max}$.

    After that we create an instance $D_{group}$ as follows:
    \begin{enumerate}
	   \item for each group $E_i$ in $E$
       (such that the number of remaining copies of $E_i$ is a positive non-zero integer)
       create a copy of an agent whose demand is $V(E_i)$, and insert it into the instance $D_{group}$. Reduce the number of copies of demand $V(E_i)$ by one.
	   \item repeat the above procedure till the number of copies of each demand becomes $0$.
    \end{enumerate}

    Now that we have an instance to work with, we can apply the $FFk$ or $FFDk$ on $D_{group}$. Let $B$ is the resulting solution. Each item $V(E_i)$ in solution $B$ will be replaced by the corresponding group. Also, to each bin in solution $B$ we add the group(s) that are connected all the time. To speed up the process, we use ternary search, and determine the solution that is leximin-preferred over other computed solutions.

    \begin{example} \label{ha3:example2}
        We illustrate the working of HA3 on the same data of \Cref{exm:watts-2}: three agents with demand $M, M-1, 1$ and supply capacity of $M+1$.
        For simplicity, we take $\mathfrak{u}$ as $1$. The alternative geometric grouping in HA3 groups the agents in $\{M, M-1, 1\}$ such that in each group there are sufficient items to have size, $V(\cdot)$, of the group $\geq \mathfrak{u} \cdot d_{\max} = M$. The alternative geometric grouping then results in $\{M\}, \{M-1, 1\}$ groups.
        Now, we determine the groups that can be connected all the time. Note that in this direction there is no group that can be connected all the time because if we choose any of the group to connect all the time, then the remaining supply will be insufficient to pack other groups (HA3 treats each group as a single agent with demand equal to the aggregate demand of agents in the group).
        HA3 then creates $\frac{k \cdot (V_{\max} = M)}{M}$ copies of the group with group size $V(\{M-1, 1\}) = M$ where $k$ is the number of copies of the group with group size $\max \{\{M\}, \{M-1,1\}\}$.
        After that, the algorithm creates the instance $D_{group}$ that contains $k$ copies of the instance $\{M, M\}$.
        Applying $FFk$ on instance $D_{group}$ will place each item in a separate bin. Overall, $2k$ bins are required to pack the items in $D_{group}$, and each item will appear in $k$ bins. Replacing an item in each bin in the final solution with the corresponding group, the final supply vector is $\{0.5, \frac{M-1}{2}, \frac{M}{2} \}$.
    \end{example}

    \begin{example} \label{ha3:example1}
        On applying the alternative geometric grouping on the same \Cref{ha1:example1} with $\mathfrak{u}=0.25$ results in the following groups:
        $\{14.2\}, \{14\}, \{6.7\}, \{6.5\}, \{3.2, 3\}, \{1.7,1.7,0.82\}, \{0.8, 0.42, 0.4, 0.22, 0.2\}$. Note that the items in each group are sufficient to make group size $V$ equal or exceed $\mathfrak{u} \cdot d_{\max} = 3.55$.
        
        The instance on which algorithm works further is $\{14.2, 14, 6.7, 6.5, 6.2, 4.22, 3.55 \}$. Note that the demand of an item in this instance is equal to the group size of the corresponding group. For example, the demand $6.2$ is equal to the group size of the group $\{3.2,3\}$. Note that the last demand $3.55$ is the pseudo group size of the group $\{0.8, 0.42, 0.4, 0.22, 0.2\}$ because the sum of all the item sizes in that group is less than $\mathfrak{u} \cdot d_{\max} = 3.55$. 
        The maximum group size is $V_{\max} = 14.2$.
        
        After that, the algorithms does further grouping
        and determine the number of copies that need to be created for each item in $\{14.2, 14, 6.7, 6.5, 6.2, 4.22, 3.55 \}$.
        We take $k=3$ as the number of copies for demand $14.2$. For the remaining demands, the algorithm determines the number of copies as explained above in the algorithm.
        
        -- For $g=0$, after applying $FFk$ and replacing an item in each bin in the final solution with the corresponding group, the final supply vector is \sloppy
        $\{0.12632,\allowbreak 0.13895, 0.25263, \allowbreak 0.26526, \allowbreak 0.43158, \allowbreak 0.50526, \allowbreak 0.89474, \allowbreak 0.89474, \allowbreak 1.10526, \allowbreak 1.17894, \allowbreak 2.11579, \allowbreak 2.21046, \allowbreak 2.24204, \allowbreak 2.39473 \}$. 
        Egalitarian allocation of supply in this case is $0.12632$. 
        
        -- For $g=1$, the item $3.55$ is connected all the time (and therefore all the items in the group $\{0.2, 0.22, 0.4, 0.42, 0.8\}$ corresponding to this item are connected all the time in the final solution). The algorithm creates the instance $D_{group}$. After applying $FFk$ on instance $D_{group}$ and replacing an item in each bin in the final solution with the corresponding group, the final supply vector (for all the items that are not connected all the time) is 
        $\{0.5125, 1.0625, 1.0625, 1.3125, 1.4, 2.5125, 2.625, 2.6625, 2.84375\}$, and the egalitarian allocation of watts is $0.5125$kW. To each bin in the final solution, the algorithm adds all the items in the group 
        $\{0.2, 0.22, 0.4, 0.42, 0.8\}$ corresponding to item $3.55$.
        
        -- Finally, the algorithm outputs a solution that is leximin preferred over other solutions. In this example, HA3 outputs the solution corresponding to $g=1$ in which all the items in the group 
        $\{0.2, 0.22, 0.4, 0.42, 0.8\}$
        are connected all the time.
    \end{example}

    \subsection{Heuristic Algorithm 4 (HA4)}
    \label{heuristic algorithm 4}
    Our last heuristic algorithm uses the ``original'' \kbp solutions to compute a leximin-preferred solution for an input instance. The algorithm works as follows:
    Let $d_{\max}$ be the maximum demand item in $D$. Now, we determine the item $d_l \in D$ such that the sum of all the item sizes $\leq d_{l}$ is at most $S - d_{\max}$, and adding the next demand ${d'}_{l} > d_l$ 
    violates the bin capacity constraint, that is, the sum of all the item sizes $\leq {d'}_{l}$ is larger than the bin capacity $S - d_{\max}$. 
    Let $G$ be the set of items that are connected all the time as explained and determined in \Cref{chap:equal supply algorithms-sec:impossibility result}, and each item in $G$ has size $\leq d_l$.

    The updated set of demands is $D' = D - G$, and the updated supply is $S' = S - V(G)$. 
    Now, we connect items in $G$ all the time, and for the remaining items, we compute the ``original'' \kbp solution (using either $FFk$ or $FFDk$) and return the solution which is leximin-preferred over other solutions. To make the algorithm run faster, we use Ternary Search
    (see Algorithm \ref{ternary search}). 
    
    The main difference with \Cref{subsec: heuristic algorithm 1} is that, here the number of copies of each item will remain same whereas in \Cref{subsec: heuristic algorithm 1} for an item $D'[i]$ in the remaining set of items the number of copies are computed as $\frac{k \cdot d_{\max}}{D'[i]}$ (where $k$ is the number of copies of item $d_{\max}$) which need not be the same for all the remaining items.

    \begin{algorithm}[H]
    	\SetAlgoLined
    	\DontPrintSemicolon
    	\KwIn{$alg$:= algorithm to use to compute \kbp solution;\\
    		$D$:= the input set of agents' demands; \\
    		$S$:= supply; \\
    		$k$:= as in \kbp;}
    	\KwOut{leximin-preferred best solution;}
    	
    	$D_{sorted}$:= sorted (in non-decreasing order) instance of $D$ \;
    	$d_{\max} = max(D)$ \;
    	determine $d_{large}$ such that sum of all item sizes $\leq d_{large}$ is at most $S - d_{\max}$ and adding the next largest element after $d_{large}$ violates the bin capacity constraint with bin size $S - d_{\max}$ \;
    	$\mathtt{bestSolution}$:= the solution returned by Ternary Search(see Algorithm \ref{ternary search}) \;
    	\caption{Heuristic Algorithm 4}
    	\label{algorithm: heuristic algorithm 4}
    \end{algorithm}

    \begin{example} \label{ha4:example2}
        We illustrate the working of algorithm HA4 on 
        \Cref{exm:watts-2}.
        That example has three agents with demand $M, M-1, 1$, respectively, and the supply capacity is $M+1$.
        The set of agents that can be connected all the time are: $\{\}, \{1\}$. 
        For the set $\{\}$, HA4 applies ``original'' \kbp to the remaining set of agents $\{M, M-1, 1\}$, and computes the egalitarian allocation of watts.
        Then, HA4 connects the agents in $\{1\}$ all the time and for the remaining set of agents $\{M, M-1\}$ applies the ``original'' \kbp solution, and computes the egalitarian allocation of watts. Also, to each bin in the solution, HA4 adds the agents in $\{1\}$. 
        Finally, HA4 outputs among the above two solutions the one that is leximin-preferred.
    \end{example}
    
    \begin{example} \label{ha4:example1}
        On applying HA4 on the same \Cref{ha1:example1}, HA4 connects the agents $\{0.2, 0.22, 0.4 ,0.42, 0.8, 0.82, 1.7, 1.7 \}$ all the time. For the remaining agents, it uses ``original'' \kbp, and results in an egalitarian allocation of $0.81819$kW of electricity as explained below:
   
        -- For $g = g_{begin} =0$, no item is connected all the time. Applying $FFk$ for $k=3$ results in the following watts allocation vector (in non-decreasing order):
        $0.06667, \allowbreak 0.07333, \allowbreak 0.13333, \allowbreak 0.14, \allowbreak 0.26666, \allowbreak 0.27333, \allowbreak 0.56666, \allowbreak 0.56666, \allowbreak 0.99999, \allowbreak 1.06666, \allowbreak 
        2.16664, \allowbreak 2.23331, \allowbreak 4.66662, \allowbreak 4.73329. $
        Egalitarian allocation of watts in this case is 0.06667 kW.   \newline
    -- Similarly, for $g = g_{end} = 8$, items $\{0.2, 0.22, 0.4 ,0.42, 0.8, 0.82, 1.7, 1.7 \}$ are connected all the time. Applying $FFk$ for $k=3$ on the remaining set of elements results in the following allocation of watts: $\{0.81819, 0.87274, 1.77274, 1.82729, 3.81822, 3.87277\}$. The egalitarian allocation of watts in this case is $0.81819$ kW. Clearly, this solution is better than the previous one in terms of the egalitarian allocation of watts among agents that are not connected all the time. In addition, the final allocation vector (that is, after adding the agents that are connected all the time) is preferred to the solution (when $g=0$). 
    
    -- After updating $g_{begin}$ and $g_{end}$ the algorithm computes other solutions. Finally, among these solutions, the one that is preferred by leximin is returned. In this example, the algorithm returns the solution for $g=8$.
    \end{example}

\section{Experimental Results}
\label{sec: Experimental Results}


    \subsection{Experiment: approximation ratio of \texorpdfstring{$FFk$}{FFk} and \texorpdfstring{$FFDk$}{FFDk} algorithms}
    \label{appendix: approximation ratio of FFk and FFDk algorithms}

    In this section we describe the second experiment to verify the approximation ratio of $FFk$ and $FFDk$ with respect to the optimal bin-packing.
    For this purpose we utilize a synthetic dataset. There are datasets available for \kbp when $k=1$, but we do not know their optimal packing when $k>1$. Next, we describe how the synthetic dataset has been created.

    \subsubsection{Data generation}
    Since computing the optimal bin-packing of a given instance is NP-hard, we constructed instances whose optimal bin-packing is known. Given the bin capacity $S$, we use \Cref{appendix:experimentalResults:generate_ll} to generate a list of item sizes that sum up to the bin capacity $S$. \Cref{appendix:experimentalResults:groupItemsGeneration} uses \Cref{appendix:experimentalResults:generate_ll} to generate an instance whose item sizes sum up to $OPT$ times $S$, and by the construction, we know that the optimal bin-packing for the instance $D$ has exactly $OPT$ bins, and moreover, the optimal bin-packing for $D_k$ has exactly $k\cdot OPT$ bins (all of them full).

    \begin{algorithm}[H]
    	\DontPrintSemicolon
    	\KwIn{A bin capacity ${S}$}
    	\KwOut{A list $L$ of item sizes whose sum is $S$.}
    	\SetKw{KwRet}{return}
	
    	Let $L$ be the list of the randomly generated item sizes. Initially $L$ is empty. \;
    	Let $s_L$ be the sum of all the item sizes in the list. Initially $s_L=0$. \;
	
    	\While{True}
    	{generate a random number $r\in (1,S)$. \; 
    		\eIf{$s_L + r > S$}
    		{Append $S - s_L$ to $L$. \;
    			break \;}
    		{Append $r$ to $L$, and update $s_L$. \;}
    	}
    	return $L$
	
    	\caption{generate\_items}
    	\label{appendix:experimentalResults:generate_ll}
    \end{algorithm}

    Our generated dataset contains multiple files. Each file in the dataset contains a fixed number ${IC}$ of instances. Each instance has several item sizes. These items have an optimal packing into ${OPT}$ bins. 
    Each bin has capacity $S$. For example let $[1,2,3,7,6,1']$ be an instance in some file where for each instance in the file ${OPT}=2$. Bin capacity is $S=10$. As we can see, there are two groups $\{7,2,1\}, \{6,3,1'\}$ which sum to 10 and hence an optimal packing requires two bins to pack the items in this instance. 

    Clearly, this dataset has some limitations in that it uses \Cref{appendix:experimentalResults:generate_ll} to generate a list of items with a sum exactly equal to the capacity of the bin $S$. Therefore, the sum of all the items in an instance is a multiple of $S$. However, in a more general scenario, it is not necessarily true.

    \begin{algorithm}
    	\DontPrintSemicolon
    	\KwIn{A bin capacity ${S}$ and the optimal number ${OPT}$ of bins to pack items in this instance.}
    	\KwOut{A list $D$ of item sizes which sum to ${OPT}\cdot{S}$, such that the optimal number of bins is indeed $OPT(D)=OPT$.}
    	
    	Let $D$ be the instance of item sizes. Initially $D$ is empty. \;
    	Let $s_D$ be the sum of all item sizes in $D$. Initially $s_D=0$. \;
    	\While{True}
    	{$L =$ generate\_items($S$) \;
            Append item sizes in $L$ to $D$. \;
    		$s_D = s_D + s_L$. \;
    		\If{$s_D / S == {OPT}$}
    		{
    			Shuffle item sizes in $D$. \;
    			break \;}
    	}
    	return $D$ \;
    	
    	\caption{generate\_instance\_items}
    	\label{appendix:experimentalResults:groupItemsGeneration}
    \end{algorithm}	

    \subsubsection{Results}
    \label{section:results}
    \label{subsection:ffkffdkresults}

    We ran our experiment for different values of $k > 1$, and for different integer values of $OPT$, $2\le OPT \le 9$.
    We computed the conjectured upper bound on the number of bins required to pack items using $FFk$ (for $k>1$) and $FFDk$ as $1.375 \cdot OPT(D_k)$ and $\frac{11 \cdot OPT(D_k) + 6}{9}$, respectively. For each file, we report the maximum number of bins required to pack items. 
    We can see in \Cref{figure:ffk-ffdk:k2k3}, and \Cref{figure:ffk-ffdk:k4k5}
    that the bins used by $FFk$ (for $k>1$) and $FFDk$ are at most $1.375 \cdot OPT(D_k)$ and  $\frac{11 \cdot OPT(D_k) + 6}{9}$, respectively, which supports our conjectures
    in \Cref{section:approx-algo:subsection:approx_algo:ffk} and \Cref{section:approx-algo:subsection:approx_algo:ffdk}.

    \begin{figure}[p]
    	\begin{subfigure}[0.8\textheight]{\linewidth}
    		\centering
    		\includegraphics[height = 0.45\textheight, width=\textwidth]{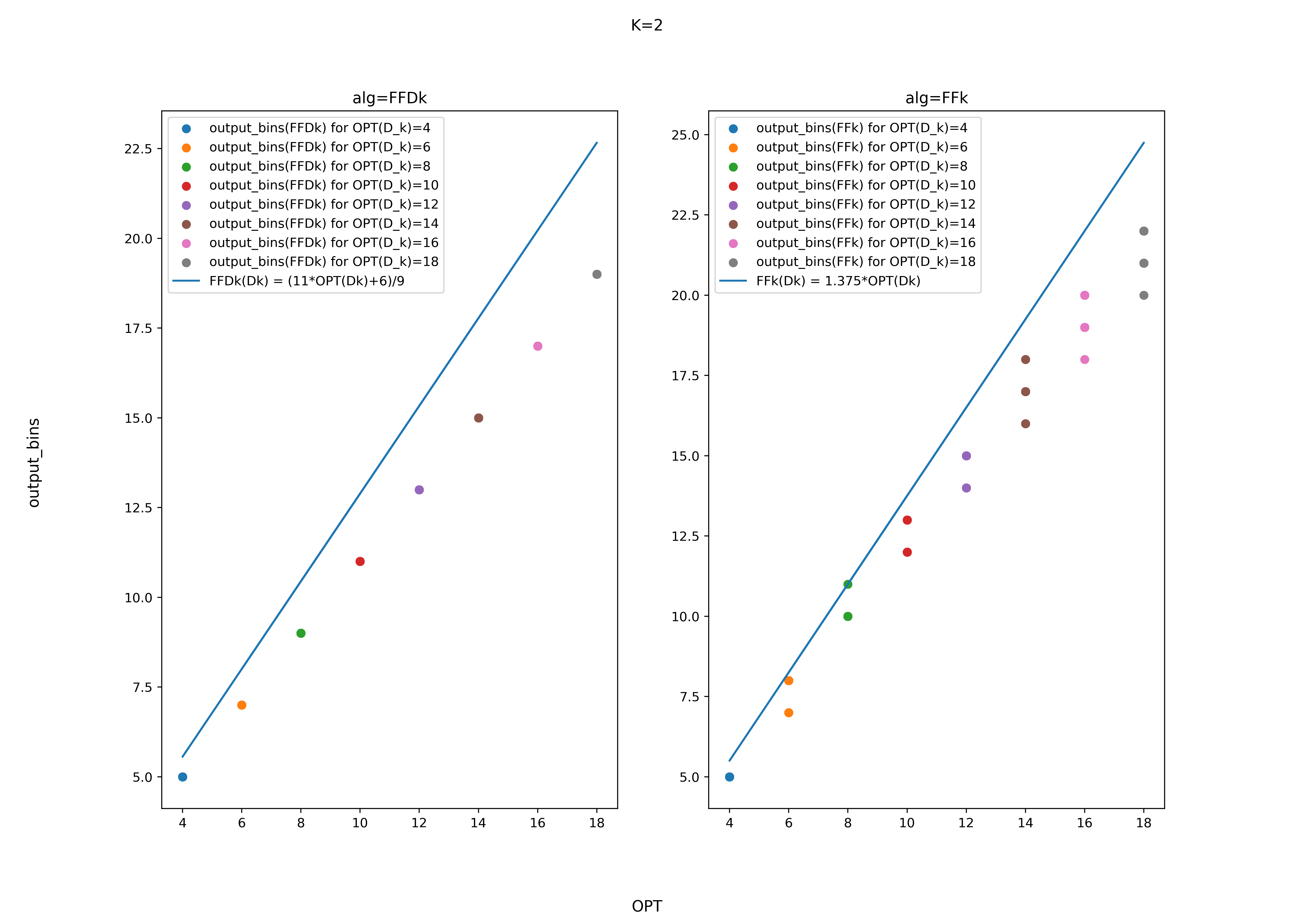}
    	\end{subfigure}
    	
    	\begin{subfigure}[0.8\textheight]{\linewidth}
    		\centering
    		\includegraphics[height = 0.45\textheight, width=\textwidth]{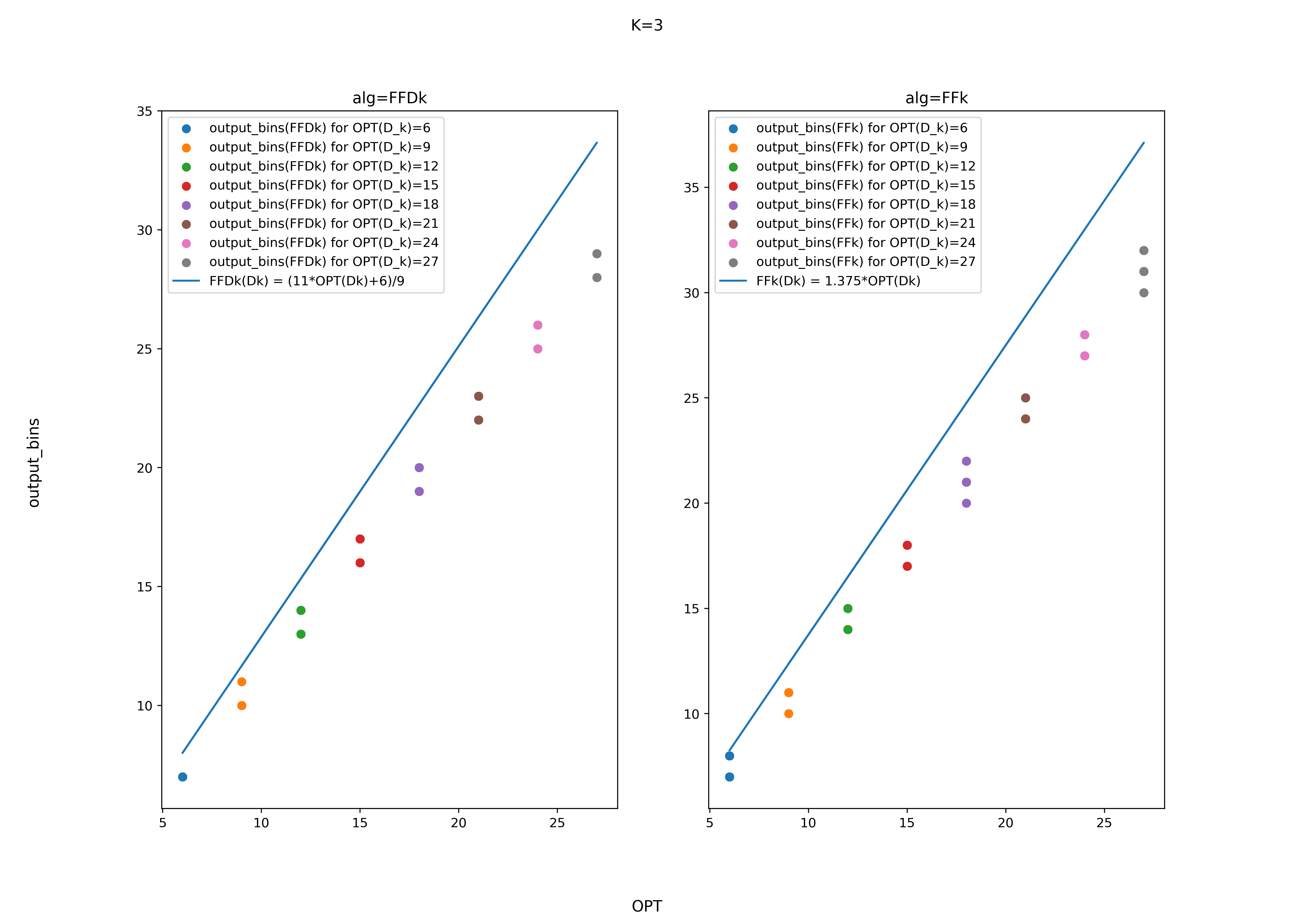}
    	\end{subfigure}		
            \caption{The optimal numbers $OPT(D_k)$ of bins and numbers of bins used by $FFk$ and $FFDk$ algorithms bins are shown at the $x$- and $y$-axis, respectively. The data points in the legend show the upper bound that can be hypothesized for $OPT(D_k) = k \cdot OPT(D)$. Each subplot's data points display the number of different output bins corresponding to each conjectured upper bound. Note that the output bins convey the conjectured upper bounds.}
    	\label{figure:ffk-ffdk:k2k3}		
    \end{figure}

    \begin{figure}[p]
    	\begin{subfigure}[0.8\textheight]{\linewidth}
    		\centering
    		\includegraphics[height = 0.45\textheight, width=\textwidth]{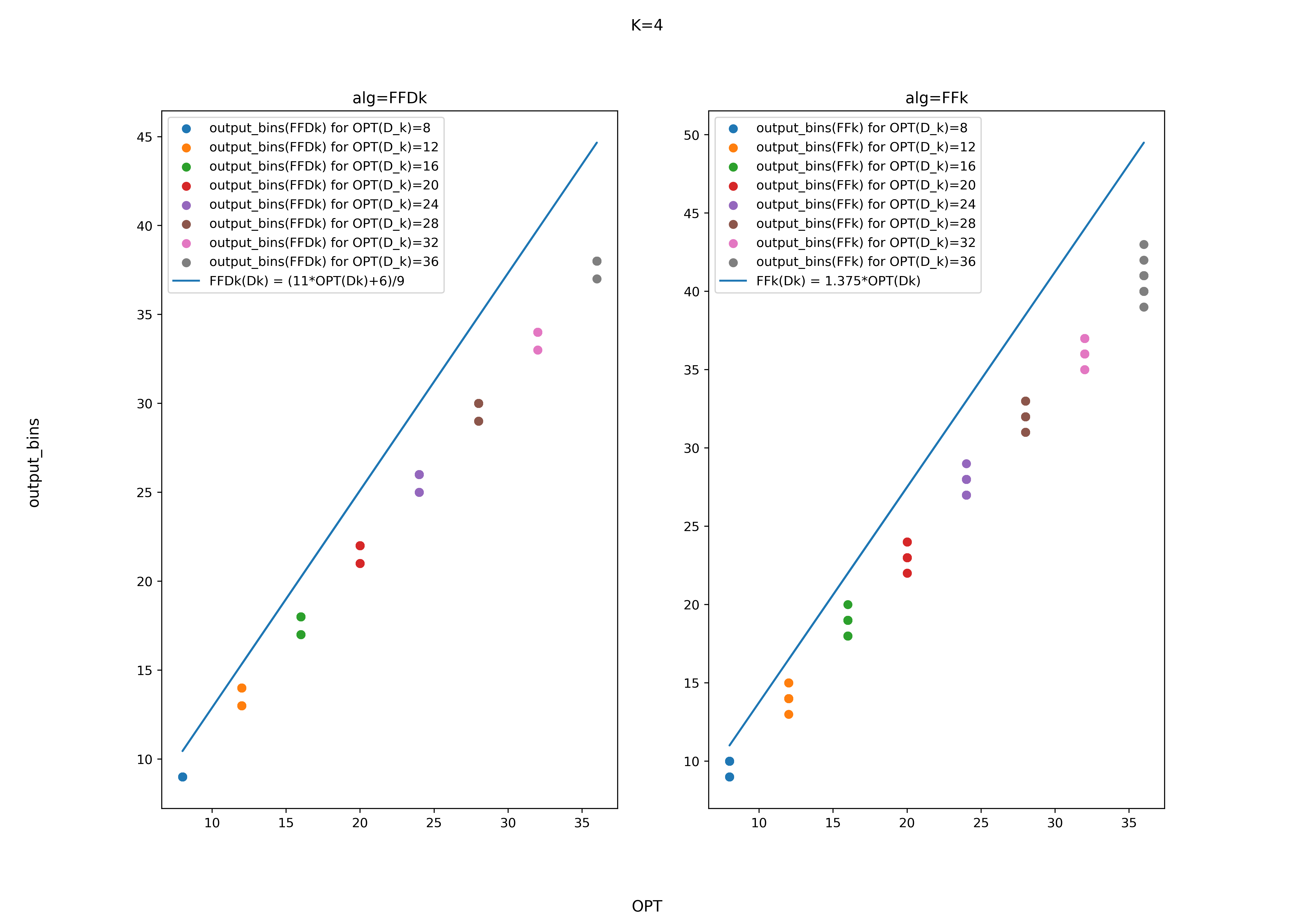}
    	\end{subfigure}
    	
    	\begin{subfigure}[0.8\textheight]{\linewidth}
    		\centering
    		\includegraphics[height = 0.45\textheight, width=\textwidth]{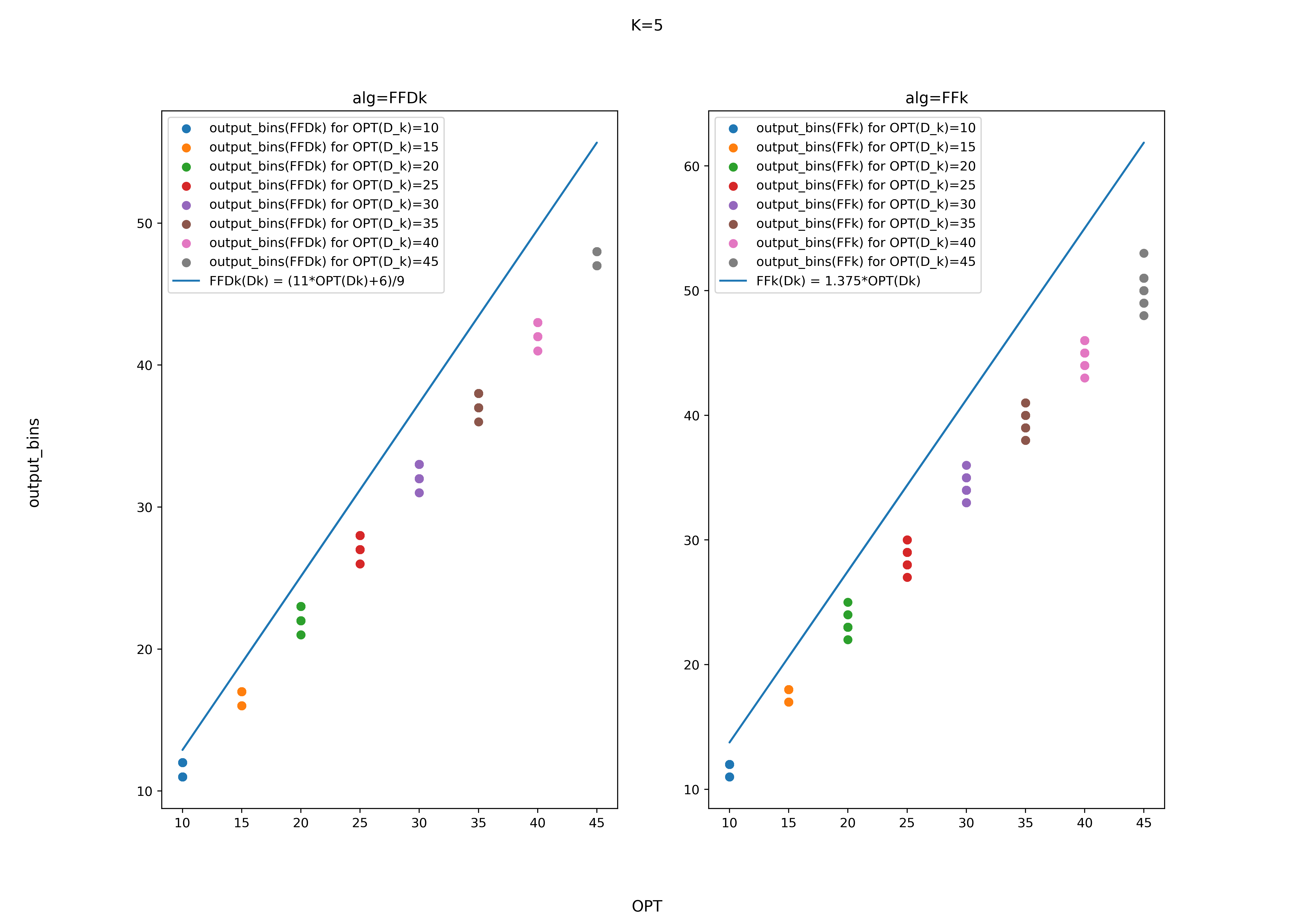}
    	\end{subfigure}		
            \caption{The optimal numbers $OPT(D_k)$ of bins and numbers of bins used by $FFk$ and $FFDk$ algorithms bins are shown at the $x$- and $y$-axis, respectively. The data points in the legend show the upper bound that can be hypothesized for $OPT(D_k) = k \cdot OPT(D)$. Each subplot's data points display the number of different output bins corresponding to each conjectured upper bound. Note that the output bins convey the conjectured upper bounds.}
    	\label{figure:ffk-ffdk:k4k5}		
    \end{figure}    


\subsection{Experiment: $FFk$ and $FFDk$ for egalitarian allocation of connection time}
\label{subsec: results-egal-connection-time}
In this section, we describe experimental results checking the performance of the $k$BP adaptations of $FFk$ and $FFDk$ to our motivating application of fair electricity distribution \footnote{All the codes are available in \url{https://github.com/dinkubag/Electricity-Distribution-Algorithms}. }.
\label{experiment:fairelecdistri}

\subsubsection{Dataset}
\label{experiment:fairelecdistri:dataset}

We use the same dataset of $367$ Nigerian households described in \cite{oluwasuji_solving_2020}.%
\footnote{
	We are grateful to Olabambo Oluwasuji for sharing the dataset with us.
}
This dataset contains the hourly electricity demand for each household for 13 weeks (2184 hours). 
In addition, they estimate for each agent and hour, the \emph{comfort} of that agent, which is an estimation of the utility the agent gets from being connected to electricity at that hour.
For more details about the dataset, readers are encouraged to refer to the papers \cite{oluwasuji_algorithms_2018,oluwasuji_solving_2020}. 
The electricity demand of agents can vary from hour to hour. We execute our algorithms for each hour separately, which gives us essentially 2184 different instances.

As in \cite{oluwasuji_solving_2020}, we use the demand figures in the dataset as mean values; we determine the actual demand of each agent at random from a normal distribution with a standard deviation of $0.05$ (results with a higher standard deviation are presented in \Cref{APPENDIX: electricity-distribution-results}).

As in \cite{oluwasuji_solving_2020}, 
we compute the supply capacity $S$ for each day by averaging the hourly estimates of  agents' demand for that day. 
We run nine independent simulations (with different randomization of agents' demands). Thus, the supply changes in accordance with the average daily demand, but cannot satisfy the maximum hourly demand.

\subsubsection{Experiment:}
\label{experiment:fairelecdistri:experiment}
For each hour, we execute the $FFk$ and $FFDk$ algorithms on the households' demands for that hour. We then use the resulting packing to allocate electricity: if the packing returns $q$ bins, then each bin is connected for $1/q$ of an hour, which means that each agent is connected for $k/q$ of an hour.

The authors of \cite{oluwasuji_solving_2020} measure the efficiency and fairness of the resulting allocation, not only by the total time each agent is connected, but also by more complex measures. In particular, they assume that each agent $i$ has a utility function, denoted $u_i$, that determines the utility that the agent receives from being connected to electricity at a given hour.
They consider three different utility models:
\begin{enumerate}
	\item The simplest model is that $u_i$ equals the amount of time the agent $i$ is connected to electricity (this is the model we mentioned in the introduction).
	\item The value $u_i$ can also be equal to the total amount of electricity that the agent $i$ receives. For each hour, the amount of electricity given to $i$ is the amount of time $i$ is connected, times $i$'s demand at that hour.
	\item They also measure the ``comfort'' of the agent $i$ in time $t$ by averaging their demand over the same hour in the past four weeks, and normalizing it by dividing by the maximum value.
\end{enumerate}

For each utility model, they consider three measures
of efficiency and fairness:
\begin{itemize}
	\item Utilitarian: the sum $
	\sum_{i} u_i(x)$ (or the average) of all agents' utilities $u_i$.
	\item Egalitarian: the minimum utility $
	\min_{i} u_i(x)	$ of a single agent,
	\item The maximum difference $\max_{i,j} \{\lvert u_i(x) - u_j(x)\rvert\}$ of utilities between each pair of agents.
\end{itemize}

\subsubsection{Results:}
\label{subsection:implementationresult}
The authors of \cite{oluwasuji_solving_2020} have proposed two models: comfort model (CM) and the supply model. The objective of the CM and the SM model is to maximize the comfort and supply respectively.
For the comparison with the results from \cite{oluwasuji_solving_2020},
we show our results for $FFk$ and $FFDk$ for $k=100$ \footnote{We have checked smaller values of $k$, and found out that the performance increases with $k$. By the time $k$ reached $100$, the performance increase was very slow, so we kept this value.} along with the results in \cite{oluwasuji_solving_2020} in tables \footnote{In \Cref{results:table:connections-to-supply}-\Cref{results:table:comfort-delivered} CM, SM, GA, CSA1, RSA, CSA2 stands for: The Comfort Model, The Supply Model, Grouper Algorithm, Consumption-Sorter Algorithm, Random-Selector Algorithm, Cost-Sorter Algorithm respectively \cite{oluwasuji_solving_2020,oluwasuji_algorithms_2018}} \Cref{results:table:connections-to-supply},\Cref{results:table:electricity-supplied}, and \Cref{results:table:comfort-delivered} . We highlight the best results in bold.
As can be seen in \Cref{results:table:connections-to-supply}, 
\Cref{results:table:electricity-supplied} and 
\Cref{results:table:comfort-delivered},
$FFk$ and $FFDk$ outperform previous results in terms of the egalitarian allocation of connection time which is the main objective of this paper. Overall, the comparison of $FFk$ and $FFDk$ with their results are as follows:
\begin{itemize}
    \item[--] $FFk$ and $FFDk$ outperform the previous results in terms of egalitarian allocation of connection time. Maximum Utility difference is also better than all previous results. In terms of utilitarian social welfare metric $FFk$ and $FFDk$ outperforms all previous results except CM.

    \item[--] $FFk$ and $FFDk$ outperform the previous results in terms of utilitarian allocation of supply. In terms of egalitarian social welfare metric and maximum utility difference, $FFk$ and $FFDk$ outperforms all previous results except SM.

    \item[--] $FFk$ and $FFDk$ outperform the previous results in terms of utilitarian allocation of comfort except the CM model. In terms of egalitarian allocation of comfort, $FFk$ and $FFDk$ performs better than previous results except the CM and SM model. Their performance is nearly equivalent to the CM model with better standard deviation and maximum utility difference.
\end{itemize}

In \Cref{APPENDIX: electricity-distribution-results}, we have graphs that show how various values of $k$ and varying levels of uncertainty affect the number of connection hours, amount of electricity delivered, and comfort.

\begin{toappendix}
    \section{More Electricity Distribution Results}	\label{APPENDIX: electricity-distribution-results}
    In this section, we discuss the variation in hours of connection, comfort, and electricity supplied with different values of $k$ and with varying levels of uncertainty (standard deviation) in the agent's demand. 

    In \Cref{appendix:more-results:k-vs-comf-util-with-delta}, we can observe that the sum of comfort the algorithm delivers to agents increase with an increasing value of $k$. However, that increase appears to saturate as $k$ increases. Similar phenomena occur for other metrics.

    \Cref{appendix:more-results:k-vs-comf-egal-with-delta} shows that the minimum comfort the algorithm delivers to agents individually increases with an increasing value of $k$. However, that increase appears to saturate as $k$ increases. Similar phenomena occur for other metrics.

    \begin{figure}[h]
    \centering\includegraphics[width=\linewidth, height=0.45\textheight]{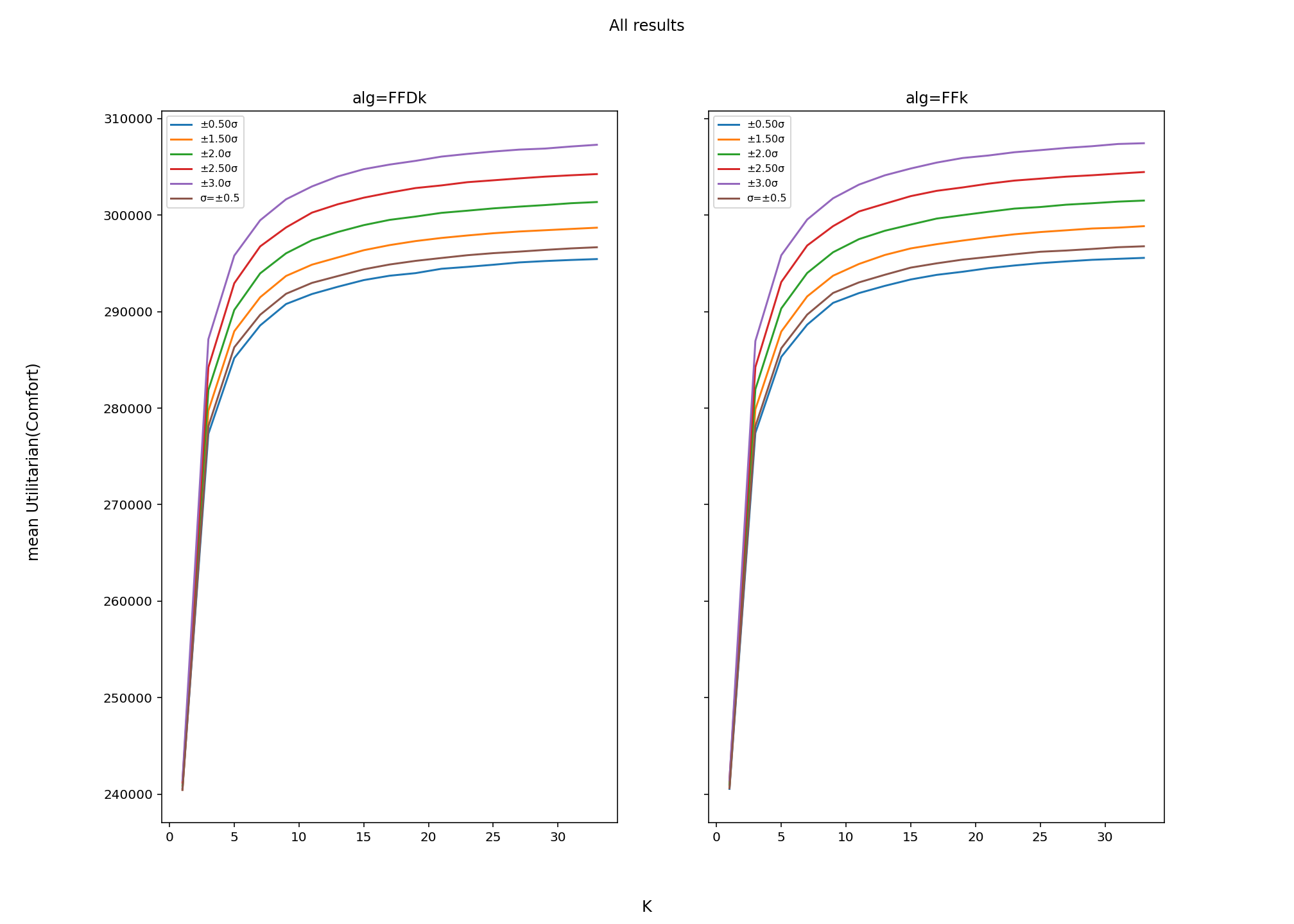}

    \caption{Values of $k$ and mean utilitarian value of comfort are shown at the $x$- and $y$-axis, respectively. Figure shows the increase in sum of comfort delivered by the algorithm to agents with increasing value of $k$ for varying levels of uncertainty $\sigma$. 
    \label{appendix:more-results:k-vs-comf-util-with-delta}
    }

    
    \centering\includegraphics[width=\linewidth, height=0.45\textheight]{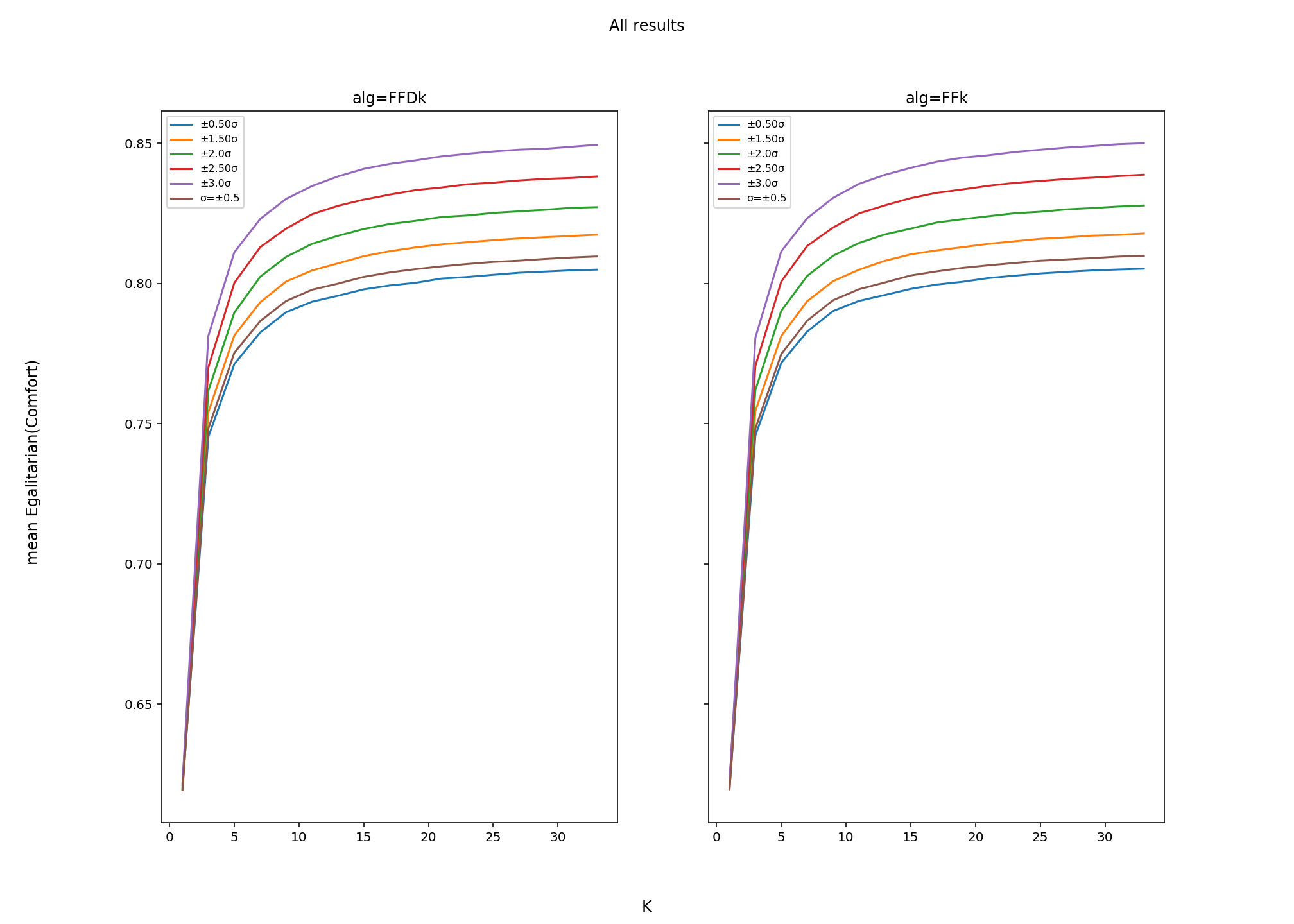}
    \caption{Values of $k$ and minimum egalitarian value of comfort are shown at the $x$- and $y$-axis, respectively. Figure shows the increase in minimum comfort delivered to agents individually with increasing value of $k$ for varying levels of uncertainty $\sigma$. 
    \label{appendix:more-results:k-vs-comf-egal-with-delta}
    }
    \end{figure}

    In \Cref{appendix:more-results:k-vs-comf-mud-with-delta}, the maximum utility difference decreases with an increasing value of $k$. 
    Similar phenomena occurs for electricity-supplied social welfare metrics. However, in the case of hours of connection, the maximum utility difference is 0 because the algorithm is executed for each hour.
    
    \begin{figure}[h]
    \centering\includegraphics[width=\linewidth, height=0.45\textheight]{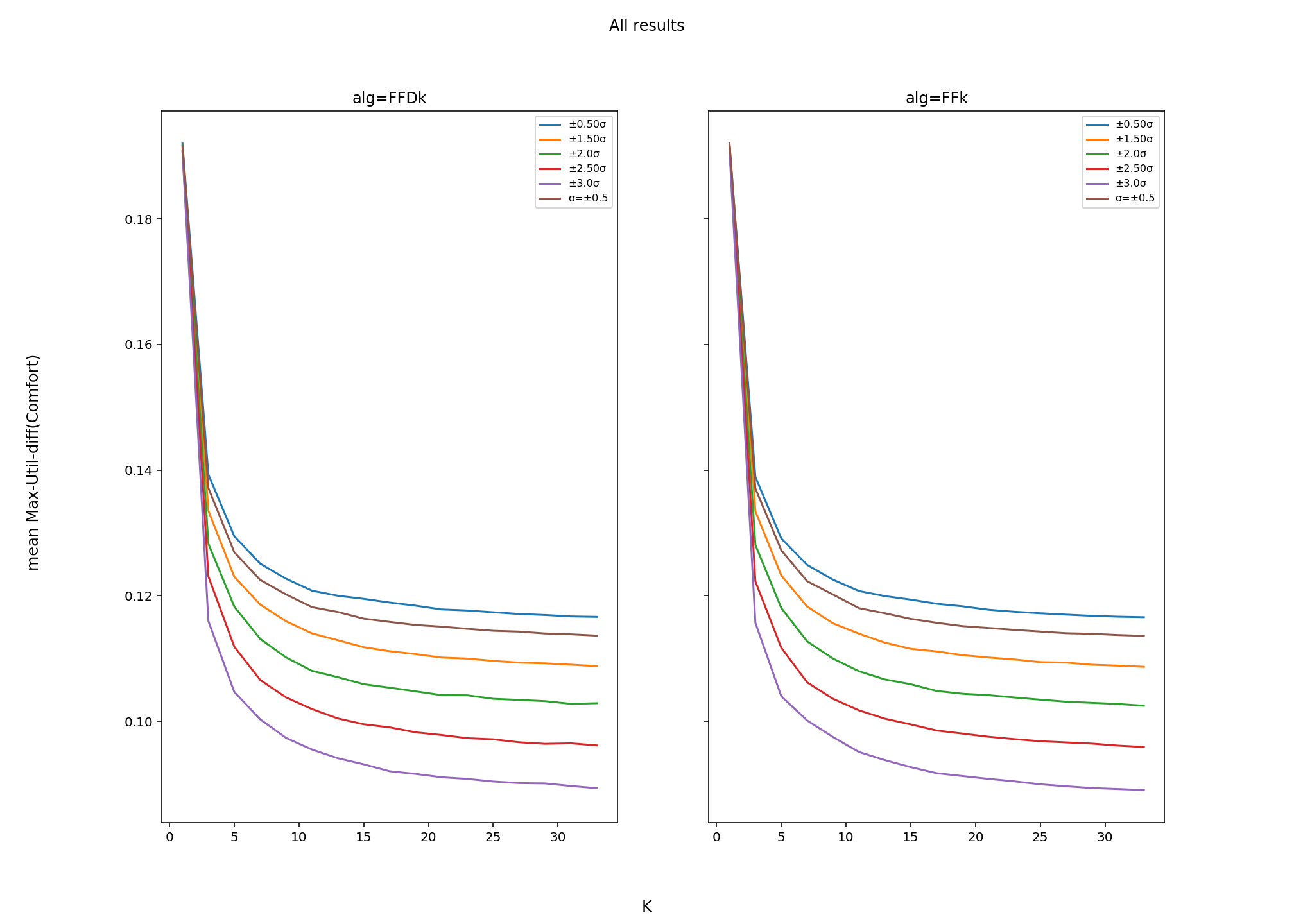}
    \caption{Values of $k$ and maximum utility difference of comfort are shown at the $x$- and $y$-axis, respectively. Figure shows the decrease in maximum utility difference with increasing value of $k$ for varying levels of uncertainty $\sigma$. }
    \label{appendix:more-results:k-vs-comf-mud-with-delta}
    \end{figure}
\end{toappendix}

\begin{table}[!ht]
	\caption{Comparing results of $FFk$ and $FFDk$ for $k=100$ with the results in \cite{oluwasuji_solving_2020} in terms of hours of connection to supply on the average, along with their standard deviation (SD) within parenthesis. In the third column, we have shown the average number of hours an agent is connected to the supply. }
	\label{results:table:connections-to-supply}
	\begin{center}
                    \scalebox{1}{
				\begin{tabular}{|p{0.15\linewidth}|p{0.2\linewidth}|p{0.15\linewidth}|p{0.2\linewidth}|p{0.15\linewidth}|}
					\toprule
					Algorithm & Utilitarian: sum(SD) & Utilitarian: average & Egalitarian(SD) & Maximum Utility Difference \\
					\toprule
					$FFk$ & {716145.3847 (13.9574)} & {1951.3498} & \textbf{1951.3498 (0.0380)} & \textbf{0.0(0.0)} \\
					\toprule
					$FFDk$ & {715891.3137 (13.3774)} & {1950.6575} &  \textbf{1950.6575 (0.0364)} & \textbf{0.0(0.0)} \\
					\toprule
					CM & 717031(3950) & 1953.7629 & 1920(3.24) & 123(2.09) \\
					\toprule
					SM & 709676(3878) & 1933.7221 & 1922(3.41) & 71(2.04) \\
					\toprule
					GA & 629534(4178) & 1715.3515 & 1609(4.69) & 695(3.28) \\
					\toprule
					CSA1 & 647439(3063) & 1764.1389 & 1764(2.27) & 1(0.00) \\
					\toprule
					RSA & 643504(4094) & 1753.4169 & 1753(4.33) & 1(0.00) \\
					\toprule
					CSA2 & 641002(3154) & 1746.5995 & 1746(2.38) & 1(0.00) \\
					\bottomrule
				\end{tabular}}
		\end{center}	
	\end{table}
	
	\begin{table}[!ht]
		\caption{Comparing results of $FFk$ and $FFDk$ for $k=100$ with the results in \cite{oluwasuji_solving_2020} in terms of electricity supplied on the average, along with their standard deviation (SD) within parenthesis.}
		\label{results:table:electricity-supplied}
		\begin{center}
                    \scalebox{1}{
				\begin{tabular}{|p{0.15\linewidth}|p{0.25\linewidth}|p{0.25\linewidth}|p{0.25\linewidth}|}
						\toprule
						Algorithm & Utilitarian(SD) & Egalitarian(SD) & Maximum Utility Difference \\
						\toprule
						$FFk$ & \textbf{1364150.4034 (58.6228)} & {0.8067 (0.0001)} & {0.1183 (0.0001)} \\
						\toprule
						$FFDk$ & \textbf{1363494.0885 (48.1455)} & {0.8063 (0.0001)} & {0.1183 (0.0002)} \\
						\toprule
						CM & 1340015(8299) & 0.78(0.01) & 0.17(0.02) \\
						\toprule
						SM & 1347801(8304) & 0.83(0.01) & 0.11(0.02) \\
						\toprule
						GA & 1297020(11264) & 0.35(0.04) & 0.58(0.03) \\
						\toprule
						CSA1 & 1296939(7564) & 0.66(0.02) & 0.28(0.02) \\
						\toprule
						RSA & 1344945(11284) & 0.68(0.03) & 0.25(0.03) \\
						\toprule
						CSA2 & 1345537(7388) & 0.63(0.03) & 0.30(0.02) \\
						\bottomrule
					\end{tabular}}
			\end{center}	
		\end{table}

		\begin{table}[!ht]
			\caption{Comparing results of $FFk$ and $FFDk$ for $k=100$ with the results in \cite{oluwasuji_solving_2020} in terms of comfort delivered on the average, along with their standard deviation (SD) within parenthesis.}
			\label{results:table:comfort-delivered}
			\begin{center}
                        \scalebox{1}{
						\begin{tabular}{|p{0.15\linewidth}|p{0.25\linewidth}|p{0.25\linewidth}|p{0.25\linewidth}|}
							\toprule
							Algorithm & Utilitarian(SD) & Egalitarian(SD) & Maximum Utility Difference \\
							\toprule
							$FFk$ & {296630.3583 (7.3626)} & {0.8085 (0.00003)} & {0.1155 (0.00004)} \\
							\toprule
							$FFDk$ & {296493.8100 (7.5569)} & {0.8081 (0.00003)} & {0.1156 (0.00003)} \\
							\toprule
							CM & 303217(3447) & 0.81(0.01) & 0.13(0.02) \\
							\toprule
							SM & 292135(3802) & 0.83(0.01) & 0.09(0.02) \\
							\toprule
							GA & 291021(5198) & 0.38(0.04) & 0.56(0.03) \\
							\toprule
							CSA1 & 291909(3201) & 0.67(0.02) & 0.25(0.02) \\
							\toprule
							RSA & 268564(5106) & 0.65(0.04) & 0.28(0.03) \\
							\toprule
							CSA2 & 270262(3112) & 0.64(0.02) & 0.28(0.02) \\
							\bottomrule
						\end{tabular}}
				\end{center}	
			\end{table}	
%
In {\Cref{APPENDIX: electricity-distribution-results}}, we discuss the variation in utilitarian, egalitarian, and maximum-utility difference with different values of $k$ and varying levels of uncertainty (standard deviation). The graphs show that the changes appear to saturate as $k$ increases.

\subsection{Experiment: Heuristic algorithms for egalitarian allocation of watts}
\label{subsec: Experiment: Heuristic algorithms for egalitarian allocation of watts}

In this section, we evaluate the performance of our heuristic algorithms (HA1, HA2, HA3, HA4) developed to distribute watts as equally as possible (see \Cref{sec: Egalitarian allocation of supply}). Before comparing the results, we want to highlight some key points that are relevant in interpreting the results of \Cref{results:table:heuristic algorithms:supply table}:
\begin{itemize}
	\item[--] The computation of supply is dependent on the data. Equal watts allocation computation is necessary for hours during which the aggregate demand is greater than the supply. Therefore, for the hours when aggregate demand is less than supply, we do not consider the allocation of watts in computing the egalitarian allocation of watts and in determining the maximum utility difference. However, they are considered when determining the utilitarian allocation.
    \item[--] 
    For each hour when the aggregate demand is greater than the supply, we compute the egalitarian allocation of watts for the agents that are not connected all the time.
    The final egalitarian allocation is the sum of the egalitarian allocations computed for each such hour.
\end{itemize}
We compare our results in terms of the number of hours an agent is connected to supply, the total supply delivered to an agent, and the total comfort delivered to an agent. We make this comparison in terms of the social welfare metrics of utilitarian, egalitarian, and maximum utility difference.
We do not compare our results with the results in \cite{oluwasuji_algorithms_2018,oluwasuji_solving_2020}, because in their article they have not reported the results in terms of supply allocation.

We highlight the best results in bold. Overall, the comparison of heuristic algorithms (HA1 to HA4) is as follows:
\begin{itemize}
	\item[--] Our main objective in developing heuristics is to allocate watts to agents as equally as possible. \Cref{results:table:heuristic algorithms:supply table} shows that HA1 with $FFDk$ outperforms all other heuristics in supplying more electricity to the worst-off agent. HA1 with $FFk$ outperforms other algorithms in distributing electricity as equally as possible (see the maximum utility difference column).
	On the other hand, HA4 with $FFk$ outperforms other algorithms in terms of total watts delivered.
    More results are discussed in \Cref{subsec: Experiment: Heuristic algorithms for egalitarian allocation of supply}.
	
	\item[--] Additionally, we compare our results in terms of hours of connection to supply (see \Cref{results:table:heuristic algorithms:connection table} in \Cref{subsec: Experiment: Heuristic algorithms for egalitarian allocation of supply}). 
    In this case, HA4 with $FFDk$ connects the worst-off agent to more number of hours than to any other heuristic algorithms. It also has the minimum ``maximum utility difference.''
	On the other hand, algorithm HA4 with $FFK$ outperforms other algorithms in total connection time delivered.
	
	\item[--] When we compare the results in terms of the comfort delivered on average (see \Cref{results:table:heuristic algorithms:comfort table} in \Cref{subsec: Experiment: Heuristic algorithms for egalitarian allocation of supply}), HA4 with $FFDk$ performs better than other algorithms in delivering comfort to the worst-off agent. It also has the minimum ``maximum utility difference,'' which means it delivers comfort more uniformly among the agents. 
	On the other hand, algorithm HA4 with $FFk$ outperforms other algorithms in terms of total comfort delivered.
\end{itemize}

\begin{table}[!ht]
	\caption{Comparing results of heuristic algorithms HA1-HA4 in terms of electricity allocated in watts on average, along with their standard deviation (SD) within parenthesis. }
	\label{results:table:heuristic algorithms:supply table}
	\begin{center}
		\scalebox{0.9}{
			\begin{tabular}{|p{0.2\linewidth}|p{0.26\linewidth}|p{0.22\linewidth}|p{0.25\linewidth}|p{0.18\linewidth}|}
				\toprule
				Algorithm & Utilitarian: sum(SD) &  Egalitarian(SD) & Maximum Utility Difference \\
				\toprule
				HA-1($k=5$, alg=$FFk$) & 1319476.193(202.97) & 2405.297(10.305) & $\mathbf{187.591(0.622)}$ \\
				\toprule
				HA-1($k=5$, alg=$FFDk$) & 1319219.591(27.98) & \textbf{2412.67(4.026)} & 187.611(0.115) \\
				\toprule
				HA-2($k=50$, alg=$FFk$) & 1305321.696(512.344) & 1837.016(4.518) & 1773.596(3.978) \\
				\toprule
				HA-2($k=50$, alg=$FFDk$) & 1305310.397(247.794) & 1833.378(8.039) & 1769.877(5.584) \\
				\toprule
				HA-3($k=5$, alg=$FFk$) & 1333406.016(555.198) & 1088.202(2.844) & 3706.019(26.488) \\
				\toprule
				HA-3($k=5$, alg=$FFDk$) & 1332752.563(336.110) & 1088.203(4.158) & 3718.627(23.409) \\
				\toprule
				HA-4($k=50$, alg=$FFk$) & $\mathbf{1363677.72(41.934)}$ & 1886.345(5.524) & 11734.702(18.849) \\
				\toprule
				HA-4($k=50$, alg=$FFDk$) & 1363650.312(39.603) & 1887.696(3.423) & 11730.615(24.72) \\
				\bottomrule
			\end{tabular}
		}
	\end{center}	
\end{table}

\begin{toappendix}
    \subsection{Experiment: Heuristic algorithms for egalitarian allocation of watts}
    \label{subsec: Experiment: Heuristic algorithms for egalitarian allocation of supply}

    \Cref{fig:max-supply-difference vs k for diff delta} shows that maximum watts difference decreases with increasing value of $k$ and with varying levels of uncertainty (standard deviation) in the agent's demand.
    The below \Cref{fig:max-supply-difference vs k for diff delta} is newly added.
    \begin{figure}
        \centering
        \includegraphics[width=\linewidth, keepaspectratio]{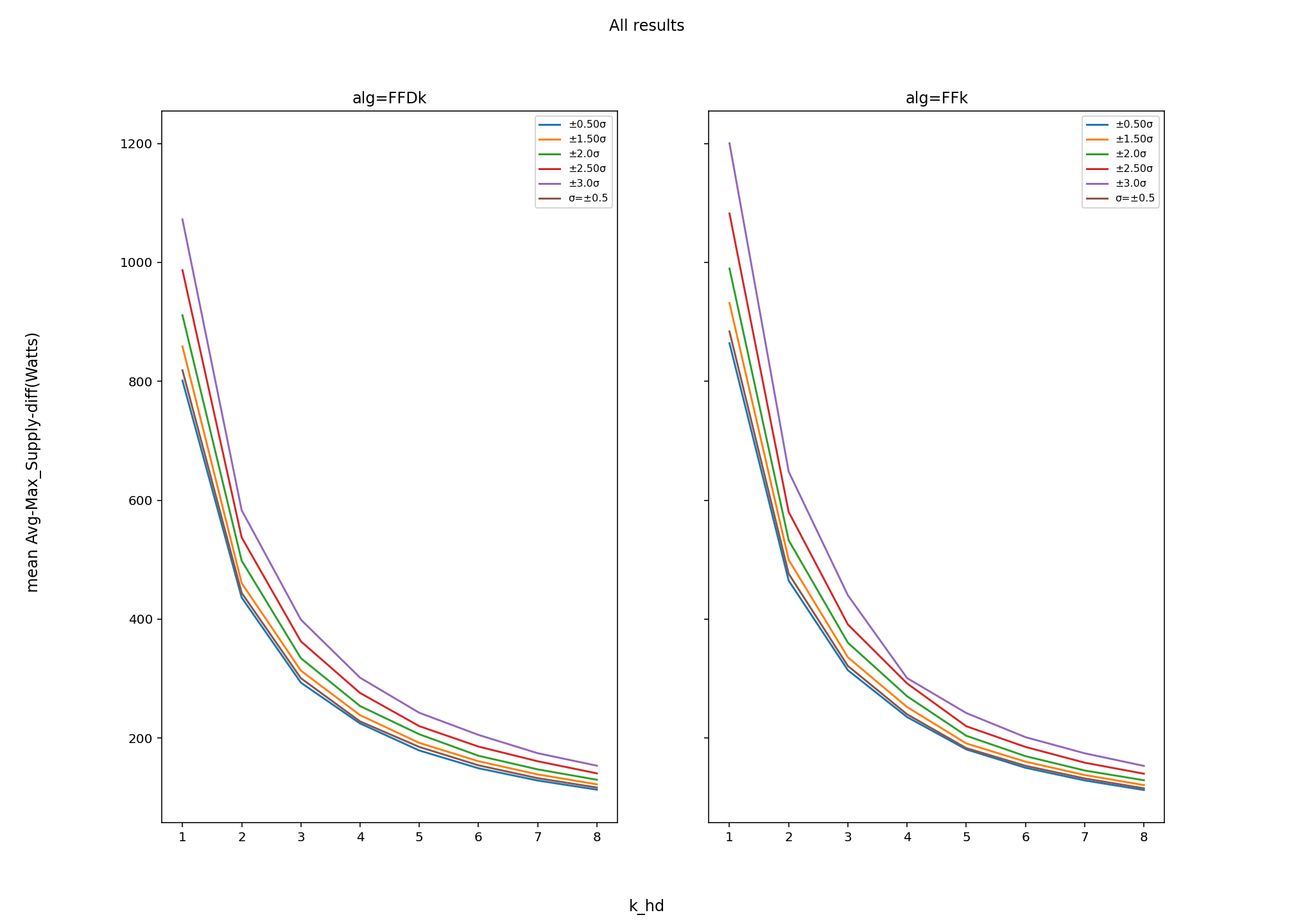}
        \caption{Values of $k$ and maximum watts difference are shown at the $x-$ and $y-$ axis, respectively. Figures shows that the maximum supply difference decreases with increasing value of $k$ for varying levels of uncertainty $\sigma$.}
        \label{fig:max-supply-difference vs k for diff delta}
    \end{figure}

    \begin{table}[!ht]
	\caption{Comparing results of heuristic algorithms HA1-HA4 in terms of hours of connection to supply on the average, along with their standard deviation (SD) within parenthesis. In the third column, we have shown the average number of hours an agent is connected to the supply. }
	\label{results:table:heuristic algorithms:connection table}
	\begin{center}
		\scalebox{0.9}{
			\begin{tabular}{|p{0.2\linewidth}|p{0.26\linewidth}|p{0.2\linewidth}|p{0.25\linewidth}|p{0.18\linewidth}|}
				\toprule
				Algorithm & Utilitarian: sum(SD) & Egalitarian(SD) & Maximum Utility Difference \\
				\toprule
				HA-1($k=5$, alg=$FFk$) & 750517.7(47.39) & 1587.06(0.79) & 596.938(0.796) \\
				\toprule
				HA-1($k=5$, alg=$FFk$) & 750441.267(27.508) & 1587.1(0.575) & 596.903(0.574) \\
				\toprule
				HA-2($k=50$, alg=$FFk$) & 743036.794(124.97) & 16454.95(2.424) & 538.047(2.434) \\
				\toprule
				HA-2($k=50$, alg=$FFDk$) & 743023.081(64.725) & 1647.564(1.116) & 536.436(1.116) \\
				\toprule
				HA-3($k=5$, alg=$FFk$) & 735908.209(234.535) & 1770.334(1.744) & 413.666(1.744) \\
				\toprule
				HA-3($k=5$, alg=$FFDk$) & 735497.902(148.749) & 1771.869(1.415) & 412.131(1.415) \\
				\toprule
				HA-4($k=50$, alg=$FFk$) & $\mathbf{752765.642(41.830)}$ & 1828.394(0.941) & 355.606(0.941) \\
				\toprule
				HA-4($k=50$, alg=$FFDk$) & 752748.122(10.153) & \textbf{1828.61(0.654)} & \textbf{355.39(0.654)} \\
				\bottomrule
			\end{tabular}
		}
	\end{center}	
\end{table}

\begin{table}[!ht]
	\caption{Comparing results of heuristic algorithms HA1-HA4 in terms of comfort delivered on average, along with their standard deviation (SD) within parenthesis.}
	\label{results:table:heuristic algorithms:comfort table}
	\begin{center}
		\scalebox{0.9}{
			\begin{tabular}{|p{0.2\linewidth}|p{0.26\linewidth}|p{0.2\linewidth}|p{0.25\linewidth}|p{0.18\linewidth}|}
				\toprule
				Algorithm & Utilitarian: sum(SD) &  Egalitarian(SD) & Maximum Utility Difference \\
				\toprule
				HA-1($k=5$, alg=$FFk$) & 310997.434(27.382) & 0.629(0.0001) & 0.371(0.0001) \\
				\toprule
				HA-1($k=5$, alg=$FFDk$) & 310942.697(18.246) & 0.629(0.0002) & 0.371(0.0002) \\
				\toprule
				HA-2($k=50$, alg=$FFk$) & 306235.496(83.102) & 0.662(0.001) & 0.338(0.001) \\
				\toprule
				HA-2($k=50$, alg=$FFDk$) & 306225.401(39.784) & 0.663(0.0003) & 0.337(0.0003) \\
				\toprule
				HA-3($k=5$, alg=$FFk$) & 303138.896(145.748) & 0.732(0.0007) & 0.267(0.0007) \\
				\toprule
				HA-3($k=5$, alg=$FFDk$) & 302895.789(82.562) & 0.734(0.0012) & 0.266(0.0012) \\
				\toprule
				HA-4($k=50$, alg=$FFk$) & $\mathbf{313024.588(30.674)}$ & 0.7686(0.0001) & 0.231(0.0001) \\
				\toprule
				HA-4($k=50$, alg=$FFDk$) & 313015.52(9.495) & \textbf{0.7688(0.0001)} & \textbf{0.231(0.0001)} \\
				\bottomrule
			\end{tabular}
		}
	\end{center}	
\end{table}
        
\end{toappendix}

\section{Conclusion and Future Directions} \label{section:conclusion}
We have shown that the existing approximation algorithms, like the First-Fit and the First-Fit Decreasing, can be extended to solve $k$BP. We have proved that, for any $k\geq 1$, the asymptotic approximation ratio for the $FFk$ algorithm is $\left(1.5+\frac{1}{5k}\right) \cdot OPT(D_k) + 3\cdot k $.
We have also proved that the asymptotic approximation ratio for the $NFk$ algorithm is $2$.
We have also demonstrated that the generalization of efficient approximation algorithms like Fernandez de la Vega-Lueker and Karmarkar Karp algorithms solves $k$BP in $(1 + 2 \cdot \epsilon)\cdot OPT(D_k) + k$ and $OPT(D_k) + O(k \cdot \log^2 {OPT(D)})$ bins respectively in polynomial time.
We have also shown the practical efficacy of $FFk$ and $FFDk$ in solving the fair electricity distribution problem. 

Given the usefulness of $k$-times bin-packing to electricity division, an interesting open question is how to determine the optimal value of $k$ --- the $k$ that maximizes the fraction of time each agent is connected --- the fraction $\frac{k}{OPT(D_k)}$. 
Note that this ratio is not necessarily increasing with $k$. For example, consider the demand vector $D = \{11,12,13\}$:
\begin{itemize}
	\item For $k=1$, $OPT(D_k)=2$, so any agent is connected $\frac{1}{2}$ of the time. 
	\item For $k=2$, $OPT(D_k)=3$, so any agent is connected  $\frac{2}{3}$ of the time.
	\item For $k=3$, $OPT(D_k)=5$, so any agent is connected only $\frac{3}{5}<\frac{2}{3}$ of the time. 
\end{itemize}

We have also shown that, for egalitarian allocation of watts, if each bin is allocated the same amount of time, then there is no upper bound on $k$ that depend only on $n$. That motivates the development of four heuristic algorithms (HA1, HA2, HA3, HA4) for egalitarian allocation of supply. 
Among the solutions computed by each of the heuristic algorithms, a leximin preferred solution is returned.

Some other questions left open are
\begin{enumerate}
	\item 	To bridge the gap in the approximation ratio of $FFk$, between the conjectured lower bound $1.375$ and the upper bound $\left(1.5+\frac{1}{5k}\right) \cdot OPT(D_k) + 3\cdot k $, which converges to $1.5$. 
	\item To prove or disprove that the conjectured bound $\frac{11}{9}OPT + \frac{6}{9}$ is tight for $FFDk$.
\end{enumerate}

Another direction of possible research is to extend other heuristic/approximation algorithms to solve \kbp and develop a composite heuristic that randomly selects a heuristic (out of a set of heuristics) for each item and packs accordingly. It then use the $1-$Opt procedure to determine the number of bins as follows:
Determine the most empty bin. For each of the other bins, check if it can be merged with the most empty bin without violating \kbp constraints. One can repeat this $1-$Opt procedure for some arbitrary number of times. Such combination of composite heuristic and $1-$Opt has been discussed in {\cite{Hall_Ghosh_Kankey_Narasimhan_Rhee_1988}}.

Another direction of search could be to generalize the variable neighborhood search schemes (VNS) such as presented in {\cite{Fleszar_Hindi_2002}}. The VNS scheme presented in {\cite{Fleszar_Hindi_2002}} used an heuristic algorithm based on 
minimum bin slack (also called as  MBS, it determines and selects a subset of items (from the remaining set of items) that can fill the remaining space in a bin and leaves minimum remaining space).
Initial solutions generated from heuristic algorithms like $FFD$ are found to be inferior. 
A future research direction could be to extend the VNS scheme to solve \kbp. Also, it will be quite interesting to see the performance of VNS scheme when the initial solution is generated from $FFk$ and $FFDk$ as we have seen that as $k$ increases the solution quality generated from $FFk$ and $FFDk$ also increases (as shown in the experimental results related to the conjectured upper bound on $FFk$ and $FFDk$).

\begin{credits}
\subsubsection{\ackname} We thank the reviewers of SAGT 2024 for their constructive comments. This research is partly funded by the Israel Science Foundation grant no. 712/20.
\end{credits}

\bibliographystyle{splncs04}
\bibliography{Bibliography/kbp-paper}

\end{document}